\documentclass[11pt,reqno]{amsart}
\usepackage[utf8]{inputenc}
\usepackage{dsfont}
\usepackage{mathtools}
\usepackage{amsthm}
\usepackage{amsmath}
\usepackage{amssymb}
\usepackage{mathrsfs}
\usepackage{lmodern}
\usepackage{ae}
\usepackage{enumerate}
\usepackage{comment}

\numberwithin{equation}{section}

\usepackage{a4wide}
\usepackage[left=2cm,right=2cm,top=2cm,bottom=2cm]{geometry}

\newtheorem{Th}{Theorem}[section]
\newtheorem{lem}[Th]{Lemma}
\newtheorem{Prop}[Th]{Proposition}

\newtheorem{Cor}[Th]{Corollary}
\newtheorem{Rk}[Th]{Remark}

\title[Scattering Theory for weak interactions]{Scattering Theory for Mathematical Models of the Weak Interaction}

\author[B. Alvarez]{Benjamin Louis Alvarez}
\address[B. Alvarez]{Institut Elie Cartan de Lorraine \\
Universit{\'e} de Lorraine, 
57045 Metz Cedex 1, France}
\email{benjamin.alvarez@univ-lorraine.fr}

\author[J. Faupin]{J{\'e}r{\'e}my Faupin}
\address[J. Faupin]{Institut Elie Cartan de Lorraine \\
Universit{\'e} de Lorraine, 
57045 Metz Cedex 1, France}
\email{jeremy.faupin@univ-lorraine.fr}
\date{\today}

\DeclareMathOperator{\slim}{s-lim}

\begin{document}

\begin{abstract}

We consider mathematical models of the weak decay of the vector bosons $W^{\pm}$ into leptons. The free quantum field hamiltonian is perturbed by an interaction term from the standard model of particle physics. After the introduction of high energy and spatial cut-offs, the total quantum hamiltonian defines a self-adjoint operator on a tensor product of Fock spaces. We study the scattering theory for such models. First, the masses of the neutrinos are supposed to be positive: for all values of the coupling constant, we prove asymptotic completeness of the wave operators. In a second model, neutrinos are treated as massless particles and we consider a simpler interaction Hamiltonian: for small enough values of the coupling constant, we prove again asymptotic completeness, using singular Mourre's theory, suitable propagation estimates and the conservation of the difference of some number operators.

\end{abstract}

\maketitle

\section{Introduction and results}

This paper is devoted to the scattering theory of mathematical models arising from Quantum Field Theory (QFT). One of our main concerns is to establish asymptotic completeness of the wave operators for models involving massless fields. In the recent literature, this problem has been notably studied for Pauli-Fierz Hamiltonians describing confined non-relativistic particles interacting with a quantized radiation field \cite{Ge02_01,DeKu13_01,FaSi14_01,DeGrKu15_01}. Asymptotic completeness has been proven for the massless spin-boson model, but proving this property for more general Pauli-Fierz Hamiltonians remains an important open problem. In this paper, among other results, asymptotic completeness for a simplified model of QFT involving a massless field is proven, thanks to the particular structure of the model.

We consider the weak interaction between the vector bosons $W^{\pm}$ and the full family of leptons. The latter involves the electron $e^-$, the positron $e^+$, the muon $\mu^-$, the antimuon $\mu^+$, the tau $\tau^-$, the antitau $\tau^+$, the associated neutrinos $\nu_{e}$, $\nu_{\mu}$, $\nu_{\tau}$ and the antineutrinos $\bar{\nu}_{e}$, $\bar{\nu}_{\mu}$, $\bar{\nu}_{\tau}$. Typical examples of processes we are interested in are the weak decay of the $W^{\pm}$ bosons into a lepton $l^\pm$ and its associated neutrino $\nu_l$ or antineutrino $\bar{\nu}_l$,
\begin{equation}
\label{process}
W^- \, \to \,  l^{-} + \bar{\nu}_{l}, \quad W^+ \, \to \, l^+ + \nu_l.
\end{equation} 

In what follows, the mass of a particle $p$ will be denoted by $m_p$. It is equal to the mass of the corresponding antiparticle. Physically, the following inequalities hold: 
\begin{equation*}
 m_{e} < m_{\mu} < m_{\tau} < m_W.
\end{equation*}
Neutrinos were usually assumed to be massless in the classical form of the standard model of particle physics, but recent experiments have provided evidences for nonzero neutrino masses (see, e.g., \cite{Ol14_01} and references therein). Since the latter are extremely small, however, it is legitimate -- and conceptually interesting -- to consider models where neutrinos are supposed to be massless. In this paper, $m_e$, $m_\mu$, $m_\tau$ and $m_W$ will be treated as strictly positive parameters (we will not use the inequalities above), and we will consider separately two cases: i) $m_{\nu_e}>0$, $m_{\nu_\mu}>0$, $m_{\nu_\tau}>0$ and ii) $m_{\nu_e}=m_{\nu_\mu}=m_{\nu_\tau}=0$. 

The interaction term for the specific process \eqref{process} is given, in the Lagrangian formalism and for each lepton channel $l$, by (see, e.g., \cite{Gr00_01,GrMu00_01} and references therein)
\begin{equation}
\label{Interac}
I =  \int\left[ \Psi_l(x)^{\dagger}\gamma^0 \gamma^{\alpha} (1-\gamma_{5}) \Psi_{\nu_l}(x)W_{\alpha}(x) + \Psi_{\nu_l}(x)^{\dagger}\gamma^0 \gamma^{\alpha} (1-\gamma_{5})\Psi_{l}(x)W_{\alpha}(x)^* \right] d^3x,
\end{equation}
with
\begin{align}
&\Psi_l(x)=(2\pi)^{-\frac{3}{2}} \sum_{s_1=\pm\frac{1}{2}}\int\left[ \frac{u(p_1,s_1)e^{\mathrm{i}p_1 \cdot x}}{(2(|p_1|^2+m_l^2)^{\frac{1}{2}})^{\frac{1}{2}}} b_{l,+}(p_1,s_1) + \frac{v(p_1,s_1)e^{-\mathrm{i}p_1 \cdot x}}{(2(|p_1|^2+m_l^2)^{\frac{1}{2}})^{\frac{1}{2}}} b^*_{l,-}(p_1,s_1) \right]d^3p_1, \label{lept} \\
& \Psi_{\nu_l}(x)=(2\pi)^{-\frac{3}{2}} \sum_{s_2=\pm\frac{1}{2}}\int\left[ \frac{u(p_2,s_2)e^{\mathrm{i}p_2\cdot x}}{(2(|p_2|^2+m_{\nu_l}^2)^{\frac{1}{2}})^{\frac{1}{2}}} c_{l,+}(p_2,s_2) + \frac{v(p_2,s_2)e^{-\mathrm{i}p_2 \cdot x}}{(2(|p_2|^2+m_{\nu_l}^2)^{\frac{1}{2}})^{\frac{1}{2}}} c^*_{l,-}(p_2,s_2) \right]d^3p_2, \label{neu} \\
& W_{\alpha}(x)=(2\pi)^{-\frac{3}{2}} \sum_{\lambda=-1, 0, 1} \int\left[ \frac{\epsilon_{\alpha}(p_3,\lambda)e^{\mathrm{i}p_3 \cdot x} }{(2(|p_3|^2+m_{W}^2)^{\frac{1}{2}})^{\frac{1}{2}}} a_+(p_3,\lambda) + \frac{ \epsilon^*_{\alpha}(p_3,\lambda) e^{-\mathrm{i}p_3 \cdot x} }{(2(|p_3|^2+m_{W}^2)^{\frac{1}{2}})^{\frac{1}{2}}} a^*_-(p_3,\lambda)  \right] d^3p_3. \label{Boson}
\end{align}
Here, $u$ and $v$ are the solutions to the Dirac equation (normalized as in \cite[(2.13)]{GrMu00_01}), $\epsilon_{\alpha}$ is a polarisation vector,  $\gamma^{\alpha}$, $\alpha = 0 , \dots , 3$ and $\gamma_{5}$ are the usual gamma matrices. Moreover, the index $l \in \{1,2,3\}$ labels the lepton families, $p_1,p_2,p_3 \in \mathbb{R}^3$ stand for the momentum variables of fermions and bosons, $s_i \in \{-\frac{1}{2},\frac{1}{2}\}$ denotes the spin of fermions and $\lambda \in \{-1,0,1\}$ the spin of bosons. The operators $b_{l,+}(p_1,s_1)$ and $b^*_{l,+}(p_1,s_1)$ are annihilation and creation operators for the electron if $l=1$, muon if $l=2$ and tau if $l=3$. The operators $b_{l,-}(p_1,s_1)$ and $b^*_{l,-}(p_1,s_1)$ are annihilation and creation operators for the associated antiparticles. Likewise, $c_{l,+}(p_2,s_2)$ and $c^*_{l,+}(p_2,s_2)$ (respectively $c_{l,-}(p_2,s_2)$ and $c^*_{l,-}(p_2,s_2)$) stand for annihilation and creation operators for the neutrinos of the $l$-family (respectively antineutrinos) and the operators $a_+(p_3,\lambda)$ and $a^*_+(p_3,\lambda)$ (respectively $a_-(p_3,\lambda)$ and $a^*_-(p_3,\lambda)$) are annihilation and creation operators for the boson $W^-$ (respectively $W^+$).

It should be mentioned that, when neutrinos are supposed to be massive, a slightly different interaction term can be found in the literature (see, e.g., \cite{Xi14_01}). More precisely, massive neutrinos fields $(\tilde{\Psi}_{\nu_1},\tilde{\Psi}_{\nu_2},\tilde{\Psi}_{\nu_3})$ may be defined by applying a $3\times3$ unitary matrix transformation to the fields $(\Psi_{\nu_1},\Psi_{\nu_2},\Psi_{\nu_3})$ in \eqref{neu}. Our results can be proven without any noticeable change if one considers such interaction terms. We will not do so in the present paper.

For shortness, we denote by $\xi_i = (p_i,s_i)$, $i=1,2$, the quantum variables for fermions, and $\xi_3 = (p_3,\lambda)$ for bosons. The following canonical commutation and anticommutation relations hold: 
\begin{align*}
\{b_{l,\epsilon}(\xi_1),b^*_{l',\epsilon'}(\xi_2)\} & = \{c_{l,\epsilon}(\xi_1),c^*_{l',\epsilon'}(\xi_2)\} = \delta_{ll'}\delta_{\epsilon \epsilon'} \delta(\xi_1-\xi_2) ,\\
~[a_{\epsilon}(\xi_1),a^*_{\epsilon'}(\xi_2)] & = \delta_{\epsilon \epsilon'} \delta(\xi_1-\xi_2) , \\
\{b_{l,\epsilon}(\xi_1),b_{l',\epsilon'}(\xi_2)\} & = \{c_{l,\epsilon}(\xi_1),c_{l',\epsilon'}(\xi_2)\} = 0 ,  \\
~[a_{\epsilon}(\xi_1),a_{\epsilon'}(\xi_2)] & = 0 , \\
\{b_{l,\epsilon}(\xi_1),c_{l',\epsilon'}(\xi_2)\} & = \{b_{l,\epsilon}(\xi_1),c^*_{l',\epsilon'}(\xi_2)\} = 0 , \\
\{b^*_{l,\epsilon}(\xi_1),c_{l',\epsilon'}(\xi_2)\} & = \{b^*_{l,\epsilon}(\xi_1),c^*_{l',\epsilon'}(\xi_2)\} = 0 , \\
~[a_{\epsilon}(\xi_1),c_{\epsilon'}(\xi_2)] & = [a_{\epsilon}(\xi_1),b_{\epsilon'}(\xi_2)] = 0 , \\
~[a_{\epsilon}(\xi_1),c^*_{\epsilon'}(\xi_2)] & = [a_{\epsilon}(\xi_1),b^*_{\epsilon'}(\xi_2)] = 0 ,\\
~[a^*_{\epsilon}(\xi_1),c_{\epsilon'}(\xi_2)] & = [a^*_{\epsilon}(\xi_1),b_{\epsilon'}(\xi_2)] = 0 , \\
~[a^*_{\epsilon}(\xi_1),c^*_{\epsilon'}(\xi_2)] & = [a^*_{\epsilon}(\xi_1),b^*_{\epsilon'}(\xi_2)] = 0 ,
\end{align*}
with $l,l'\in\{1,2,3\}$, $\epsilon,\epsilon'=\pm$.

 Inserting \eqref{lept}--\eqref{Boson} into \eqref{Interac}, integrating with respect to $x$, and using the convention
\begin{equation*}
\int d\xi_1d\xi_2d\xi_3 = \sum_{s_1=\pm \frac{1}{2}} \sum_{s_2=\pm \frac{1}{2}} \sum_{\lambda=-1, 0, 1}\int d^3p_1d^3p_2d^3p_3,
\end{equation*}
we arrive at the formal expression
\begin{align}
H_I := \sum_{j=1}^4 H^{(j)}_I := \sum_{l=1}^3 \sum_{\epsilon=\pm} \int & \left\{ \left[ G^{(1)}_{l,\epsilon} (\xi_1,\xi_2,\xi_3)b^*_{l,\epsilon}(\xi_1)c^*_{l,-\epsilon}(\xi_2)a_{\epsilon}(\xi_3) + \mathrm{h.c.} \right]\right. \notag \\
& + \left[ G^{(2)}_{l,\epsilon} (\xi_1,\xi_2,\xi_3)b^*_{l,-\epsilon}(\xi_1)c^*_{l,\epsilon}(\xi_2)a^*_{\epsilon}(\xi_3) + \mathrm{h.c.} \right] \notag \\
&+\left[ G^{(3)}_{l,\epsilon} (\xi_1,\xi_2,\xi_3)b^*_{l,-\epsilon}(\xi_1)c_{l,-\epsilon}(\xi_2)a^*_{\epsilon}(\xi_3) + \mathrm{h.c.} \right] \notag \\
&+ \left.  \left[ G^{(4)}_{l,\epsilon} (\xi_1,\xi_2,\xi_3)b^*_{l,\epsilon}(\xi_1)c_{l,\epsilon}(\xi_2) a_{\epsilon}(\xi_3) + \mathrm{h.c.} \right] \right\}d\xi_1d\xi_2d\xi_3 ,  \label{interaction_term}
\end{align}
where we set $-\epsilon=\mp$ if $\epsilon=\pm$. The kernels $G_{l,\epsilon}^{(j)}$, $j = 1 , \dots , 4$, are of the form
\begin{align}
&G_{l,\epsilon}^{(1)}( \xi_1 , \xi_2 , \xi_3 ) = f^{(1)}_{l,\epsilon,1}( \xi_1 ) f^{(1)}_{l,\epsilon,2}( \xi_2 ) f^{(1)}_{l,\epsilon,3}( \xi_3 ) \delta ( - p_1 - p_2 + p_3 ) , \label{eq:kernel1} \\
&G_{l,\epsilon}^{(2)}( \xi_1 , \xi_2 , \xi_3 ) = f^{(2)}_{l,\epsilon,1}( \xi_1 ) f^{(2)}_{l,\epsilon,2}( \xi_2 ) f^{(2)}_{l,\epsilon,3}( \xi_3 ) \delta ( p_1 + p_2 + p_3 ) , \\
&G_{l,\epsilon}^{(3)}( \xi_1 , \xi_2 , \xi_3 ) = f^{(3)}_{l,\epsilon,1}( \xi_1 ) f^{(3)}_{l,\epsilon,2}( \xi_2 ) f^{(3)}_{l,\epsilon,3}( \xi_3 ) \delta ( - p_1 + p_2 - p_3 ) , \\
&G_{l,\epsilon}^{(4)}( \xi_1 , \xi_2 , \xi_3 ) = f^{(4)}_{l,\epsilon,1}( \xi_1 ) f^{(4)}_{l,\epsilon,2}( \xi_2 ) f^{(4)}_{l,\epsilon,3}( \xi_3 ) \delta ( - p_1 + p_2 + p_3 ) , \label{eq:kernel4}
\end{align}
where the maps $p_i \mapsto f^{(j)}_{l,\epsilon,i}( \xi_i )$ are bounded in any compact set of $\mathbb{R}^3$.
Their explicit expressions are given in Appendix \ref{app:interaction}.

An important property of the interaction Hamiltonian \eqref{interaction_term} is that it preserves the lepton number in the sense that, formally, $N_{l^-} + N_{\nu_l} - N_{l^+} - N_{\bar{\nu}_l}$ commutes with $H_I$. Here, $N_p$ stands for the number operator corresponding to a particle $p$. 

We observe that the first term in \eqref{interaction_term}, $H_I^{(1)}$, describes explicitly processes like \eqref{process}, while $H_I^{(2)}$ prevents the bare vacuum from being a bound state, as expected from physics. To study a process like \eqref{process}, it is reasonable, in a first approximation, to keep only these first two terms, thus considering the simpler interaction Hamiltonian $H_I = H_I^{(1)} + H_I^{(2)}$. Under the assumption that the masses of the neutrinos vanish, we will make this approximation. The advantage is that the differences of number operators $N_{l^-} - N_{\bar{\nu}_l}$ and $N_{l^+} - N_{\nu_l}$ are preserved by the Hamiltonian. This property will be essential in some of our arguments. On the other hand, if we assume that the masses of the neutrinos do not vanish, our argument applies without requiring that such a quantity be conserved, and therefore the full interaction Hamiltonian \eqref{interaction_term} can be studied.

Now, the free Hamiltonian is given by
\begin{align}
H_0  = & \sum_{l=1}^{3} \sum_{\epsilon=\pm} \int \omega_l^{(1)}(\xi_1) {b^*_{l,\epsilon}}(\xi_1){b_{l,\epsilon}}(\xi_1) d\xi_1  +  \sum_{l=1}^{3} \sum_{\epsilon=\pm} \int \omega_l^{(2)}(\xi_2) {c^*_{l,\epsilon}}(\xi_2){c_{l,\epsilon}}(\xi_2) d\xi_2 \notag \\
  & +  \sum_{\epsilon=\pm} \int \omega^{(3)}(\xi_3) {a^*_{\epsilon}}(\xi_3){a_{\epsilon}}(\xi_3) d\xi_3, \label{Hlibre}
\end{align}
with the dispersion relations
\begin{equation}
\label{disprela}
\omega^{(1)}_l(\xi_1) = \sqrt{p^2_1+m_l^2}, \quad \omega^{(2)}_l(\xi_2) = \sqrt{p^2_2+m_{\nu_l}^2}, \quad  \omega^{(3)}(\xi_3)  =  \sqrt{p^2_3+m_{W^{\pm}}^2} .
\end{equation}
The total Hamiltonian is defined by
\begin{equation}
\label{total_hamiltonian}
H = H_0 + H_I .
\end{equation}

Since the kernels $G^{(j)}_{l,\epsilon}$ are singular, the formal expressions \eqref{interaction_term}--\eqref{total_hamiltonian} do not define a self-adjoint operator in Fock space (see the next section for the precise definition of the Hilbert space that we consider). In order to obtain such a self-adjoint operator, following a standard procedure in constructive QFT (see e.g. \cite{GlJa77_01} and references therein), we introduce ultraviolet and spatial cut-offs in the interaction Hamiltonian. Of course, eventually, it would be desirable to find a renormalization procedure allowing one to remove those cut-offs. This constitutes an important open problem which is beyond the scope of this paper. 

Let $\Lambda > 0$ be a fixed ultraviolet parameter and let $B( 0 , \Lambda )$ denotes the ball centered at $0$ and of radius $\Lambda$ in $\mathbb{R}^3$. In the formal expression \eqref{interaction_term}, we introduce ultraviolet cut-offs, i.e., we replace $f^{(j)}_{l,\epsilon,i}( \xi_i )$ by $\chi_{ B( 0 , \Lambda ) }( p_i )f^{(j)}_{l,\epsilon,i}( \xi_i )$, for some smooth function $\chi_{ B( 0 , \Lambda ) }$ supported in $B( 0 , \Lambda )$, and we replace the Dirac delta function $\delta(p)$ by an approximation, $\delta_n(p) = n^3 \delta_1( np )$, for some smooth and compactly supported function $\delta_1$. The resulting kernels are still denoted by the same symbols $G^{(j)}_{l,\epsilon}$. In particular, $G^{(j)}_{l,\epsilon}$ are now square integrable. As will be shown in the next section, square integrability of the kernels is actually sufficient to prove that $H=H_0+H_I$ defines a self-adjoint operator in Fock space.

Some of our results will be proven in the weak coupling regime. We will therefore study in this paper an abstract class of Hamiltonians given by
\begin{equation}
H = H_0 + g H_I , \label{total_hamiltonian_g}
\end{equation}
where $g$ is a real coupling parameter, $H_0$ is defined in \eqref{Hlibre}, and $H_I$ is given by \eqref{interaction_term} with abstract kernels $G^{(j)}_{l,\epsilon}$. The latter will always be supposed to be square integrable and, in some cases, stronger regularity assumptions on $G^{(j)}_{l,\epsilon}$ will be required. This will be made more precise below.

The spectral theory of such models of the weak interaction has been studied in particular in \cite{AmGrGu07_01,AsBaFaGu11_01,BaFaGu16_01,BaFaGu16_02,BaFaGuappear,BaGu09_01} (see also \cite{BaDiGu04_01,Ta14_01} for related models of QED). Without entering into details, the results established in these references show that, for weak coupling, and under suitable assumptions on the kernels, $H$ is self-adjoint and has a ground state (i.e. $E:=\inf \sigma( H )$ is an eigenvalue of $H$), and the essential spectrum of $H$ coincides with the semi axis $\sigma_{\mathrm{ess}}(H) = [ E + m_{\nu} , \infty )$, with $m_\nu=\min(m_{\nu_e},m_{\nu_\mu},m_{\nu_\tau})$. In particular, if the masses of the neutrinos vanish, the ground state energy is an eigenvalue of $H$ embedded into its essential spectrum. Moreover, except for the ground state energy, the spectrum of $H$ below the electron mass is purely absolutely continuous.

In this paper, we complement the previous spectral results by studying the structure of the essential spectrum in the whole semi-axis $[ E + m_{\nu} , \infty )$ (not only below the electron mass) and by relaxing the weak coupling assumption in the case where the masses of the neutrinos do not vanish.

Our main purpose is then to study scattering theory for models of the form \eqref{interaction_term}--\eqref{total_hamiltonian_g} and, in particular, to prove asymptotic completeness. Scattering theory for models of non-relativistic matter coupled to a massive, bosonic quantum field -- massive Pauli-Fierz Hamiltonians -- has been considered by many authors. See, among others, \cite{Ar83_01,De03_01,DeGe99_01,DeGe04_01,Fr74_01,FrGrSc01_01,FrGrSc02_01,FrGrSc04_01,FrGrSc07_01,HuSp95_01,Ho68_01,Ho69_01,Sp97_01}; see also \cite{Am04_01} for fermionic Pauli-Fierz systems, \cite{DeGe00_01} for spatially cut-off $P(\varphi)_2$ Hamiltonians, \cite{GePa09_01} for abstract QFT Hamiltonians, and \cite{BoFaSi12_01,DeKu13_01,DeKu15_01,DeGrKu15_01,FaSi14_01,FaSi14_02,Ge02_01} for massless Pauli-Fierz Hamiltonians. A large part of the techniques used in the present paper are adapted from the ones developed in these references.

The first step of the approach to scattering theory that we follow consists in establishing the existence and basic properties of the asymptotic creation and annihilation operators
\begin{equation}
a^{\pm,\sharp}_{\epsilon} (h)  :=  \underset{t \to +\infty}{\slim} e^{ \pm \mathrm{i} t H} e^{\mp \mathrm{i} tH_0} a^\sharp_{\epsilon}(h) e^{\pm \mathrm{i} tH_0} e^{\mp \mathrm{i} t H } , \label{eq:def_asympt_a}
\end{equation}
for any $h$ in $L^2( \mathbb{R}^3 \times \{ -1,0,1 \} )$, where $a^{\sharp}$ stands for $a$ or $a^*$ and
\begin{equation}
a_\epsilon^*(h) = \int h(\xi_3) a_\epsilon^*(\xi_3) d\xi_3, \quad a_\epsilon(h) = \int \bar{h}(\xi_3) a_\epsilon(\xi_3) d\xi_3. \label{eq:def_a}
\end{equation}
The fermionic asymptotic creation and annihilation operators $b^{\pm,\sharp}_{l,\epsilon}(h)$ and $c^{\pm,\sharp}_{l,\epsilon}(h)$ are defined similarly, for $h \in L^2( \mathbb{R}^3 \times \{-\frac12,\frac12\})$.

Let $\mathscr{H}$ be the Hilbert space of the model, defined as a tensor product of Fock spaces, see the next section for precise definitions. The space of asymptotic vacua is defined by
\begin{equation*}
\mathscr{K}^\pm := \{ u , \, d^\pm( h ) u = 0 \text{ for all asymptotic annihilation operator } d^\pm( h ) \}, 
\end{equation*}
where $d^\pm(h)$ stands for either $a_\epsilon^\pm(h)$, with $h \in L^2( \mathbb{R}^3 \times \{ -1,0,1 \} )$, or $b^{\pm}_{l,\epsilon}(h)$ or $c^{\pm}_{l,\epsilon}(h)$, with $h \in L^2( \mathbb{R}^3 \times \{-\frac12,\frac12\})$. There is a natural definition of isometric wave operators
\begin{equation}
\Omega^\pm : \mathscr{K}^\pm \otimes \mathscr{H} \to \mathscr{H}, \label{eq:def_Omega_1}
\end{equation}
with the property that
\begin{equation}
\Omega^\pm ( \mathds{1} \otimes d^{\sharp}(h) ) = d^{\pm,\sharp}(h) \Omega^\pm, \label{eq:def_Omega_2}
\end{equation}
where, again, $d^{\sharp}(h)$ stands for any kind of creation or annihilation operator. Asymptotic completeness of $\Omega^\pm$ is the statement that $\Omega^\pm$ are unitary and that $\mathscr{K}^\pm = \mathscr{H}_{\mathrm{pp}}( H )$, where $\mathscr{H}_{\mathrm{pp}}( H )$ denotes the pure point spectral subspace of $H$. An interpretation of asymptotic completeness is that any evolving state $e^{-\mathrm{i}tH}u$, with $u \in \mathscr{H}$, can be decomposed, asymptotically as time $t$ goes to infinity, into a bound state together with asymptotically free particles.

Recall that the parameter $j\in\{1,\dots,4\}$ labels the different interaction terms in \eqref{interaction_term} and that the index $l \in \{1,2,3\}$ labels the lepton families. Moreover $\epsilon=\pm$. In what follows, for shortness, we say that
\begin{equation*}
\text{``$G \in L^2$'' if, for all $j$, $l$ and $\epsilon$, $G_{l,\epsilon}^{(j)}$ is square integrable}.
\end{equation*}
Recall that $s_1,s_2$ denote the spin variables for fermions and that $\lambda$ denotes the spin variable for bosons. We say that 
\begin{equation*}
\text{``$G \in \mathbb{H}^\mu$'' if, for all $j$, $l$, $\epsilon$, $s_1$, $s_2$ and $\lambda$, $G_{l,\epsilon}^{(j)}( s_1 , \cdot , s_2 , \cdot , \lambda , \cdot )$ belongs to the Sobolev space $\mathbb{H}^{\mu}( \mathbb{R}^9 )$}.
\end{equation*}

Remembering that the dispersion relation $\omega_l^{(i)}$, $i=1,2$, $l=1,2,3$ and $\omega^{(3)}$ are defined in \eqref{disprela}, we set
\begin{align}
& a_{(i),l} := \frac{\mathrm{i}}{2} \big ( \nabla_{p_i} \cdot \nabla \omega^{(i)}_l(p_i) + \nabla \omega^{(i)}_l( p_i ) \cdot \nabla_{p_i} \big ), \quad i=1,2 , \quad l = 1 ,2 , 3 , \label{eq:defa12} \\
& a_{(3)} := \frac{\mathrm{i}}{2} \big ( \nabla_{p_3} \cdot \nabla \omega^{(3)}(p_3) + \nabla \omega^{(3)}( p_3 ) \cdot \nabla_{p_3} \big ), \label{eq:defa3} \\
& b_{(i),l} :=  \frac{\mathrm{i}}{2} \big ( ( p_i \cdot \nabla \omega^{(i)}_l (p_i ) )^{-1} p_i \cdot \nabla_{p_i} + \nabla_{p_i} \cdot p_i  ( p_i \cdot \nabla \omega^{(i)}_l (p_i ) )^{-1} \big ) , \quad i=1,2 , \quad l = 1 ,2 , 3 , \label{eq:defb12} \\
& b_{(3)} :=  \frac{\mathrm{i}}{2} \big ( ( p_3 \cdot \nabla \omega^{(3)} (p_3 ) )^{-1} p_3 \cdot \nabla_{p_3} + \nabla_{p_3} \cdot p_3  ( p_3 \cdot \nabla \omega^{(3)} (p_3 ) )^{-1} \big ), \label{eq:defb3}
\end{align}
as partial differential operators acting on $L^2( d\xi_1 d\xi_2 d\xi_3 )$. See Section \ref{sec:spectral} for details concerning the domains and properties of these operators. It should be noted that, in the case where the masses of the neutrinos vanish, we have that $a_{(2),l} = b_{(2),l}$.

To shorten the statement of some of our results below, we introduce the notation ``$a_{(i),\cdot} G \in L^2$'' with the following meaning: for $i=1,2$, we say that
\begin{equation*}
\text{``$a_{(i),\cdot} G \in L^2$'' if, for all $j$, $l$ and $\epsilon$, $a_{(i),l} G_{l,\epsilon}^{(j)}$ is square integrable},
\end{equation*}
and, for $i=3$, we say that
\begin{equation*}
\text{``$a_{(3),\cdot} G \in L^2$'' if, for all $j$, $l$ and $\epsilon$, $a_{(3)} G_{l,\epsilon}^{(j)}$ is square integrable}.
\end{equation*}
The notation $a_{(i),\cdot} a_{(i'),\cdot} G \in L^2$ is defined analogously, and likewise for $b_{(i),\cdot} G \in L^2$ and $b_{(i),\cdot} b_{(i'),\cdot} G \in L^2$. Given these conventions, the notation $|p_3|^{-1} a_{(i),\cdot} G \in L^2$ have an obvious meaning.

Our main results can be stated as follows.
\begin{Th}\label{thm:main}
$\quad$

\begin{itemize}
\item[(i)] Suppose that the masses of the neutrinos $m_{\nu_e}$, $m_{\nu_{\mu}}$, $m_{\nu_{\tau}}$ are positive and consider the Hamiltonian \eqref{total_hamiltonian_g} with $H_I$ given by \eqref{interaction_term}. Assume that
\begin{equation*}
G \in L^2 , \quad a_{(i),\cdot} G \in L^2, \quad i = 1,2,3, 
\end{equation*}
and that $G \in \mathbb{H}^{1+\mu}$ for some $\mu > 0$. Then the wave operators $\Omega^\pm$ exist and are asymptotically complete. Suppose in addition that
\begin{equation*}
 b_{(i),\cdot} G \in L^2, \quad i = 1,2,3, \quad b_{(i),\cdot} b_{(i'),\cdot} G \in L^2 , \quad i,i'=1,2,3.
\end{equation*}
Then there exists $g_0 > 0$, which does not depend on $m_{\nu_e}$, $m_{\nu_{\mu}}$, $m_{\nu_{\tau}}$, such that, for all $|g| \le g_0$, $H-E$ is unitarily equivalent to $H_0$.
\item[(ii)]Suppose that the masses of the neutrinos $m_{\nu_e}$, $m_{\nu_{\mu}}$, $m_{\nu_{\tau}}$ vanish and consider the Hamiltonian $H = H_0 + g ( H_I^{(1)} + H_I^{(2)} )$ with $H_I^{(1)}$ and $H_I^{(2)}$ given by \eqref{interaction_term}. Assume that
\begin{equation*}
G \in L^2, \quad a_{(i),\cdot} G \in L^2, \quad |p_3|^{-1} a_{(i),\cdot} G \in L^2, \quad i = 1,2,3 ,
\end{equation*}
and that $G \in \mathbb{H}^{1+\mu}$ for some $\mu > 0$. Then there exists $g_0>0$ such that, for all $|g| \le g_0$, the wave operators $\Omega^\pm$ exist and are asymptotically complete.
Suppose in addition that
\begin{equation*}
 b_{(i),\cdot} G \in L^2, \quad i = 1,2,3, \quad b_{(i),\cdot} b_{(i'),\cdot} G \in L^2 , \quad i,i'=1,2,3.
\end{equation*}
Then there exists $g'_0>0$ such that, for all $|g|\le g'_0$, $H-E$ is unitarily equivalent to $H_0$.
\end{itemize}
\end{Th}
\begin{Rk}
\begin{itemize}
\item[(i)]
As mentioned above, physically, the masses of the neutrinos are extremely small. We emphasize that the second part of Theorem \ref{thm:main} (i) holds for $|g|$ small enough, uniformly with respect to the masses of the neutrinos. As a consequence, from the observation that the Hamiltonian with massive neutrinos converges to that with massless neutrinos, in the norm resolvent sense, as the masses of the neutrinos go to $0$, one can deduce that, in the massless case, $H-E$ is approximately unitarily equivalent to $H_0$. This means that there exists a sequence of unitary operators $(U_n)$ such that $U_n H_0 U_n^* \to H-E$, as $n\to\infty$, in the norm resolvent sense. The conclusion of Theorem \ref{thm:main} (ii), which concerns also the massless case, is, of course, significantly stronger since it shows that $H-E$ and $H_0$ are unitarily equivalent (not only approximately unitarily equivalent). However, Theorem \ref{thm:main} (ii) only holds for the simplified Hamiltonian $H = H_0 + g ( H_I^{(1)} + H_I^{(2)} )$.
\item[(ii)] Suppose that the kernels $G^{(j)}_{l,\epsilon}$ are of the forms \eqref{eq:kernel1}--\eqref{eq:kernel4}, where, as explained above, the $\delta$-functions are regularized and ultraviolet cut-offs are introduced. Suppose in addition that the functions $f^{(j)}_{l,\epsilon,i}$ belong to $\mathrm{C}^\infty( \mathbb{R}^3 \setminus \{ 0 \} )$ and satisfy
\begin{equation*}
\big | \partial^\alpha_{p_{i,\ell} } f^{(j)}_{l,\epsilon,i} ( \xi_1 , \xi_2 , \xi_3 ) \big| \le \mathrm{C}_\alpha |p_i|^{\nu-\alpha} , \quad \alpha \in \mathbb{N} , \quad p_i = ( p_{i,1} , p_{i,2} , p_{i,3} ),
\end{equation*}
for some real parameter $\nu$. Then, in the massive case (Theorem \ref{thm:main} (i)), the conditions $G \in L^2$ and $a_{(i),\cdot} G \in L^2$ are satisfied for any $\nu > - 3/2$, and $G \in \mathbb{H}^{1+\mu}$ is satisfied provided that $\nu > -1/2 + \mu$. The conditions $ b_{(i),\cdot} G \in L^2$ are satisfied for $\nu > 1/2$, and $b_{(i),\cdot} b_{(i'),\cdot} G \in L^2$ for $\nu > 5/2$. The same holds in the massless case (Theorem \ref{thm:main} (ii)), except that $a_{(2),\cdot} G \in L^2$ is satisfied for $\nu > -1/2$ (recall that $a_{(2),l} G = b_{(2),l} G$ in the massless case), and, if $i=2$ or $i'=2$, $b_{(i),\cdot} b_{(i'),\cdot} G \in L^2$ is satisfied for $\nu > 3/2$. Besides, $|p_3|^{-1} a_{(i),\cdot} G \in L^2$ holds as soon as $\nu > -1/2$.
\item[(iii)] The assumption that $|p_3|^{-1} a_{(i),\cdot} G \in L^2$ in Theorem \ref{thm:main} (ii) can be replaced by $|p_1|^{-1} a_{(i),\cdot} G \in L^2$.
\end{itemize}
\end{Rk}
To prove Theorem \ref{thm:main}, we follow the general approach of \cite{DeGe99_01,DeGe00_01,Am04_01} that we adapt to the present context. In addition to the fact that the model we study involves both bosons and fermions, one of our main achievements, compared to \cite{DeGe99_01,DeGe00_01,Am04_01}, is that the second part of our results in (i) hold with a restriction on the coupling constant which is uniform w.r.t. the masses of the neutrinos, and our results in (ii) are proven for an Hamiltonian involving massless particles. As mentioned above, for Pauli-Fierz Hamiltonians, contributions to scattering theory involving massless particles include \cite{Ge02_01,DeKu13_01,FaSi14_01,FaSi14_02,DeKu15_01,DeGrKu15_01}. In particular, for the massless spin-boson model, asymptotic completeness has been established in \cite{FaSi14_01,DeGrKu15_01}, using the uniform bound on the number of emitted particles proven in \cite{DeKu13_01}. This property, however, has not been proven for more general massless Pauli-Fierz Hamiltonians, yet. In our setting, controlling the number of emitted particles is made possible thanks to the fact that $N_{\mathrm{lept}} - N_{\mathrm{neut}}$ commutes with $H$, with $N_{\mathrm{lept}}$ the number of leptons and antileptons and $N_{\mathrm{neut}}$ the number of neutrinos and antineutrinos.

In our proof of Theorem \ref{thm:main}, as in previous works, one of the main issues consists in finding a good choice of a ``conjugate operator'' $A$, such that the commutator $[H,\mathrm{i}A]$ is positive in the sense of Mourre \cite{Mo80_01}, or in a related weaker sense \cite{GeGeMo04_01,GeGeMo04_02}. From such a positive commutator estimate, one deduces spectral properties such that the absence of singular continuous spectrum for $H$, and suitable propagation estimates allowing one to establish asymptotic completeness.

In the case where all particles are massive ((i) in Theorem \ref{thm:main}), our conjugate operator is the sum of the second quantizations of the operators \eqref{eq:defa12}--\eqref{eq:defa3}. This is a natural generalization of the conjugate operator chosen for instance in \cite{DeGe99_01,DeGe00_01,Am04_01}. Our main achievement here is that we show that the operators $H-E$ and $H_0$ are unitary equivalent for small $|g|$, uniformly in the masses of the neutrinos. To prove this, we use, in an essential way, the fact that neutrinos are fermions, together with a suitable application of the $N_\tau$ estimates of Glimm and Jaffe \cite{GlJa77_01} and an extension of Mourre's theory allowing for non-self-adjoint conjugate operators \cite{GeGeMo04_01}.

The proof of Theorem \ref{thm:main} (ii) constitutes the main novelty compared to previous results in the literature. Here we cannot follow directly the approach of \cite{DeGe99_01,DeGe00_01,Am04_01} because of the presence of massless particles. To obtain a useful Mourre estimate, we combine singular Mourre's Theory  \cite{GeGeMo04_01} together with an induction argument of \cite{DeGe99_01} and smallness of the coupling constant $g$. Using this Mourre estimate and the fact that $N_{\mathrm{lept}} - N_{\mathrm{neut}}$ commutes with $H$, we then establish propagation estimates. Our propagation estimates resemble those proven in \cite{DeGe99_01}, but with a different time-dependent propagation observable. Our choice of the (one-particle) propagation observable is inspired in part by that used in \cite{FaSi14_01,FaSi14_02}: it is a time-dependent modification of the usual ``position operator'', especially fitted to handle singularities due to the presence of massless particles. We do not take the same propagation observable as in \cite{FaSi14_01,FaSi14_02} because we follow a different approach, closer to that of \cite{DeGe99_01}, to prove asymptotic completeness.

As in \cite{BoFaSi12_01,FaSi14_01,FaSi14_02}, an important ingredient to prove the propagation estimates is the control of the observable $\mathrm{d}\Gamma( |p_2|^{-1} )$ along the time evolution (where $p_2$ is the momentum of a neutrino). More precisely, we prove that, for suitable initial states, the expectation of this observable along the evolution grows slower than linearly in time $t$, which is crucial to estimate some remainder terms in the propagation estimates. 

Our paper is organized as follows. In Section \ref{sec:selfadjoint}, we show that the Hamiltonian \eqref{total_hamiltonian_g} defines a self-adjoint operator on a Hilbert space given as a tensor product of antisymmetric and symmetric Fock spaces. The result holds without any restriction on the coupling constant $g$. Section \ref{sec:selfadjoint} also contains the proof of some technical estimates that are used in the subsequent sections.

In Section \ref{sec:spectral}, we recall results giving the existence of a ground state and the location of the essential spectrum of $H$, and we study the essential spectrum by means of suitable versions of Mourre's conjugate operator method.

Section \ref{sec:propag} is devoted to the proof of several propagation estimates.

Finally, in Section \ref{sec:AC}, we prove some properties of the asymptotic fields and wave operators, and we use the results of Sections \ref{sec:spectral} and \ref{sec:propag} to prove Theorem \ref{thm:main}.

Technically, our main contributions, compared to previous works, are the proof of the Mourre estimate in Theorem \ref{MourreI_3}, and the proof of propagation estimates in Theorems \ref{th:propag_massless} and \ref{th:propag_massless_2}.

For the convenience of the reader, the complete expression of the formal interaction Hamiltonian \eqref{interaction_term} is given in Appendix \ref{app:interaction}, and the definitions and properties of some operators in Fock space are recalled in Appendices \ref{sec:standard_def} and \ref{sec:asympt_a}. Technical computations are gathered in Appendix \ref{app:technical}.

Throughout the paper, the notation $a \lesssim b$, for positive numbers $a$ and $b$, stands for $a \le \mathrm{C} b$ where $\mathrm{C}$ is a positive constant independent of the parameters involved. \\

\noindent \textbf{Acknowledgements}. We thank J.-C. Guillot for interesting remarks. J.F. is grateful to J.-M. Barbaroux and J.-C. Guillot for many discussions and fruitful collaborations.

\section{Self-adjointness of the Hamiltonian and technical estimates}\label{sec:selfadjoint}

In this section we show that the Hamiltonian of the model, formally defined in \eqref{total_hamiltonian_g}, identifies with a self-adjoint operator in an appropriate Hilbert space. The definition of the Hilbert space is given in Section \ref{2intro} and self-adjointness of the Hamiltonian is proven in Section \ref{sec:self-adj}. In Section \ref{subsec:estimates}, we prove some technical estimates that will be used in the next sections.

\subsection{Hilbert space}\label{2intro}

The Hilbert space of the model is a tensor product of Fock spaces for fermions and bosons. For fermions, we define $\Sigma_1:=\mathbb{R}^3 \times \{-\frac{1}{2},\frac{1}{2}\}$ and, for bosons,  $\Sigma_2:=\mathbb{R}^3 \times \{-1,0,1\}$. The one-particle Hilbert space for fermions is $\mathfrak{h}_1:=L^2(\Sigma_1)$ and for bosons $\mathfrak{h}_2:=L^2(\Sigma_2)$. The anti-symmetric Fock space for fermions is denoted by $\mathfrak{F}_a:=\oplus^{\infty}_{n=0}\otimes^n_{a}\mathfrak{h}_1$, where $\otimes_a$ stands for the antisymmetric tensor product and where we use the usual convention $\otimes^0_{a}\mathfrak{h}_1=\mathbb{C}$. The symmetric Fock space for bosons is $\mathfrak{F}_s:=\oplus^{\infty}_{n=0}\otimes^n_{s}\mathfrak{h}_2$, where $\otimes_s$ stands for the symmetric tensor product and $\otimes^0_{s}\mathfrak{h}_2=\mathbb{C}$. Every family $l$ of leptons contains either an electron, a muon or a tau, the associated antiparticle, and a neutrino and its antineutrino. Consequently, the Hilbert space for each lepton family is
\begin{equation*}
\mathfrak{F}_l:=\bigotimes^4 \mathfrak{F}_a ,
\end{equation*}
and we denote the full leptonic Hilbert space by 
\begin{equation*}
\mathfrak{F}_L:=\bigotimes^3 \mathfrak{F}_l.
\end{equation*}
Analogously, the bosonic Hilbert space is given by
\begin{equation*}
\mathfrak{F}_W:=\bigotimes^2 \mathfrak{F}_s.
\end{equation*}
The total Hilbert space is
\begin{equation*}
\mathscr{H}:=\mathfrak{F}_W \otimes \mathfrak{F}_L.
\end{equation*}
In other words, $\mathscr{H}$ is the tensor product of $14$ Fock spaces, $2$ symmetric Fock spaces for the bosons $W^\pm$ and $12$ anti-symmetric Fock spaces for the fermions.

The number operators for neutrinos and antineutrinos are defined by
\begin{equation*}
N_{\nu_l} := \sum_{\epsilon = \pm} \int c^*_{l,\epsilon}(\xi_2) c_{l,\epsilon}(\xi_2) d\xi_2 , \quad N_{\mathrm{neut}} := \sum_{l=1}^3 N_{\nu_l},
\end{equation*}
and likewise
\begin{equation*}
N_{l}:=\sum_{\epsilon = \pm} \int b^*_{l,\epsilon}(\xi_1) b_{l,\epsilon}(\xi_1) d\xi_1,  \quad N_{\mathrm{lept}} := \sum_{l=1}^3 N_l, \quad N_{W}:=\sum_{\epsilon = \pm} \int a^*_{\epsilon}(\xi_3) a_{\epsilon}(\xi_3) d\xi_3.
\end{equation*}
The total number operator is
\begin{align*}
N  = N_{\mathrm{lept}} + N_{\mathrm{neut}} + N_W.
\end{align*}

\subsection{Self-adjointness }\label{sec:self-adj}

Using the Kato-Rellich theorem together with $N_\tau$ estimates \cite{GlJa77_01}, it is proven in \cite{BaGu09_01} that the total Hamiltonian $H$ defined in \eqref{total_hamiltonian_g} is a self-adjoint operator in $\mathscr{H}$, with domain $\mathscr{D}\left( H\right)=\mathscr{D}\left( H_0\right) $, provided that the kernels $G_{l,\epsilon}^{(j)}$ are square integrable and $g \ll 1$. In this section, we extend this result to any value of $g$.

Recall that the notation ``$G \in L^2$'' means that, for all $j$, $l$ and $\epsilon$, $G_{l,\epsilon}^{(j)}$ is square integrable. We denote by $\|G\|_2$ the sum over  $j$, $l$ and $\epsilon$ of the $L^2$-norms of $G_{l,\epsilon}^{(j)}$,
\begin{equation*}
\| G \|_2 := \sum_{j,l,\epsilon} \big \| G_{l,\epsilon}^{(j)} \big \|_2.
\end{equation*}
\begin{Th}
\label{Th2M}
Suppose that $G \in L^2$. Then, for all $g \in \mathbb{R}$, the Hamiltonian $H$ in \eqref{total_hamiltonian_g} is self-adjoint with domain $\mathscr{D}\left( H\right)=\mathscr{D}\left( H_0\right) $.
\end{Th} 
The proof of Theorem \ref{Th2M} will be a consequence of the following two lemmas. The first one is a direct application of the $N_\tau$ estimates of Glimm and Jaffe (see \cite[Proposition 1.2.3(c)]{GlJa77_01}).
\begin{lem}
\label{lemSCW_1}
Suppose that $G \in L^2$. Then
\begin{equation*}
\|H_I (N_{\mathrm{lept}} + N_W + 1)^{-1} \| \lesssim \| G \|_2.
\end{equation*}
\end{lem}
\begin{proof}
Consider for instance the term
\begin{equation}
H_{I,1,+}^{(1)} := \int G^{(1)}_{1,+}(\xi_1, \xi_2, \xi_3) b^*_{1,+}(\xi_1)c^*_{1,-}(\xi_2)a_{+}(\xi_3) d \xi_1 d \xi_2 d \xi_3 , \label{H_Il1+-}
\end{equation}
occurring in $H_I$. Proceeding as in Proposition 1.2.3(b) of \cite{GlJa77_01}, one easily verifies that
\begin{equation*}
\Big\|N_{\mathrm{lept}}^{-\frac12} H_{I,1,+}^{(1)} N_W^{-\frac12} \Big \| \lesssim \| G^{(1)}_{1,+} \|_2. 
\end{equation*}
Using that $N_{\mathrm{lept}}^{-\frac12} H_{I,1,+}^{(1)} = H_{I,1,+}^{(1)} (N_{\mathrm{lept}}+1)^{-\frac12}$, we obtain that 
\begin{equation*}
\| H_{I,1,+}^{(1)} (N_{\mathrm{lept}} + N_W + 1)^{-1} \| \lesssim \| G \|_2.
\end{equation*}
The other terms occurring in $H_I$ can be treated in an analogous way.
\end{proof}
The second lemma is a slight generalization of \cite[Proposition 1.2.3(c)]{GlJa77_01}.
\begin{lem}
\label{lemSCW}
Suppose that for all $j$, $l$, $\epsilon$, $s_1$, $s_2$ and $\lambda$, $G_{l,\epsilon}^{(j)}( s_1 , \cdot , s_2 , \cdot , \lambda , \cdot )$ belongs to the Schwartz space $\mathscr{S}( \mathbb{R}^9 )$. Then
\begin{equation*}
\|H_I (N_W+1)^{-\frac{1}{2}}\| \le \mathrm{C}(G ) ,
\end{equation*}
where $\mathrm{C}(G)$ is a positive constant depending on $G_{l,\epsilon}^{(j)}$.
\end{lem}
\begin{proof}
Let $\{e_i\}$ be an orthonormal basis of $L^2(\mathbb{R})$ composed of eigenvectors, corresponding to the eigenvalues $\lambda_{i}=(2i+1)$, of the one-dimensional harmonic oscillator $h_{\mathrm{ho}} :=  - \frac{d^2}{dx^2} + x^2$. We consider the orthonormal basis $\{ e_{i_1} \otimes \cdots \otimes e_{i_9} \}$ in $L^2( \mathbb{R}^9 )$. Below we use the notation $\mathrm{i}_1 = ( i_1 , i_2 , i_3)$, $\mathrm{i}_2 = ( i_3 , i_4 , i_5)$, $\mathrm{i}_3 = ( i_6 , i_7 , i_9)$ and a sum over $(\mathrm{i}_1,\mathrm{i}_2,\mathrm{i}_3)$ corresponds to a sum over $( i_1 , \dots , i_9 ) \in \mathbb{N}^9$. Moreover $e_{\mathrm{i}_1} := e_{i_1}\otimes e_{i_2} \otimes e_{i_3}$ and likewise for $e_{\mathrm{i}_2}$ and $e_{\mathrm{i}_3}$.

As in the proof of the previous lemma, we consider for instance the term $H_{I,1,+}^{(1)}$ occurring in $H_I$, see \eqref{H_Il1+-}. Let
\begin{equation*}
H_{I,1,+}^{(1)}(s_1,s_2,\lambda) := \int G^{(1)}_{1,+}(s_1, p_1 , s_2 , p_2, \lambda , p_3) b^*_{1,+}( p_1 , s_1 )c^*_{1,-}( p_2 , s_2 )a_{+}( p_3 , \lambda ) d p_1 d p_2 d p_3 , 
\end{equation*}
so that
\begin{equation*}
H_{I,1,+}^{(1)} = \sum_{s_1,s_2,\lambda} H_{I,1,+}^{(1)}(s_1,s_2,\lambda).
\end{equation*}
To prove the lemma, it suffices to verify that, for any fixed $s_1$, $s_2$ and $\lambda$,
\begin{equation*}
\| H_{I,1,+}^{(1)}(s_1,s_2,\lambda) (N_W+1)^{-\frac{1}{2}}\| \le \mathrm{C}(G ) .
\end{equation*}

Decomposing $G^{(1)}_{1,+}(\cdot,s_1,\cdot,s_2,\cdot,s_3)$ into the orthonormal basis $\{ e_{i_1} \otimes \cdots \otimes e_{i_9} \}$, we see that
\begin{equation*}
H_{I,1,+}^{(1)}(s_1,s_2,\lambda) = \sum_{\mathrm{i}_1,\mathrm{i}_2,\mathrm{i}_3} \alpha_{\mathrm{i}_1,\mathrm{i}_2,\mathrm{i}_3} b^*_{1,+,s_1}(e_{\mathrm{i}_1}) \otimes c^*_{1,-,s_2}(e_{\mathrm{i}_2}) \otimes a_{+,\lambda}(e_{\mathrm{i}_3}) ,
\end{equation*}
where we have set $\alpha_{\mathrm{i}_1,\mathrm{i}_2,\mathrm{i}_3} := \langle e_{\mathrm{i}_1} \otimes e_{\mathrm{i}_2} \otimes e_{\mathrm{i}_3} , G^{(1)}_{1,+}(\cdot,s_1,\cdot,s_2,\cdot,s_3) \rangle$, $b^*_{1,+,s_1}(h)= \int_{\mathbb{R}^3} h(p_1) b^*_{1,+}(s_1,p_1) dp_1$, for $h \in L^2(\mathbb{R}^3)$, and likewise for $c^*_{1,-,s_2}(h)$ and $a_{+,\lambda}(h)$. This yields
\begin{align*}
& \left\| H_{I,1,+}^{(1)}(s_1,s_2,\lambda) (N_W+1)^{-\frac{1}{2}}  \right\| \\
& = \sum_{\mathrm{i}_1,\mathrm{i}_2,\mathrm{i}_3} \left\| \alpha_{\mathrm{i}_1,\mathrm{i}_2,\mathrm{i}_3} b^*_{1,+,s_1}(e_{\mathrm{i}_1}) \otimes c^*_{1,-,s_2}(e_{\mathrm{i}_2}) \otimes a_{+,\lambda}(e_{\mathrm{i}_3}) (N_W+1)^{-\frac{1}{2}}  \right\|  \leq  \sum_{\mathrm{i}_1,\mathrm{i}_2,\mathrm{i}_3} | \alpha_{\mathrm{i}_1,\mathrm{i}_2,\mathrm{i}_3} | .
\end{align*}
Observe that
\begin{align*}
\alpha_{\mathrm{i}_1,\mathrm{i}_2,\mathrm{i}_3} = \Big ( \prod_{\ell=1}^9 \frac{1}{(2i_\ell+1)} \Big ) \big \langle e_{\mathrm{i}_1} \otimes e_{\mathrm{i}_2} \otimes e_{\mathrm{i}_3} , ( \otimes_{\ell=1}^9 h_{\mathrm{ho}} ) G^{(1)}_{1,+}(\cdot,s_1,\cdot,s_2,\cdot,s_3) \big \rangle.
\end{align*}
The Cauchy-Schwarz inequality then gives
\begin{align*}
\sum_{\mathrm{i}_1,\mathrm{i}_2,\mathrm{i}_3} | \alpha_{\mathrm{i}_1,\mathrm{i}_2,\mathrm{i}_3} | \le \Big ( \sum_{\mathrm{i}_1,\mathrm{i}_2,\mathrm{i}_3} \Big ( \prod_{\ell=1}^9 \frac{1}{(2i_\ell+1)^2} \Big ) \Big )^{\frac12} \big\|( \otimes_{\ell=1}^9 h_{\mathrm{ho}} ) G^{(1)}_{1,+}(\cdot,s_1,\cdot,s_2,\cdot,s_3) \big \|_{L^2(\mathbb{R}^9)} ,
\end{align*}
which concludes the proof.
\end{proof}
Now we are ready to prove Theorem \ref{Th2M}.
\begin{proof}[Proof of Theorem \ref{Th2M}.]
In this proof, we underline the dependence of the interaction Hamiltonian on the kernels $G_{l,\epsilon}^{(j)}$ by writing $H_I = H_I( G )$. According to Lemma \ref{lemSCW_1}, for $\Psi \in \mathscr{D} (N_{\mathrm{lept}} + N_W)$, there exists $a>0$ such that, for all $G \in L^2$,
\begin{equation*}
 \|H_I(G) \Psi\| \leq a \| G \|_2 \left( \| (N_{\mathrm{lept}} + N_W ) \Psi \| +  \|\Psi\| \right).
\end{equation*}
Let $\delta > 0$. There exist $G_\delta \in L^2$ such that for all $j$, $l$, $\epsilon$, $s_1$, $s_2$ and $\lambda$, $G_{\delta,l,\epsilon}^{(j)}( s_1 , \cdot , s_2 , \cdot , \lambda , \cdot )$ belongs to the Schwartz space $\mathscr{S}( \mathbb{R}^9 )$ and
\begin{equation*}
\| G - G_\delta \|_2 \leq \frac{\delta}{a}.
\end{equation*}
Hence
\begin{equation}
 \| ( H_I(G) - H_I(G_\delta) ) \Psi\| \leq \varepsilon  \left( \| (N_{\mathrm{lept}} + N_W ) \Psi \| +  \|\Psi\| \right). \label{a1}
\end{equation}
Moreover, by Lemma \ref{lemSCW},
\begin{equation}
\| H_I( G_\delta ) \Psi\| \leq \mathrm{C}( G_\delta ) \| (N_W+1)^{\frac{1}{2}} \Psi\| , \label{a2}
\end{equation}
with $\mathrm{C}( G_\delta )>0$. For $\mu > 0$, we have that
\begin{align*}
\| (N_W+1)^{\frac{1}{2}} \Psi\|^2 \le \mu \| N_W \Psi \|^2 + ( ( 4 \mu )^{-1} + 1 ) \| \Psi \|^2 \le \big ( \mu^{\frac12} \| N_W \Psi \| + ( ( 4 \mu )^{-1} + 1 )^{\frac12} \| \Psi \| \big )^2.
\end{align*}
Inserting this into \eqref{a2} and choosing $\mu^{\frac12} = \varepsilon \mathrm{C}( G_\delta )^{-1}$, this implies that
\begin{equation}
\| H_I( G_\delta ) \Psi\| \leq \delta \|N_W \Psi\| + \mathrm{c}_\delta \|\Psi\| , \label{a3}
\end{equation}
for some positive constant $\mathrm{c}_\delta$.

Equations \eqref{a1} and \eqref{a3} show that $H_I(G)$ is relatively $(N_{\mathrm{lept}} + N_W)$-bounded with relative bound $0$. Since the masses $m_l$ and $m_W$ are positive, it is not difficult to deduce that $N_{\mathrm{lept}} + N_W$ is relatively $H_0$-bounded. Therefore $H_I(G)$ is relatively $H_0$-bounded with relative bound $0$. Applying the Kato-Rellich theorem concludes the proof.
\end{proof}
\subsection{Technical estimates}\label{subsec:estimates}

This section is devoted to some technical lemmas and properties which will be used later. Proofs and notations in this section are close to those of \cite{DeGe99_01}.

\subsubsection{Number-energy estimates}

We begin with the following result in the case where all particles are supposed to be massive.
\begin{lem}\label{lemV1Numbestimates}
Suppose that the masses of the neutrinos $m_{\nu_e}$, $m_{\nu_{\mu}}$, $m_{\nu_{\tau}}$ are positive and consider the Hamiltonian \eqref{total_hamiltonian_g} with $H_I$ given by \eqref{interaction_term}. Assume that $G \in L^2$.
\begin{itemize}
\item[(i)] For all $m \in \mathbb{Z}$, uniformly for z in a compact set of $\left\{ z \in \mathbb{C}, \pm  |\Im{z}| >0 \right\}$, the operator $(N+1)^{-m}(H-z)^{-1}(N+1)^{m+1}$ extends to a bounded operator satisfying
\begin{equation*}
\big\|(N+1)^{-m}(H-z)^{-1}(N+1)^{m+1}\big\| = \mathcal{O}(|\Im(z)|^{ - \alpha_{m}}) ,
\end{equation*}
where $\alpha_{m}$ denotes an integer depending on $m$.
\item[(ii)] Let $\chi \in \mathrm{C}^{\infty}_0(\mathbb{R})$. Then, for all $m,p \in \mathbb{N}$,  $N^m \chi(H) N^p$ extends to a bounded operator.
\end{itemize}
\end{lem}
\begin{proof}
As in the proof of Theorem \ref{Th2M}, we underline the dependence of the interaction Hamiltonian on the kernels $G_{l,\epsilon,\epsilon'}^{(j)}$ by writing $H_I = H_I( G )$. First, observe that
\begin{equation*}
[H,N] = [H_I(G),N].
\end{equation*}
A direct computation gives
\begin{align}
& [H_I^{(1)}(G) , N_{\mathrm{lept}} ] = [H_I^{(1)}(G) , N_{\mathrm{neut}} ] = - [H_I^{(1)}(G) , N_{W} ] = \mathrm{i} H_I^{(1)}( \mathrm{i} G ), \label{commut_1} \\
& [H_I^{(2)}(G) , N_{\mathrm{lept}} ] = [H_I^{(2)}(G) , N_{\mathrm{neut}} ] = [H_I^{(2)}(G) , N_{W} ] = \mathrm{i} H_I^{(2)}( \mathrm{i} G ), \label{commut_2} \\
& [H_I^{(3)}(G) , N_{\mathrm{lept}} ] = - [H_I^{(3)}(G) , N_{\mathrm{neut}} ] = [H_I^{(3)}(G) , N_{W} ] = \mathrm{i} H_I^{(3)}( \mathrm{i} G ), \label{commut_3} \\
& [H_I^{(4)}(G) , N_{\mathrm{lept}} ] = - [H_I^{(4)}(G) , N_{\mathrm{neut}} ] = - [H_I^{(4)}(G) , N_{W} ] = \mathrm{i} H_I^{(4)}( \mathrm{i} G ). \label{commut_4}
\end{align}
In particular, since $G \in L^2$, Lemma \ref{lemSCW_1} together with the fact that $N_{\mathrm{lept}} + N_W$ is relatively $H$-bounded show that $\| [H,N](H+z)^{-1} \| = \mathcal{O}( |\Im{z}|^{-1} )$. Likewise,
\begin{align}
\| \mathrm{ad}_N^j(H) (H-z)^{-1} \| = \mathcal{O}( |\Im{z}|^{-1} ), \quad j \in \mathbb{N}. \label{est_commutator}
\end{align}

Next, for $m \in \mathbb{N}$, commuting $N^m$ through $(H-z)^{-1}$, we obtain that
\begin{equation*}
(H-z)^{-1}N^m = N^m (H-z)^{-1}+ \sum_{k=1}^m N^{m-l}(H-z)^{-1}B_{m,k}(z) ,
\end{equation*}
where, by \eqref{est_commutator}, the operator $B_{m,k}(z)$ satisfies $\|B_{m,k}(z)\| = \mathcal{O}( |\Im{z}|^{-c_{m,k}} )$, with $c_{m,k}$ a positive integer. Therefore
\begin{align*}
(N+1)^{-m} (H-z)^{-1} (N+1)^{m+1} &=  (N+1) (H-z)^{-1} + (N+1)^{-m} \sum_{i=1}^m N^{m-l} (H-z)^{-1} B_{m,k}(z) \\
&= \mathcal{O}( |\Im z|^{-\alpha_m}),
\end{align*}
where in the second equality we used that $N$ is $H$-bounded. This proves (i).

To prove (ii), let $\chi \in \mathrm{C}^{\infty}_0(\mathbb{R})$ and $m,p \in \mathbb{N}$. Then 
\begin{align*}
N^m \chi(H) N^p = \Big ( \prod_{k=0}^{m-1} N^{m-k} (H+\mathrm{i})^{-1} N^{k-m+1} \Big ) (H+\mathrm{i})^m \chi(H)(H+\mathrm{i})^p \Big ( \prod_{k=0}^{p-1} N^{k-p+1} (H+\mathrm{i})^{-1} N^{p-k} \Big ) ,
\end{align*}
where we used the convention $\prod_{k=0}^{m-1} A_k := A_0 A_1 \dots A_{m-1}$ for any operators $A_0,\dots,A_{m-1}$. Hence (i) implies (ii).
\end{proof}
In the case where the masses of the neutrinos vanish, the statement of the previous lemma can be modified as follows.
\begin{lem}\label{lemV1NumbestimatesV2}
Suppose that the masses of the neutrinos $m_{\nu_e}$, $m_{\nu_{\mu}}$, $m_{\nu_{\tau}}$ vanish and consider the Hamiltonian \eqref{total_hamiltonian_g} with $H_I$ given by \eqref{interaction_term}. 
Assume that $G \in L^2$.
\begin{itemize}
\item[(i)] For all $m \in \mathbb{Z}$, uniformly for z in a compact set of $\left\{ z \in \mathbb{C}, \pm  |\Im{z}| >0 \right\}$, the operator $(N+1)^{-m}(H-z)^{-1}(N+1)^{m}$ extends to a bounded operator satisfying
\begin{equation*}
\big\|(N+1)^{-m}(H-z)^{-1}(N+1)^{m}\big\| = \mathcal{O}(|\Im(z)|^{ - \alpha_{m}}) ,
\end{equation*}
where $\alpha_{m}$ denotes an integer depending on $m$. Moreover, the operator $(N_{\mathrm{lept}} + N_W+1)^{-m}(H-z)^{-1}(N_{\mathrm{lept}} + N_W+1)^{m+1}$ extends to a bounded operator satisfying
\begin{equation*}
\big\|(N_{\mathrm{lept}} + N_W +1)^{-m}(H-z)^{-1}( N_{\mathrm{lept}} + N_W + 1)^{m+1}\big\| = \mathcal{O}(|\Im(z)|^{ - \beta_{m}}) ,
\end{equation*}
where $\beta_{m}$ denotes an integer depending on $m$.
\item[(ii)] Let $\chi \in \mathrm{C}^{\infty}_0(\mathbb{R})$. Then, for all $m,p \in \mathbb{N}$,  $(N_{\mathrm{lept}} + N_W)^m \chi(H) (N_{\mathrm{lept}} + N_W)^p$ extends to a bounded operator.
\end{itemize}
\end{lem}
\begin{proof}
The proof is analogous to that of Lemma \ref{lemV1Numbestimates}. The only difference is that $N_{\mathrm{lept}}$ and $N_W$ are still relatively $H$-bounded, but $N_{\mathrm{neut}}$ is not anymore.
\end{proof}
\subsubsection{Number energy estimates in the ``extended'' setting}

In the remainder of this section, we give results that concern an auxiliary ``extended Hamiltonian''. The latter is introduced similarly as in \cite{DeGe99_01}. It will be used in Section \ref{sec:AC}.

 The ``extended Hilbert space'' is 
 \begin{equation*}
 \mathscr{H}^{\mathrm{ext}} := \mathscr{H} \otimes \mathscr{H}
 \end{equation*}
 and the extended Hamiltonian, acting on $\mathscr{H}^{\mathrm{ext}}$, is defined by 
 \begin{equation*}
 H^{\mathrm{ext}} := H\otimes \mathds{1}+\mathds{1}\otimes H_0.
 \end{equation*}
 More details on the extended objects are given in Appendix \ref{sec:standard_def}. Roughly speaking, the idea is that the first component of $\mathscr{H}^{\mathrm{ext}}$ corresponds to bound states, while the second component corresponds to states localized near infinity. The number operators in the extended setting are defined by 
 \begin{equation*}
 N_{\mathrm{lept},0} := N_{\mathrm{lept}}\otimes\mathds{1}, \quad N_{\mathrm{lept},\infty}:=\mathds{1}\otimes N_{\mathrm{lept}}, 
 \end{equation*}
 and likewise for $N_{\mathrm{neut}}$, $N_W$ and the total number operator $N$.

The proof of Lemma \ref{lemV1Numbestimates} can be adapted in a straightforward way to obtain the following results. \pagebreak[2]
\begin{lem}
\label{lemV2Numbestimates}
$\quad$
\begin{itemize}
\item[(i)] Under the conditions of Lemma \ref{lemV1Numbestimates}, for all $m \in \mathbb{Z}$, we have that 
\begin{equation*}
\big\|(N_0+N_\infty+1)^{-m} (H^{\mathrm{ext}} - z )^{-1}  ( N_0 + N_\infty + 1 )^{m+1} \big \| = \mathcal{O}(|\Im(z)|^{-\alpha_{m}}) ,
\end{equation*}
uniformly for z in a compact set of $\left\{ z \in \mathbb{C}, \pm  |\Im{z}| >0 \right\}$, where $\alpha_{m}$ denotes an integer depending on $m$. Moreover, for all $\chi \in \mathrm{C}^{\infty}_0(\mathbb{R})$ and $m,p \in \mathbb{N}$, $(N_0 + N_\infty )^m \chi(H^{\mathrm{ext}}) (N_0 + N_\infty )^p$ extends to a bounded operator.
\item[(ii)] Under the conditions of Lemma \ref{lemV1NumbestimatesV2}, for all $m \in \mathbb{Z}$, we have that
\begin{equation*}
\big\|(N_0+N_\infty+1)^{-m} (H^{\mathrm{ext}} - z )^{-1}  ( N_0 + N_\infty + 1 )^{m} \big \| = \mathcal{O}(|\Im(z)|^{-\alpha_{m}}) ,
\end{equation*}
uniformly for z in a compact set of $\left\{ z \in \mathbb{C}, \pm  |\Im{z}| >0 \right\}$, where $\alpha_{m}$ denotes an integer depending on $m$. Moreover,
\begin{align*}
& \big\| ( N_{\mathrm{lept},0}+N_{W,0}  +  N_{\mathrm{lept},\infty}+N_{W,\infty} +1 )^{-m} (H^{\mathrm{ext}}-z)^{-1} \\
& \qquad \qquad \qquad (N_{\mathrm{lept},0}+N_{W,0} + N_{\mathrm{lept},\infty}+N_{W,\infty} +1)^{m+1} \big \| = \mathcal{O}(|\Im(z)|^{-\alpha_{m}}),
\end{align*}
and for all $\chi \in \mathrm{C}^{\infty}_0(\mathbb{R})$ and $m,p \in \mathbb{N}$, $( N_{\mathrm{lept},0}+N_{W,0}   +   N_{\mathrm{lept},\infty}+N_{W,\infty} )^m \chi(H^{\mathrm{ext}}) ( N_{\mathrm{lept},0}+N_{W,0}   +   N_{\mathrm{lept},\infty}+N_{W,\infty} )^p$ extends to a bounded operator.
\end{itemize}
\end{lem}

Let $(j_{1,0}, \dots, j_{14,0}) \in \mathrm{C}^{\infty}_0(\mathbb{R}^3)$ and $(j_{1,\infty}, \dots, j_{14,\infty}) \in \mathrm{C}^{\infty}(\mathbb{R}^3)$ be families of function satisfying $j_{\ell,0} \ge 0$, $j_{\ell,\infty} \ge 0$, $j^2_{\ell,0}+j^2_{\ell,\infty} = 1$ and $ j_{\ell,0}=1 $ near $0$. Recall that the total Hilbert space $\mathscr{H}$ is a tensor product of $14$ Fock spaces. We set $j=\left((j_{1,0},j_{1,\infty}), \dots, (j_{14,0},j_{14,\infty}) \right)$ and, for any $R\geq 1$, $j^R = ((j^R_{1,0},j^R_{1,\infty}), \dots, (j^R_{14,0},j^R_{14,\infty}) )$, with $j^R_{\ell,\sharp}=j_{\ell,\sharp}(\frac{x}{R})$, $x = \mathrm{i} \nabla_{p}$, and $p$ is the momentum of the particle labeled by $\ell \in \{ 1 , \dots , 14 \}$. The definition of the partition of unity $\check{\Gamma}(j^R) : \mathscr{H} \to \mathscr{H}^{\mathrm{ext}} $, adapted from \cite{DeGe99_01} and \cite{Am04_01}, is recalled in Appendix \ref{sec:standard_def}.

Recall that we write ``$G \in \mathbb{H}^\mu$'' if, for all $j$, $l$, $\epsilon$, $s_1$, $s_2$ and $\lambda$, $G_{l,\epsilon}^{(j)}( s_1 , \cdot , s_2 , \cdot , \lambda , \cdot )$ belongs to the Sobolev space $\mathbb{H}^{\mu}( \mathbb{R}^9 )$. Note that it is equivalent to assume that, for all $j$, $l$, $\epsilon$, $s_1$, $s_2$ and $\lambda$,
\begin{equation*}
\forall R >1, \, \big \| \mathds{1}_{[R,\infty)}(|x_i|)G_{l,\epsilon}^{(j)}( s_1 , \cdot , s_2 , \cdot , \lambda , \cdot )\big\|_2 \lesssim  R^{-\mu},
\end{equation*}
where  $x_i = \mathrm{i} \nabla_{p_i}$, $i=1,2,3$.
\begin{lem}
\label{numberestimates3}
Suppose that the masses of the neutrinos $m_{\nu_e}$, $m_{\nu_{\mu}}$, $m_{\nu_{\tau}}$ are positive and consider the Hamiltonian \eqref{total_hamiltonian_g} with $H_I$ given by \eqref{interaction_term}. Assume that $G \in L^2$. Let $j^R$ be defined as above. Then
\begin{equation}\label{1st_estim}
 (H^{\mathrm{ext}}+\mathrm{i})^{-1} \check{\Gamma}(j^R)-\check{\Gamma}(j^R)(H+\mathrm{i})^{-1} = o(R^{0}) , \quad R \to \infty.
\end{equation}
In particular, for any $\chi, \chi' \in \mathrm{C}_0^\infty( \mathbb{R} )$, we have that
\begin{equation}\label{3rd_estim}
\big \{ \chi(H^{\mathrm{ext}}) \check{\Gamma}(j^R)-\check{\Gamma}(j^R)\chi(H) \big \} \chi'(H) = o(R^0), \quad R \to \infty.
 \end{equation}
If $G \in \mathbb{H}^\mu$ with $\mu > 0$, then
\begin{equation}\label{2nd_estim}
(H^{\mathrm{ext}}+\mathrm{i})^{-1} \check{\Gamma}(j^R)-\check{\Gamma}(j^R)(H+\mathrm{i})^{-1} = \mathcal{O}(R^{-\min(1,\mu)}) , \quad R \to \infty ,
\end{equation}
and in particular, for any $\chi, \chi' \in \mathrm{C}_0^\infty( \mathbb{R} )$ and $\mu \ge 1$, we have that
\begin{equation}\label{3rd_estim_2}
\big \{ \chi(H^{\mathrm{ext}}) \check{\Gamma}(j^R)-\check{\Gamma}(j^R)\chi(H) \big \} \chi'(H) \in \mathcal{O}(R^{-1}) , \quad R \to \infty.
 \end{equation}
\end{lem}
\begin{proof}
We first note that
\begin{align*}
& (H^{\mathrm{ext}}+\mathrm{i})^{-1} \check{\Gamma}(j^R)-\check{\Gamma}(j^R)(H+\mathrm{i})^{-1} = (H^{\mathrm{ext}}+\mathrm{i})^{-1} \big ( \check{\Gamma}(j^R) H - H^{\mathrm{ext}} \check{\Gamma}(j^R) \big ) (H+\mathrm{i})^{-1} ,
\end{align*}
and a direct computation gives
\begin{equation*}
H^{\mathrm{ext}}_0 \check{\Gamma}(j^R) - \check{\Gamma}(j^R)H_0 = \mathrm{d} \check{\Gamma} (j^R,[\omega,j^R]) ,
\end{equation*}
where the operator $\mathrm{d} \check{\Gamma} (j^R,[\omega,j^R])$ is defined in Appendix \ref{sec:standard_def}. It follows from Lemma \ref{estimationN2} that
\begin{align*}
 \left\| \mathrm{d} \check{\Gamma} (j^R,[\omega,j^R]) (  N_{0} + N_{\infty} +1 )^{-1} \right\|  & = \left\| (  N_{0} + N_{\infty} +1 )^{-\frac12} \mathrm{d} \check{\Gamma} (j^R,[\omega,j^R]) ( N + 1 )^{-\frac12} \right\| \\
& \le \big \| [\omega,j^R]^*[\omega,j^R] \big \|^{\frac12} = \mathcal{O}( R^{-1} ).
\end{align*}
Moreover, we have that
\begin{equation*}
( H_I(G)\otimes\mathds{1}) \check{\Gamma}(j^R)-\check{\Gamma}(j^R) H_I(G) = \sum_{j=1}^4 ( H_I^{(j)}(G) \otimes\mathds{1}) \check{\Gamma}(j^R)-\check{\Gamma}(j^R) H_I^{(j)}(G).
\end{equation*}
Considering for instance the term
\begin{equation*}
H_{I,1,+}^{(1)}(G) = \int G^{(1)}_{1,+}(\xi_1, \xi_2, \xi_3) b^*_{1,+}(\xi_1)c^*_{1,-}(\xi_2)a_{+}(\xi_3) d \xi_1 d \xi_2 d \xi_3 , 
\end{equation*}
occurring in $H_I^{(1)}(G)$. Using the intertwining properties of Lemma \ref{lem:B6}, one verifies that
\begin{equation*}
( H_{I,1,+}^{(1)}(G) \otimes\mathds{1}) \check{\Gamma}(j^R)-\check{\Gamma}(j^R) H_{I,1,+}^{(1)}(G)  
\end{equation*}
can be expressed as a sum of operators of the form
\begin{equation*}
\int j^R( x_1 , x_2 , x_3 ) G^{(1)}_{1,+}(\xi_1, \xi_2, \xi_3) b^{*,\sharp}_{1,+}(\xi_1) c^{*,\sharp}_{1,-}(\xi_2) a^\sharp_{+}(\xi_3) d \xi_1 d \xi_2 d \xi_3 ,
\end{equation*}
where $j^R(x_1,x_2,x_3)$ stands for either $1-j^R_{\ell_1,0}(x_{1})j^R_{\ell_2,0}(x_2)j^R_{\ell_3,0}(x_3)$ or $j^R_{\ell_1,\sharp}(x_1)j^R_{\ell_2,\sharp}(x_2)j^R_{\ell_3,\sharp}(x_3)$, with at least one of the $j^R_{\ell,\sharp}$  equal to $j^R_{\ell,\infty}$. Moreover $a^\sharp_{+}(\xi_3)$ stands for either $a_{+}(\xi_3) \otimes \mathds{1}$ or $\mathds{1} \otimes a_{+}(\xi_3)$ and likewise for $b^{*,\sharp}_{1,+}(\xi_1)$ and $c^{*,\sharp}_{1,-}(\xi_2)$.

Therefore, proceeding as in Lemma \ref{lemSCW_1}, one deduces that if $G \in L^2$, then
\begin{equation*}
\big\|\big\{(H_I(G)\otimes\mathds{1}) \check{\Gamma}(j^R)-\check{\Gamma}(j^R) H_I(G) \big\} (N_0+N_\infty+1)^{-1}\big\| = o(R^{0}) , \quad R \to \infty.
\end{equation*} 
Similarly, if $G \in \mathbb{H}^\mu$, we obtain that 
\begin{equation*}
\big\|\big\{(H_I(G)\otimes\mathds{1}) \check{\Gamma}(j^R)-\check{\Gamma}(j^R) H_I(G) \big\} (N_0+N_\infty+1)^{-1}\big\| = \mathcal{O}(R^{-\mu}), \quad R \to \infty.
\end{equation*}
Putting together the previous estimates proves \eqref{1st_estim} and \eqref{2nd_estim}.

To prove \eqref{3rd_estim} and \eqref{3rd_estim_2}, let $\tilde{\chi}\in \mathrm{C}^{\infty}_{0}(\mathbb{\mathbb{C}})$ be an almost analytic extension of $\chi$ satisfying: 
\begin{align*}
& \tilde{\chi}|_{\mathbb{R}} =  \chi , \quad  |\partial_{\bar{z}}\tilde{\chi}(z)| \leq  \mathrm{C}_n |\Im(z)|^n,~~n\in \mathbb{N}.
\end{align*}
Using the Helffer-Sj{\"o}strand functional calculus, we have that 
\begin{align*}
&  \left(\chi(H^{\mathrm{ext}}) \check{\Gamma}(j^R)-\check{\Gamma}(j^R)\chi(H) \right)\chi'(H) \\
&  \int \partial_{\bar{z}}\tilde{\chi}(z)  \left((z-H^{\mathrm{ext}})^{-1} \check{\Gamma}(j^R)-\check{\Gamma}(j^R)(z-H)^{-1} \right)\chi'(H)dz \wedge d\bar{z}.
\end{align*}
Combining \eqref{1st_estim}, \eqref{2nd_estim} and Lemma \ref{lemV1Numbestimates}, we obtain \eqref{3rd_estim} and \eqref{3rd_estim_2}.
\end{proof}

In the case where the masses of the neutrinos $m_{\nu_e}$, $m_{\nu_{\mu}}$, $m_{\nu_{\tau}}$ vanish, the statement of Lemma \ref{numberestimates3} is modified as follows.
\begin{lem}
\label{numberestimates13}
Suppose that the masses of the neutrinos $m_{\nu_e}$, $m_{\nu_{\mu}}$, $m_{\nu_{\tau}}$ vanish and consider the Hamiltonian \eqref{total_hamiltonian_g} with $H_I$ given by \eqref{interaction_term}. Assume that $G \in L^2$. Let $j^R$ be defined as above. Then
\begin{equation}\label{1st_estim_massless}
 (H^{\mathrm{ext}}+\mathrm{i})^{-1} \check{\Gamma}(j^R)-\check{\Gamma}(j^R)(H+\mathrm{i})^{-1} (N_{\mathrm{neut}}+1)^{-1} \in o(R^{0}).
\end{equation}
If $G \in \mathbb{H}^\mu$ with $\mu > 0$, then
\begin{equation}\label{2nd_estim_massless}
 (H^{\mathrm{ext}}+\mathrm{i})^{-1} \check{\Gamma}(j^R)-\check{\Gamma}(j^R)(H+\mathrm{i})^{-1} (N_{\mathrm{neut}}+1)^{-1} \in \mathcal{O}(R^{-\min(1,\mu)}).
\end{equation}
In particular, for any $\chi, \chi' \in \mathrm{C}_0^\infty( \mathbb{R} )$ and $\mu \ge 1$, we have that
\begin{equation}\label{3rd_estim_massless}
 \left(\chi(H^{\mathrm{ext}}) \check{\Gamma}(j^R)-\check{\Gamma}(j^R)\chi(H) \right) \chi'(H) (N_{\mathrm{neut}}+1)^{-1}  \in \mathcal{O}(R^{-1}).
 \end{equation}
\end{lem}
\begin{proof}
It suffices to adapt the proof of Lemma \ref{numberestimates3}, using Lemma \ref{lemV1NumbestimatesV2}. 
\end{proof}
\section{Spectral theory}\label{sec:spectral}

In this section, we do the spectral analysis of the Hamiltonian $H$. Section \ref{subsec:spectral_massive} is devoted to the case where the masses of the neutrinos are supposed to be positive while Section \ref{subsec:spectral_massless} concerns the case where the neutrinos are supposed to be massless. In both cases, we begin with recalling results giving the existence of a ground state, next we study the structure of the essential spectrum by means of suitable versions of Mourre's conjugate operator method.

\subsection{Massive neutrinos}\label{subsec:spectral_massive}

In this section, we do the spectral analysis of $H$ under the assumption that
\begin{equation*}
m_\nu = \min ( m_{\nu_e} , m_{\nu_\mu} , m_{\nu_\tau} ) > 0.
\end{equation*}

\subsubsection{Existence of a ground state and location of the essential spectrum}

Recall the notation $E= \inf \sigma( H )$. The next theorem shows, in particular, that $H$ has a ground state, i.e., that $E$ is an eigenvalue of $H$.
\begin{Th}
\label{Thess2}
Suppose that the masses of the neutrinos $m_{\nu_e}$, $m_{\nu_{\mu}}$, $m_{\nu_{\tau}}$ are positive and consider the Hamiltonian \eqref{total_hamiltonian_g} with $H_I$ given by \eqref{interaction_term}. Assume that $G \in L^2$. Then
\begin{equation*}
\sigma_{\mathrm{ess}}(H)=[E + m_{\nu}, \infty) .
\end{equation*}
In particular, $E$ is a discrete eigenvalue of $H$. 
\end{Th}
\begin{proof}
The fact that
\begin{equation}\label{first_incl}
\sigma_{\mathrm{ess}} ( H ) \subset [ E + m_\nu , \infty )
\end{equation}
is a consequence of \eqref{3rd_estim} in Lemma \ref{numberestimates3}. Indeed, if $\chi$ belongs to $\mathrm{C}_0^\infty( ( - \infty , E + m_\nu ))$, using that the operator $\check{\Gamma}(j^R)$ defined in the previous section is isometric (see Appendix \ref{sec:standard_def}), we have that
\begin{align*}
\chi ( H ) = \check{\Gamma}(j^R)^* \check{\Gamma}(j^R) \chi ( H ) = \check{\Gamma}(j^R)^* \chi( H^{ \mathrm{ext} } ) \check{\Gamma}(j^R) + o(R^{0}) , \quad R \to \infty.
\end{align*}
Since $N_\infty = \mathds{1} \otimes N$ commutes with $H^{ \mathrm{ext} }$ in $\mathscr{H} = \mathscr{H}^{\mathrm{ext}} \otimes  \mathscr{H}^{\mathrm{ext}}$, and since
\begin{equation*}
H^{ \mathrm{ext} } \mathds{1}_{ [ 1 , \infty ) }( N_\infty ) \ge ( E + m_\nu ) \mathds{1}_{ [ 1 , \infty ) }( N_\infty ),
\end{equation*}
this yields
\begin{align*}
\chi ( H ) & = \check{\Gamma}(j^R)^* ( \mathds{1} \otimes \Pi_\Omega ) \chi ( H^{ \mathrm{ext} } ) \check{\Gamma}(j^R) + o(R^0) = \check{\Gamma}(j^R)^* ( \mathds{1} \otimes \Pi_\Omega ) \check{\Gamma}(j^R) \chi( H ) + o(R^0) , \quad R \to \infty ,
\end{align*}
where $\Pi_\Omega$ denotes the projection onto the vacuum in $\mathscr{H}$. The second equality in the previous equation is another consequence of  \eqref{3rd_estim}. The inclusion \eqref{first_incl} then follows from the observation that the operator $\check{\Gamma}(j^R)^* ( \mathds{1} \otimes \Pi_\Omega ) \check{\Gamma}(j^R) \chi( H )$ is compact (see e.g. \cite{DeGe99_01} for a similar argument).

The converse inclusion can be proven by constructing a Weyl sequence associated to $\lambda$ for any $\lambda \in [ E + m_\nu , \infty )$ in the same way as in \cite[Theorem 4.1]{DeGe99_01} or \cite[Theorem 4.3]{Am04_01}.  
\end{proof}
\subsubsection{Spectral analysis for any value of the coupling constant}\label{subsubsec:Mourre1}

In this section, we study the structure of the essential spectrum of $H$ using Mourre's conjugate operator theory \cite{Mo80_01,AmBoGe96_01}. We begin by recalling the main facts on which we will rely. We refer the reader to \cite{AmBoGe96_01} for more details.

Given two self-adjoint operators $A$, $B$ on a Hilbert space $\mathscr{H}$, $B$ is said to be of class $\mathrm{C}^n(A)$ if and only if the map
\begin{equation*}
s \mapsto e^{- \mathrm{i} s A} ( B - z )^{-1} e^{\mathrm{i} s A} \Phi,
\end{equation*}
is of class $\mathrm{C}^n(\mathbb{R})$ for all $\Phi \in \mathscr{H}$. One says that $B$ satisfies a Mourre estimate with respect to $A$ (with compact remainder) on an interval $[a,b]\subset\mathbb{R}$ if there exist a positive constant $\mathrm{c}_0$ and a compact operator $K$ such that
\begin{equation}
\mathds{1}_{[a,b]}(B)[B,\mathrm{i}A]\mathds{1}_{[a,b]}(B) \ge \mathrm{c}_0\mathds{1}_{[a,b]}(B) + K. \label{recall_Mourre}
\end{equation}
If $B \in \mathrm{C}^1( A )$ and $B$ satisfies a Mourre estimate w.r.t. $A$ on $[a,b]$, then $B$ has only finitely many eigenvalues, with finite multiplicities, in $[a,b]$. If the Mourre estimate is strict (i.e. if $K=0$ in \eqref{recall_Mourre}), then $B$ does not have eigenvalues in $[a,b]$. Moreover, one can prove that if $B$ satisfies a Mourre estimate with compact remainder on $[a,b]$, then, on any interval $[c,d]\subset[a,b]$ such that $[c,d]\cap\sigma_{\mathrm{pp}}(B)=\emptyset$, $B$ satisfies a strict Mourre estimate. If $B \in \mathrm{C}^2(A)$ and satisfies the Mourre estimate \eqref{recall_Mourre}, then $B$ has no singular continuous spectrum in $[a,b]$, i.e. $\sigma_{\mathrm{sc}}(B)\cap[a,b]=\emptyset$. More precisely, one can prove that $B$ satisfies a limiting absorption principle in any interval $[c,d]\subset[a,b]$ such that $[c,d]\cap\sigma_{\mathrm{pp}}(B)=\emptyset$. As for this last result, we mention that the condition $B\in\mathrm{C}^2(A)$ can be weakened to ``$B\in\mathrm{C}^{1,1}(A)$'', for instance. To simplify the presentation, we will not use such a weaker notion of regularity in this paper.

Recall that the operators $a_{(i),l}$, $i=1,2$, $l=1,2,3$, and $a_{(3)}$ have been defined in \eqref{eq:defa12}--\eqref{eq:defa3}. In particular, the operators $a_{(i),l}$ are self-adjoint and their domains are given by $\mathscr{D}(a_{(i),l})=\{ h \in \mathfrak{h}_{1} , a_{(i),l}h \in \mathfrak{h}_{1} \}$, where $a_{(i),l}h$ should be understood in the sense of distributions. Likewise, $a_{(3)}$ is self-adjoint with domain $\mathscr{D}(a_{(3)})=\{ h \in \mathfrak{h}_{2} , a_{(3)} h \in \mathfrak{h}_{2} \}$. Using the notation 
\begin{equation*}
\mathrm{d}\Gamma( q ) = \sum_{i=1}^{14} \mathrm{d}\Gamma( q_i )
\end{equation*}
as operators on $\mathscr{H}$, with $q = ( q_1 , \dots , q_{14} )$, $q_1, q_2$ operators on $\mathfrak{h}_2$ and $(q_3, \dots , q_{14})$ operators on $\mathfrak{h}_1$ (see Appendix \ref{sec:standard_def} for details regarding this notation and recall that the total Hilbert space $\mathscr{H}$ is the tensor product of $2$ symmetric Fock spaces for bosons and $12$ anti-symmetric Fock spaces for fermions), we set
\begin{align}
& A:=\mathrm{d}\Gamma(a) , \notag \\
& a:= (a_{(3)}, a_{(3)}, a_{(1),1}, a_{(1),1} , a_{(2),1} , a_{(2),1} , a_{(1),2} , a_{(1),2} , a_{(2),2} , a_{(2),2} , a_{(1),3} , a_{(1),3} , a_{(2),3} , a_{(2),3} ). \label{eq:conjugate_1}
\end{align}
Hence the following notations will be used: $a_1 = a_{(3)}$, $a_2=a_{(3)}$, $a_3=a_{(1),1}$, \dots, $a_{14} = a_{(2),3}$.

Now, a direct computation gives
\begin{equation}
[H_0,\mathrm{i} A] =  \mathrm{d}\Gamma(|\nabla\omega|^2),  \label{commut_H0}
\end{equation}
in the sense of quadratic forms, with (see \eqref{disprela})  $$\nabla\omega := (\nabla\omega^{(3)} , \nabla\omega^{(3)} , \nabla\omega^{(1)}_1 , \nabla\omega^{(1)}_1 , \nabla\omega^{(2)}_1 , \nabla\omega^{(2)}_1 , \nabla\omega^{(1)}_2 , \nabla\omega^{(1)}_2 , \nabla\omega^{(2)}_2 , \nabla\omega^{(2)}_2 , \nabla\omega^{(1)}_3 , \nabla\omega^{(1)}_3 , \nabla\omega^{(2)}_3 , \nabla\omega^{(2)}_3 ).$$ 
Moreover, writing $H_I=H_I(G)$, we have that
 \begin{align}
[H_I(G),\mathrm{i}A] = [H^{(1)}_I(G)+H^{(2)}_I(G)+H^{(3)}_I(G)+H^{(4)}_I(G),\mathrm{i}A], \label{commut_HI}
\end{align}
where
\begin{align}
& [H_I^{(1)}(G) , \mathrm{i}A ] =  \sum_{i=1}^{2} \mathrm{i} H_I^{(1)}( \mathrm{i} a_i G ) - \sum_{i=3}^{14} \mathrm{i} H_I^{(1)}( \mathrm{i} a_i G ) , \label{commut_11} \\
& [H_I^{(2)}(G) , \mathrm{i}A ] = -\sum_{i=1}^{14} \mathrm{i} H_I^{(2)}( \mathrm{i} a_i G ), \label{commut_21} \\
& [H_I^{(3)}(G) , \mathrm{i}A ] = -\sum_{i\in\{1,2,5,6,9,10,13,14\}}^{14} \mathrm{i} H_I^{(3)}( \mathrm{i} a_i G ) + \sum_{i\in\{3,4,7,8,11,12\}}^{14} \mathrm{i} H_I^{(3)}( \mathrm{i} a_i G ), \label{commut_31} \\
& [H_I^{(4)}(G) , \mathrm{i}A ] =- \sum_{i\in\{5,6,9,10,13,14\}}^{14} \mathrm{i} H_I^{(4)}( \mathrm{i} a_i G ) + \sum_{i\in\{1,2,3,4,7,8,11,12\}}^{14} \mathrm{i} H_I^{(4)}( \mathrm{i} a_i G ). \label{commut_41}
\end{align}

The main result of this section is Theorem \ref{MourreI} below, which proves that $H$ satisfies a Mourre estimate with respect to $A$ in any interval that does not intersect the set of thresholds (see \eqref{def_threshold} below for the definition of thresholds in our context). As recalled at the beginning of this section, in order to be able to deduce useful spectral properties of $H$, in addition to the Mourre estimate, one needs to establish that $H$ is regular enough w.r.t. $A$. This is the purpose of the following lemma.

Recall that the notation ``$a_{(i),\cdot} G \in L^2$'' has been introduce above the statement of Theorem \ref{thm:main}.

%
%
\begin{lem}
\label{C1}
Suppose that the masses of the neutrinos $m_{\nu_e}$, $m_{\nu_{\mu}}$, $m_{\nu_{\tau}}$ are positive and consider the Hamiltonian \eqref{total_hamiltonian_g} with $H_I$ given by \eqref{interaction_term}. Assume that
\begin{equation}
G \in L^2, \quad a_{(i),\cdot} G \in L^2, \quad i = 1,2,3. \label{eq:hypclassC1}
\end{equation}
Then $H$ is of class $\mathrm{C}^{1}(A)$. If in addition
\begin{equation}
a_{(i),\cdot} a_{(i'),\cdot} G \in L^2, \quad i,i' = 1,2,3, \label{eq:hypclassC2}
\end{equation}
then $H$ is of class $\mathrm{C}^2( A )$.
\end{lem}
\begin{proof}
Proceeding as in \cite[Section 4]{FrGrSi08_01}, using that $\mathscr{D}( H ) = \mathscr{D}( H_0 )$, it is not difficult to verify that, for all $s \in \mathbb{R}$, $e^{-\mathrm{i}sA}\mathscr{D}(H)\subset\mathscr{D}(H)$. In order to prove that $H$ is of class $\mathrm{C}^1( A )$ if \eqref{eq:hypclassC1} holds, it then suffices to verify that $[ H , \mathrm{i} A ]$ extends to an $H$-bounded operator. From \eqref{commut_H0} and the fact that $|\nabla \omega|^2$ is bounded, it follows that $[H_0,\mathrm{i}A]$ is relatively $N$-bounded, and therefore relatively $H_0$-bounded since the masses of all the particles are positive. Moreover, it follows from Lemma \ref{lemSCW_1} and \eqref{commut_HI}--\eqref{commut_41} that $[H_I(G),\mathrm{i}A]$ is also $H_0$-bounded under assumption \eqref{eq:hypclassC1}.

Likewise, to prove that $H$ is of class $\mathrm{C}^2( A )$ if assumption \eqref{eq:hypclassC2} holds, it suffices to verify that the second commutator $[ [ H , \mathrm{i} A ] , \mathrm{i} A ]$ extends to a relatively $H$-bounded operator. The result then follows from computing $[ [ H , \mathrm{i} A ] , \mathrm{i} A ]$ in the same way as in \eqref{commut_H0}--\eqref{commut_41} and using the same arguments as before. Details are left to the reader.
\end{proof}
The set of thresholds is defined by
\begin{equation}\label{def_threshold}
\tau = \sigma_{\mathrm{pp}}(H) + \Big \{ \sum^{14}_{i=1} m_i  n_i, \, n_i \in \mathbb{N}, \, (n_1,\dots,n_{14}) \neq (0,\dots,0) \Big \},
\end{equation}
where $m_i$ denotes the mass of the particle $i$ (i.e. $m_1=m_2=m_W$, $m_3=m_4=m_e$, $m_5=m_6=m_{\nu_e}$, \emph{etc}).
\begin{Th}\label{MourreI}
Suppose that the masses of the neutrinos $m_{\nu_e}$, $m_{\nu_{\mu}}$, $m_{\nu_{\tau}}$ are positive and consider the Hamiltonian \eqref{total_hamiltonian_g} with $H_I$ given by \eqref{interaction_term}. Assume that \eqref{eq:hypclassC1} holds. Let $\lambda \in \mathbb{R}\setminus \tau$. There exists $\varepsilon>0$, $\mathrm{c}_0>0$ and a compact operator $K$ such that 
\begin{equation}
\mathds{1}_{[\lambda-\varepsilon,\lambda+\varepsilon]}(H) [H,\mathrm{i}A] \mathds{1}_{[\lambda-\varepsilon,\lambda+\varepsilon]}(H) \geq \mathrm{c}_0 \mathds{1}_{[\lambda-\varepsilon,\lambda+\varepsilon]}(H)+K. \label{eq:Mourre_estimate_1}
\end{equation}
In particular, for all interval $[\lambda_1,\lambda_2]$ such that $[\lambda_1,\lambda_2]\cap \tau = \emptyset$, $H$ has at most finitely many eigenvalues with finite multiplicities in $[\lambda_1,\lambda_2]$ and, as a consequence, $\sigma_{\mathrm{pp}}(H)$ can accumulate only at $\tau$, which is a countable set.

If in addition \eqref{eq:hypclassC2} holds,  then $\sigma_{\mathrm{sc}}(H)=\emptyset$.
\end{Th}
The proof of \eqref{eq:Mourre_estimate_1} uses arguments developed in \cite{DeGe99_01} and \cite{Am04_01} (see also \cite{DeGe97_01}). In Theorem \ref{MourreII_1} below, we will give a complete proof of an analogous Mourre estimate in the more difficult case where neutrinos are supposed to be massless. It is not difficult to adapt the proof of Theorem \ref{MourreII_1} to obtain \eqref{eq:Mourre_estimate_1}, under the hypotheses of Theorem \ref{MourreI}. We do not give the details. The fact that \eqref{eq:Mourre_estimate_1} implies that $H$ has at most finitely many eigenvalues with finite multiplicities in any compact interval disjoint from $\tau$ is a consequence of Lemma \ref{C1} together with the abstract results of Mourre's theory recalled at the beginning of this section. The same holds for the absence of singular continuous spectrum assuming \eqref{eq:hypclassC2}.

\subsubsection{Spectral analysis for small coupling constant and regularized kernels}\label{subsubsec:massive_smallg}

In this section, we improve the results of Theorem \ref{MourreI} by imposing stronger conditions on the kernels and treating the coupling constant $g$ as a small parameter. The main idea consists in considering a different, non-self-adjoint conjugate operator, in order to obtain a global Mourre estimate without compact remainder. Extensions of the original  Mourre's theory \cite{Mo80_01} to settings with non-self-adjoint conjugate operators have been considered by different authors (see, in particular, \cite{HuSp95_01,AmBoGe96_01,GeGeMo04_01,FaMoSk11_01,BaFaGu16_01}). Yet a further extension, sometimes called singular Mourre's theory, concerns the case where the commutator of the Hamiltonian with the conjugate operator is not relatively bounded w.r.t. the Hamiltonian itself (see \cite{GeGeMo04_01,GeGeMo04_02,FaMoSk11_01}). Singular Mourre's theory will not be used in this section, but it will turn out to be an important tool in order to treat the situation where neutrinos are supposed to be massless, see Section \ref{subsec:spectral_massless}.

Before stating the results, we briefly recall the definitions and main elements that will be used in the context of Mourre's theory with a non-self-adjoint conjugate operator. For more details, we refer the reader to \cite[Appendix B]{BaFaGu16_01} for a short presentation, and to \cite{GeGeMo04_01} for a detailed and complete theory.

Let $H$ be a self-adjoint operator on a Hilbert space $\mathscr{H}$ and let $A$ be a closed and maximal symmetric operator on $\mathscr{H}$. Assuming that $\mathrm{dim}( \mathrm{Ker}( A^* - \mathrm{i} ) )= 0$, it is known that $A$ generates a semigroup of isometries $\{ W_t \}_{ t \ge 0 }$. Let $\mathscr{G} := \mathscr{D}( | H |^{\frac{1}{2}} )$ be equipped with the norm defined by $\| \phi \|^2_{ \mathscr{G} } :=  \| |H|^{\frac{1}{2}} \phi \|^2 + \| \phi \|^2$, for all $\phi \in \mathscr{G}$, and let $\| \phi \|^2_{ \mathscr{G}^* } := \| ( |H| + 1 )^{-\frac{1}{2}} \phi \big \|^2$. Observe that the dual space $\mathscr{G}^*$ of $\mathscr{G}$ identifies with the completion of $\mathscr{H}$ with respect to the norm $\| \cdot \|_{\mathscr{G}^*}$. In particular, $H$ identifies with an element of $\mathscr{L} ( \mathscr{G} ; \mathscr{G}^* )$. Suppose that, for all $t > 0$, $W_t$ and $W_t^*$ preserve $\mathscr{G}$ and that, for all $\phi \in \mathscr{G}$,
\begin{equation*}
\sup_{0<t<1} \| W_t \phi \|_\mathscr{G} < \infty, \quad \sup_{0<t<1} \| W_t^* \phi \|_\mathscr{G} < \infty.
\end{equation*}
This condition implies, in particular, that the restriction of $W_t$ to $\mathscr{G}$ is a $\mathrm{C}_0$-semigroup and that $W_t$ extends to a $\mathrm{C}_0$-semigroup on $\mathscr{G}^*$ which is denoted by the same symbol. One says that $H$ is of class $\mathrm{C}^1( A ; \mathcal{G} ; \mathcal{G}^* )$ if there exists a positive constant $\mathrm{c}$ such that, for all $0 \le t \le 1$,
\begin{equation}
\| W_{t} H - H W_{t} \|_{ \mathscr{L} ( \mathscr{G} ; \mathscr{G}^* ) } \le \mathrm{c} t . \label{def_C1GG*}
\end{equation}
In this case there is an operator $H' \in \mathscr{L} ( \mathscr{G} ; \mathscr{G}^* )$ such that, for all $\phi \in \mathscr{D}( H )$,
\begin{equation}\label{defH'}
\lim_{ t \to 0^+ } \frac{1}{t} \big ( \langle \phi , W_t H \phi \rangle - \langle H \varphi , W_t \phi \rangle \big ) = \langle \phi , H' \phi \rangle ,
\end{equation}
and one can verify that, in a suitable sense, $H'$ identifies with the quadratic form $[ H , \mathrm{i} A ]$. One says that $H \in \mathrm{C}^2( A ; \mathscr{G} ; \mathscr{G}^* )$ if $H$ belongs to $\mathrm{C}^1( A ; \mathscr{G} ; \mathscr{G}^* )$ and $H'$ belongs to $\mathrm{C}^1( A ; \mathscr{G} ; \mathscr{G}^* )$. If $H$ is of class $\mathrm{C}^1( A ; \mathscr{G} ; \mathscr{G}^* )$, one says that $H$ satisfies a (strict) Mourre estimate on an open interval $I$ if there exist constants $\mathrm{c}_0>0$ and $\mathrm{c} \in \mathbb{R}$, such that, in the sense of quadratic forms on $\mathscr{D}( H )$,
\begin{equation*}
 H' \ge \mathrm{c}_0 \mathds{1} - \mathrm{c} ( \mathds{1} - \mathds{1}_I(H) ) ( \mathds{1} + H^2 )^{\frac12}.
\end{equation*}
As in the setting of Section \ref{subsubsec:Mourre1}, if $H$ is of class $\mathrm{C}^1( A ; \mathscr{G} ; \mathscr{G}^* )$ and satisfies a strict Mourre estimate on an interval $I$, then $H$ does not have eigenvalues in $I$. If $H$ is of class $\mathrm{C}^2( A ; \mathscr{G} ; \mathscr{G}^* )$, then a limiting absorption principle holds in any interval where $H$ satisfies a strict Mourre estimate. In particular $H$ does not have singular continuous spectrum in such an interval.

Recall that the operators $b_{(i),l}$, $i=1,2$, $l=1,2,3$, acting on $\mathfrak{h}_1$, and $b_{(3)}$ acting on $\mathfrak{h}_2$ have been defined in \eqref{eq:defb12}--\eqref{eq:defb3}.
In particular, the operators $b_{(i),l}$ with domains $\mathrm{C}_0^{\infty}(\mathbb{R}^3 \setminus \{ 0 \}) \times \{-\frac12,\frac12\}$ are symmetric, and likewise $b_{(3)}$ with domain $\mathrm{C}_0^{\infty}(\mathbb{R}^3 \setminus \{ 0 \}) \times \{-1,0,1 \}$ is symmetric. Their closures are denoted by the same symbols. Using the notations of Appendix \ref{sec:standard_def} as before, we set
\begin{align}\label{conjugate_2}
&B:=\mathrm{d}\Gamma(b) , \notag \\
& b:= (b_{(3)}, b_{(3)}, b_{(1),1}, b_{(1),1} , b_{(2),1} , b_{(2),1} , b_{(1),2} , b_{(1),2} , b_{(2),2} , b_{(2),2} , b_{(1),3} , b_{(1),3} , b_{(2),3} , b_{(2),3} ) ,
\end{align}
(hence $b_1 = b_{(3)}$, $b_2=b_{(3)}$, $b_3=b_{(1),1}$, \dots, $b_{14} = b_{(2),3}$.). Then $\mathrm{dim}( \mathrm{Ker}( B^* - \mathrm{i} ) )= 0$, $B$ generates a $\mathrm{C}_0$-semigroup of isometries $\{W_t\}_{t\ge0}$ and a direct computation gives
\begin{equation}
[H_0,\mathrm{i} B] = N,  \label{commut_H0_2}
\end{equation}
in the sense of quadratic forms. Moreover, the commutators $[H_I^{(j)}(G) , \mathrm{i}B]$, $j=1,\dots,4$ are given by \eqref{commut_11}--\eqref{commut_41} with $b_i$ instead of $a_i$.

A straightforward modification of \cite[Section 5.1]{BaFaGu16_01} shows that $W_t$ and $W^*_t$ preserve $\mathscr{D}(|H|^{\frac{1}{2}})$, and that for all $\phi \in \mathscr{D}(|H|^{\frac{1}{2}})$,
\begin{align}
\sup_{0<t<1} \|W_t \phi\|_{ \mathscr{D}(|H|^{\frac{1}{2}})} < \infty, \quad \sup_{0<t<1} \|W^*_t \phi\|_{\mathscr{D}(|H|^{\frac{1}{2}})} < \infty. \label{eq:bpreserve}
\end{align}
Before proving a Mourre estimate and deducing from it spectral properties of $H$, we must show, as in the previous section, that $H$ is regular enough with respect to the conjugate operator.

%
%
\begin{lem}
\label{lem:C2}
Suppose that the masses of the neutrinos $m_{\nu_e}$, $m_{\nu_{\mu}}$, $m_{\nu_{\tau}}$ are positive and consider the Hamiltonian \eqref{total_hamiltonian_g} with $H_I$ given by \eqref{interaction_term}. Assume that
\begin{equation}
G \in L^2, \quad b_{(i),\cdot} G \in L^2, \quad i = 1,2,3. \label{eq:hypclassC1_3}
\end{equation}
Then $H$ is of class $\mathrm{C}^{1}(B ; \mathscr{D}(|H|^{\frac{1}{2}}) ; \mathscr{D}(|H|^{\frac{1}{2}})^*)$. If in addition
\begin{equation}
b_{(i),\cdot} b_{(i'),\cdot} G \in L^2, \quad i,i' = 1,2,3, \label{eq:hypclassC2_4}
\end{equation}
then $H$ is of class $\mathrm{C}^2( B ; \mathscr{D}(|H|^{\frac{1}{2}}) ; \mathscr{D}(|H|^{\frac{1}{2}})^* )$.
\end{lem}
\begin{proof}
Suppose that \eqref{eq:hypclassC1_3} holds. In order to verify that $H$ is of class $\mathrm{C}^{1}(B ; \mathscr{D}(|H|^{\frac{1}{2}}) ; \mathscr{D}(|H|^{\frac{1}{2}})^*)$, since \eqref{eq:bpreserve} holds, it suffices to prove (see \cite[Proposition 5.2]{GeGeMo04_01}) that the quadratic form $[H,B]$ defined on $\mathscr{D}( |H|^{\frac12} )\cap\mathscr{D}(B)$ by $\langle \phi, [H,B] \phi \rangle = \langle \phi , H B \phi \rangle - \langle B^* \phi , H \phi \rangle$ extends to an element of $\mathscr{L} ( \mathscr{D}(|H|^{\frac{1}{2}}) ; \mathscr{D}(|H|^{\frac{1}{2}})^* )$. By \eqref{commut_H0_2} and the fact that $N$ is relatively $H$-bounded, it is clear that $[H_0,B]$ extends to an element of $\mathscr{L} ( \mathscr{D}(|H|^{\frac{1}{2}}) ; \mathscr{D}(|H|^{\frac{1}{2}})^* )$. That $[H_I,B]$ also extends to an element of $\mathscr{L} ( \mathscr{D}(|H|^{\frac{1}{2}}) ; \mathscr{D}(|H|^{\frac{1}{2}})^* )$ follows from the expression of the commutator (given by \eqref{commut_11}--\eqref{commut_41} with $b_{i}$ instead of $a_{i}$, as mentioned above) together with Lemma \ref{lemSCW_1}, which can be applied under the hypothesis \eqref{eq:hypclassC1_3}.

To prove that $H$ is of class $\mathrm{C}^2( B ; \mathscr{D}(|H|^{\frac{1}{2}}) ; \mathscr{D}(|H|^{\frac{1}{2}})^* )$ under the further assumption \eqref{eq:hypclassC2_4}, it suffices to verify similarly that $[[H,B],B]$ extends to an element of $\mathscr{L} ( \mathscr{D}(|H|^{\frac{1}{2}}) ; \mathscr{D}(|H|^{\frac{1}{2}})^* )$. But we have that $[[H_0,B],B]=0$ and a similar computation as before shows that $[[H_I,B],B]$ extends to an element of $\mathscr{L} ( \mathscr{D}(|H|^{\frac{1}{2}}) ; \mathscr{D}(|H|^{\frac{1}{2}})^* )$ by Lemma \ref{lemSCW_1}.
\end{proof}

The next theorem establishes a global Mourre estimate for $H$, from which, using the regularity properties proven in the previous lemma, we can deduce the desired spectral properties of $H$.
\begin{Th}
\label{MourreII}
Consider the Hamiltonian \eqref{total_hamiltonian_g} with $H_I$ given by \eqref{interaction_term} and assume that \eqref{eq:hypclassC1_3} holds. There exist $g_0>0$, $\mathrm{c}>0$ and $\mathrm{d}>0$ such that, for all values of the masses of the neutrinos $m_{\nu_e}>0$, $m_{\nu_{\mu}}>0$, $m_{\nu_{\tau}}>0$,
\begin{equation}
\label{eq:Mourreunif_B2}
[H,\mathrm{i}B]\ge \mathrm{c}\mathds{1}-\mathrm{d} \Pi_\Omega ,
\end{equation}
where $\Pi_\Omega$ denotes the projection onto the vacuum in $\mathscr{H}$. In particular, $E = \inf \sigma ( H )$ is the only eigenvalue of $H$, and $E$ is non-degenerate. If, in addition, \eqref{eq:hypclassC2_4} holds,  then the spectrum of $H$ in $[E+m_\nu,\infty)$ is purely absolutely continuous.
\end{Th}
\begin{proof}
Recall that, in the sense of quadratic forms on $\mathscr{D}(B) \cap \mathscr{D}(H_0)$ we have that $[H_0,\mathrm{i}B]=N$ and that $[H_I,\mathrm{i}B]$ is relatively $N$-bounded by Lemma \ref{lemSCW_1}. Therefore there exists $\mathrm{c}_1>0$ and $\mathrm{c}_2>0$, which do not depend on the masses of the neutrinos, such that
\begin{equation*}
\langle \psi , [H_I, \mathrm{i} B] \psi \rangle \leq \mathrm{c}_1 \langle \psi , N \psi \rangle + \mathrm{c}_2 \|\psi\|^2.
\end{equation*}
This yields
\begin{align*}
[H, \mathrm{i}B]  = [H_0,\mathrm{i}B] + g [H_I, \mathrm{i}B] & \geq N - \mathrm{c}_1 g N - \mathrm{c}_2 g  \\
& \geq (1 - \mathrm{c}_1 g)( N + \Pi_{\Omega} - \Pi_{\Omega} ) - \mathrm{c}_2 g  \\
& \geq (1 - \mathrm{c}_1 g) \mathds{1} - (1 - \mathrm{c}_1 g)\Pi_{\Omega}  - \mathrm{c}_2 g   \\
& \geq (1 - (\mathrm{c}_1+\mathrm{c}_2) g ) \mathds{1} - (1 - \mathrm{c}_1 g)\Pi_{\Omega}  ,
\end{align*}
which proves \eqref{eq:Mourreunif_B2}.

The fact that \eqref{eq:Mourreunif_B2} implies that $H$ has at most one eigenvalue is a consequence of the virial theorem (see Lemma 10 of \cite{HuSp95_02}), which holds since $H \in \mathrm{C}^{1}(B ; \mathscr{D}(|H|^{\frac{1}{2}}) ; \mathscr{D}(|H|^{\frac{1}{2}})^*)$ by Lemma \ref{lem:C2}, together with the fact that $\mathrm{dim}( \mathrm{Ran} ( \Pi_\Omega ) ) = 1$. See, e.g., \cite[Lemma 10]{HuSp95_01}.

That the spectrum of $H$ in $[E+m_\nu,\infty)$ is purely absolutely continuous if \eqref{eq:hypclassC2_4} is satisfied is a consequence of the abstract results recalled at the beginning of this section together with the fact that $H \in \mathrm{C}^{2}(B ; \mathscr{D}(|H|^{\frac{1}{2}}) ; \mathscr{D}(|H|^{\frac{1}{2}})^*)$, by Lemma \ref{lem:C2}.
\end{proof}
\subsection{Massless neutrinos}\label{subsec:spectral_massless}

In this section, we suppose that
\begin{equation*}
m_{\nu_e} = m_{\nu_\mu} = m_{\nu_\tau} = 0.
\end{equation*}
As in the previous section, we first recall suitable assumptions implying existence of a ground state for $H$, next we study the structure of the essential spectrum using a suitable version of Mourre's theory.

\subsubsection{Existence of a ground state and location of the essential spectrum}

We recall here, without proof, a result due to \cite{BaGu09_01,AsBaFaGu11_01}.
\begin{Th}
\label{Thess}
Suppose that the masses of the neutrinos $m_{\nu_e}$, $m_{\nu_{\mu}}$, $m_{\nu_{\tau}}$ vanish and consider the Hamiltonian \eqref{total_hamiltonian_g} with $H_I$ given by \eqref{interaction_term}. Assume that $G \in L^2$ and that $|p_2|^{-1} G \in L^2$. Then, there exists $g_0 > 0$ such that, for all $|g| \le g_0$, $H$ has a unique ground state, i.e., $E=\inf \sigma( H )$ is a non-degenerate eigenvalue of $H$. Moreover,
\begin{equation*}
\sigma(H)=\sigma_{\mathrm{ac}}(H)=[E, \infty).
\end{equation*}
\end{Th}
Let us mention that the conclusion of Theorem \ref{Thess} does not exclude the presence of eigenvalues or singular continuous spectrum in the interval $[E,\infty)$. Proving absence of embedded eigenvalues and of singular continuous spectrum will be one of the main purposes of the next section. Theorem \ref{Thess} is proven in \cite{BaGu09_01,AsBaFaGu11_01} for $H = H_0 + g ( H_I^{(1)} + H_I^{(2)} )$, but the proof goes through without any substantial modification for the full Hamiltonian $H = H_0 + g ( H_I^{(1)} + H_I^{(2)} + H_I^{(3)} + H_I^{(4)} )$.

\subsubsection{Spectral analysis}

Now, we turn to the study of the essential spectrum $[E,\infty)$ of $H$. Let us mention that, in \cite{AsBaFaGu11_01}, it is proven that the spectrum of $H$ in $[E,m_e)$ is purely absolutely continuous (except for the ground state energy, which is an eigenvalue as recalled in Theorem \ref{Thess}). This result is proven using Mourre's theory, with a conjugate operator given as the generator of dilatations ``restricted to low-energies''. The idea of employing such a conjugate operator originated in \cite{FrGrSi08_01}, see also \cite{ChFaFrSi12_01}. The method of \cite{FrGrSi08_01} is particularly efficient in that it only requires that the kernels of the interaction Hamiltonian belong to the domain of the generator of dilatations, which does not require much regularity of the kernels in the low-energy regime. However, it is presently not known how to extend this approach to prove the absence of singular continuous spectrum in the whole interval $[E,\infty)$ and not only in $[E,m_e)$.

On the other hand, if one assumes that the kernels belong to the domain of the operator $b_{i}$ in \eqref{eq:defb12}--\eqref{eq:defb3}, one can proceed as in Section \ref{subsubsec:massive_smallg}. The conjugate operator is still given by \eqref{conjugate_2}. Note that $b_{(2),l} = a_{(2),l}$ when the masses of the neutrinos vanish. As in Section \ref{subsubsec:massive_smallg}, one verifies that $B$ is the generator of a $\mathrm{C}_0$-semigroup of isometries and that $[H_0,\mathrm{i}B]=N$. In particular, $[H_0,\mathrm{i}B]$ is not relatively $H$-bounded anymore and therefore one cannot apply the abstract setting recalled at the beginning of Section \ref{subsubsec:massive_smallg}. Nevertheless, one can use singular Mourre's theory developed in \cite{GeGeMo04_01} (see also \cite{FaMoSk11_01}). Before stating our results, we briefly explain the setting of this theory, focusing on the differences with that of Section \ref{subsubsec:massive_smallg}.

Let $\mathscr{H}$ be a complex Hilbert space and consider two self-adjoint operators $H$ and $M$ on $\mathscr{H}$, with $M \ge 0$, and a symmetric operator $R$ relatively $H$-bounded. We suppose that $H \in \mathrm{C}^1( M )$ and that $[ H , \mathrm{i} M ]$ is relatively $H$-bounded. Let $H' := M+R$ on $\mathscr{D}(M) \cap \mathscr{D}(H)$. We set $\mathscr{G} = \mathscr{D}( M^{\frac{1}{2}} )\cap \mathscr{D}( |H|^{\frac{1}{2}} )$, equipped with the norm defined by  $\| \phi \|^2_{ \mathscr{G} } := \| M^{\frac{1}{2}} \phi \|^2_{ \mathscr{H} } + \| |H|^{\frac{1}{2}} \phi \|^2_{ \mathscr{H} } + \| \phi \|^2_{ \mathscr{H} }$, for all $\phi \in \mathscr{G}$. As in Section \ref{subsubsec:massive_smallg}, letting $\| \phi \|_{ \mathscr{G}^* } := \| ( M + | H | + 1 )^{-\frac{1}{2}} \phi \|_{ \mathscr{H} }$,  one verifies that $\mathscr{G}^*$ identifies with the completion of $\mathscr{H}$ with respect to $\| \cdot \|_{ \mathscr{G}^* }$. Hence $H$ and $M$ identify with elements of $\mathscr{L}( \mathscr{G} ; \mathscr{G}^* )$. As in Section \ref{subsubsec:massive_smallg}, the conjugate operator $A$ is supposed to be closed, maximal symmetric and such that $\dim \mathrm{Ker} (A^* - \mathrm{i}) = 0$. The $\mathrm{C}_0$-semigroup of isometries generated by $A$ is denoted by $\{ W_t \}_{t \ge 0}$ and we suppose that, for all $\phi \in \mathscr{G}$,
\begin{equation*}
\sup_{0<t<1} \| W_t \phi \|_\mathscr{G} < \infty, \quad \sup_{0<t<1} \| W_t^* \phi \|_\mathscr{G} < \infty.
\end{equation*}
As before, $H$ is said to belong to $\mathrm{C}^1( A ; \mathscr{G} ; \mathscr{G}^* )$ if \eqref{def_C1GG*} holds. In this case there is an operator $H' \in \mathscr{L} ( \mathscr{G} ; \mathscr{G}^* )$ satisfying \eqref{defH'} and one requires that $H' = M + R$. 

Given these conditions, assuming that $H$ belongs to $\mathrm{C}^1( A ; \mathscr{G} ; \mathscr{G}^* )$ (or to $\mathrm{C}^2( A ; \mathscr{G} ; \mathscr{G}^* )$) and that the Mourre estimate holds, the conclusions concerning the spectrum of $H$ are the same as in Section \ref{subsubsec:massive_smallg}. Therefore, choosing $M = [ H_0 , \mathrm{i} B ] = N$ and $R = [ H_I , \mathrm{i} B ]$ and proceeding exactly as in the proofs of Lemma \ref{lem:C2} and Theorem \ref{MourreII}, we obtain the following result.
\begin{Th}
\label{MourreII_1}
Suppose that the masses of the neutrinos $m_{\nu_e}$, $m_{\nu_{\mu}}$, $m_{\nu_{\tau}}$ vanish and consider the Hamiltonian \eqref{total_hamiltonian_g} with $H_I$ given by \eqref{interaction_term}. Assume that \eqref{eq:hypclassC1_3} holds.  There exist $g_0>0$, $\mathrm{c}>0$ and $\mathrm{d}>0$ such that, 
\begin{equation*}
[H,\mathrm{i}B]\ge \mathrm{c}\mathds{1}-\mathrm{d} \Pi_\Omega ,
\end{equation*}
where $\Pi_\Omega$ denotes the projection onto the vacuum in $\mathscr{H}$. In particular, $H$ has at most one eigenvalue, which is non-degenerate. If, in addition, \eqref{eq:hypclassC2_4} holds,  then, except for the ground state energy $E$, the spectrum of $H$ in $[E,\infty)$ is purely absolutely continuous.
\end{Th}
Theorem \ref{MourreII_1} provides a complete description of the spectrum of $H$ for small enough values of $g$ and under strong assumptions on the kernels $G$. However, in view of applications to scattering theory in Section \ref{sec:AC} -- more precisely, in order to prove the propagation estimates in Section \ref{sec:propag} that will be subsequently used in Section \ref{sec:AC} -- the conjugate operator $B$ chosen in Theorem \ref{MourreII_1} is too singular. The singularity here comes from the presence of massive particles. Indeed, for massive particles, the operators $b_i$ are strongly singular near the origin $p_i=0$ because of the factor $( p_i \cdot \nabla \omega^{(i)}_l (p_i ) )^{-1}$ in the definitions \eqref{eq:defb12}--\eqref{eq:defb3}.

For this reason, we need to prove a Mourre estimate with another conjugate operator, namely the operator $A$ defined in \eqref{eq:conjugate_1}. Notice that this operator is not self-adjoint when neutrinos are supposed to be massless, because in this case the operators $a_{(2),l}$ are not self-adjoint, only maximal symmetric. The domains of the operators $a_{(2),l}$ are explicitly given as follows: let $T:L^2(\mathbb{R}^3 \times \{-\frac12 , \frac12\}) \to L^2( \mathbb{R}_+ ) \otimes L^2( \mathrm{S}^2 ) \otimes \mathbb{C}^2$ be the unitary operator, going from cartesian coordinate to polar coordinate,  defined by $(Tu)(r,\theta,s) := ru(r\theta,s)$. Then
\begin{equation}\label{eq:domain_conjugate}
\mathscr{D}( a_{(2),l} ) = T^{-1} \big ( \mathbb{H}_0^1( \mathbb{R}_+ ) \otimes L^2( \mathrm{S}^2 ) \otimes \mathbb{C}^2 \big ) , \quad a_{(2),l} = T^{-1} ( \mathrm{i} \partial_r ) T ,
\end{equation}
where $\mathbb{H}_0^1( \mathbb{R}_+ )$ is the usual Sobolev space with Dirichlet boundary condition at $0$.

As in Section \ref{subsubsec:massive_smallg}, one verifies that $\mathrm{dim}( \mathrm{Ker}( A^* - \mathrm{i} ) )= 0$, and that $A$ generates a $\mathrm{C}_0$-semigroup of isometries $\{ \tilde{W}_t \}_{ t \ge 0 }$ such that $\tilde{W}_t$ and $\tilde{W}^*_t$ preserve $\mathscr{G} := \mathscr{D}(|H|^{\frac{1}{2}}) \cap \mathscr{D}( N_{\mathrm{neut}}^{\frac12} )$. Moreover, for all $\phi \in \mathscr{G}$,
\begin{align*}
\sup_{0<t<1} \| \tilde{W}_t \phi\|_{ \mathscr{G} } < \infty, \quad \sup_{0<t<1} \| \tilde{W}^*_t \phi\|_{\mathscr{G} } < \infty . 
\end{align*}
A direct computation gives
\begin{equation}
[H_0,\mathrm{i} A] = N_{\mathrm{neut}} + \sum_{l=1}^3 \mathrm{d} \Gamma( | \nabla \omega^{(1)}_l( p_1 ) |^2 ) + \mathrm{d} \Gamma( | \nabla \omega^{(3)}( p_3 ) |^2 ) ,  \label{commut_H0_8}
\end{equation}
in the sense of quadratic forms, where
\begin{align*}
& \mathrm{d} \Gamma( | \nabla \omega^{(1)}_l( p_1 ) |^2 ) = \sum_{\epsilon=\pm} \int | \nabla \omega^{(1)}_l( p_1 ) |^2 {b^*_{l,\epsilon}}(\xi_1){b_{l,\epsilon}}(\xi_1) d\xi_1 , \\
& \mathrm{d} \Gamma( | \nabla \omega^{(3)}( p_3 ) |^2 ) =  \sum_{\epsilon=\pm} \int | \nabla \omega^{(3)}( p_3 ) |^2 {a^*_{\epsilon}}(\xi_3){a_{\epsilon}}(\xi_3) d\xi_3 .
\end{align*}
Moreover, the commutators $[H_I^{(j)}(G) , \mathrm{i} A]$, $j=1,\dots,4$, are given by \eqref{commut_11}--\eqref{commut_41}. In particular, the commutator $[H,\mathrm{i}A]$ is not relatively $H$-bounded. For this reason, we work in the setting of singular Mourre's theory.  

The following lemma can be proven in the same way as Lemma \ref{lem:C2}.
\begin{lem}\label{lem:C8}
Suppose that the masses of the neutrinos $m_{\nu_e}$, $m_{\nu_{\mu}}$, $m_{\nu_{\tau}}$ vanish and consider the Hamiltonian \eqref{total_hamiltonian_g} with $H_I$ given by \eqref{interaction_term}. Assume that \eqref{eq:hypclassC1} holds. Then $H$ is of class $\mathrm{C}^{1}( A ; \mathscr{G} ; \mathscr{G}^* )$. If in addition \eqref{eq:hypclassC2} holds,  then $H$ is of class $\mathrm{C}^2( A ; \mathscr{G} ; \mathscr{G}^* )$.
\end{lem}
Recall that the set of thresholds, $\tau$, is defined in \eqref{def_threshold}. We are now ready to prove the main result of this section.
\begin{Th}\label{MourreI_3}
Suppose that the masses of the neutrinos $m_{\nu_e}$, $m_{\nu_{\mu}}$, $m_{\nu_{\tau}}$ vanish and consider the Hamiltonian \eqref{total_hamiltonian_g} with $H_I$ given by \eqref{interaction_term}. Assume that
\begin{equation}
G \in L^2, \quad a_{(i),\cdot} G \in L^2, \quad |p_3|^{-1} a_{(i),\cdot} G \in L^2, \quad i = 1,2,3. \label{eq:hypclassC1_3_1_1_1}
\end{equation}
There exists $g_0>0$ such that, for all $|g| \le g_0$ and $\lambda \in \mathbb{R}\setminus \tau$, there exist $\varepsilon>0$, $\mathrm{c}_0>0$, $\mathrm{d}>0$ and a compact operator $K$ such that
\begin{equation}
[H,\mathrm{i} A ] \geq \mathrm{c}_0 ( N_{\mathrm{neut}} + \mathds{1} ) - \mathrm{d} ( \mathds{1} - \mathds{1}_{[\lambda-\varepsilon,\lambda+\varepsilon]}(H) ) ( 1 + H^2 )^{\frac12} + K. \label{eq:Mourre_estimate_1_4}
\end{equation}
In particular, for all interval $[\lambda_1,\lambda_2]$ such that $[\lambda_1,\lambda_2]\cap \tau = \emptyset$, $H$ has at most finitely many eigenvalues with finite multiplicities in $[\lambda_1,\lambda_2]$ and, as a consequence, $\sigma_{\mathrm{pp}}(H)$ can accumulate only at $\tau$, which is a countable set. If in addition
\begin{equation}
a_{(i),\cdot} a_{(i'),\cdot} G \in L^2, \quad i,i' = 1,2,3, \label{eq:hypclassC2_2_4}
\end{equation}
then $\sigma_{\mathrm{sc}}(H)=\emptyset$.
\end{Th}
\begin{Rk}
The following weaker version of the Mourre estimate,
\begin{equation*}
[H,\mathrm{i} A ] \geq \mathrm{c}_0  \mathds{1} - \mathrm{d} ( \mathds{1} - \mathds{1}_{[\lambda-\varepsilon,\lambda+\varepsilon]}(H) ) ( 1 + H^2 )^{\frac12} + K,
\end{equation*}
would be sufficient for the conclusions of Theorem \ref{MourreI_3} to hold. But the Mourre estimate with the operator $N_{\mathrm{neut}}$ in right-hand-side of \eqref{eq:Mourre_estimate_1_4} will be important in our proof of some propagation estimates in Section \ref{sec:propag}.
\end{Rk}
\begin{proof}[Proof of Theorem \ref{MourreI_3}]
In this proof, for any interval $\Delta$, $\chi_{\Delta}$ will refer to a function in $\mathrm{C}_{0}^{\infty}(\mathbb{R})$ such that $\Delta \subset \mathrm{supp}\left(\chi_{\Delta} \right)$. For all $\lambda \in \mathbb{R}$, we set 
\begin{align*}
d(\lambda) = \inf \Big \{ & \mu \in \mathbb{R} \, \text{ such that } \, \mu = \sum_{ \substack{ l=1,2,3 \\ \epsilon =\pm }} \sum_{i=1}^{n_{l,\epsilon}} | \nabla \omega^{(1)}_l( p_{1,i,l,\epsilon} ) |^2 + \sum_{\epsilon=\pm} \sum_{i=1}^{n_{\epsilon}} | \nabla \omega^{(3)}( p_{3,i,\epsilon} ) |^2 ,  \\
& \lambda_1 + \sum_{ \substack{ l=1,2,3 \\ \epsilon =\pm }} \sum_{i=1}^{n_{l,\epsilon}} \omega^{(1)}_l( p_{1,i,l,\epsilon} ) + \sum_{\epsilon=\pm} \sum_{i=1}^{n_{\epsilon}} \omega^{(3)}_l( p_{3,i,\epsilon} )  =\lambda, \, \lambda_1 \in \sigma_{\mathrm{pp}}(H) ,  \\
& n_{l,\epsilon} , n_\epsilon \in \mathbb{N} \text{ and at least one of the } n_{l,\epsilon} \text{ or } n_\epsilon \text{ is} \neq 0, \, p_{1,i,l,\epsilon} , p_{3,i,\epsilon} \in \mathbb{R}^3 \Big \} ,
\end{align*}
with the convention that $\inf \emptyset = 0$.
The definition of $\tilde{d}(\lambda)$ is the same, except that we do not impose the restriction that at least one of the $n_{l,\epsilon}$ or $n_\epsilon$ is $\neq0$. One can then verify that  $\tilde{d}(\lambda) = d(\lambda)$, if $\lambda \notin \sigma_{\mathrm{pp}}(H)$, and $\tilde d(\lambda)=0$ if $\lambda \in \sigma_{\mathrm{pp}}(H)$.
We also introduce, for $\kappa > 0$, 
\begin{equation*}
\Delta^{\kappa}_{\lambda} = [\lambda-\kappa,\lambda+\kappa], \quad d^{\kappa}(\lambda) = \inf_{\mu \in \Delta^{\kappa}_{\lambda}} d(\mu), \quad \tilde{d}^{\kappa}(\lambda) = \inf_{\mu \in \Delta^{\kappa}_{\lambda}} \tilde{d}(\mu).
\end{equation*}

Recalling \eqref{commut_H0_8}, we set
\begin{equation*}
H'_0 := \frac12 N_{\mathrm{neut}} + \sum_{l=1}^3 \mathrm{d} \Gamma( | \nabla \omega^{(1)}_l( p_1 ) |^2 ) + \mathrm{d} \Gamma( | \nabla \omega^{(3)}( p_3 ) |^2 ) ,
\end{equation*}
so that
\begin{equation*}
[H_0,\mathrm{i}A] = \frac12 N_{\mathrm{neut}} + H'_0.
\end{equation*}
We also set $H'_I := [ H_I , \mathrm{i} A ]$.

We follow the general strategy of the proof of \cite[Theorem 4.3]{DeGe99_01}. Let $m_1 := \inf ( m_e , m_\mu , m_\tau , m_W )>0$. We will prove by induction that the following properties hold for any $n \in \mathbb{N}^*$.  
\begin{itemize}
\item[$H_1(n):$] Let $\varepsilon>0$ and $\lambda \in [E,E+n m_1 )$. There exist a constant $\mathrm{c}$, a compact operator $K_0$ and an interval $\Delta$ containing $\lambda$ such that
\begin{equation*}
H'_0 + g H'_I \ge \Big (\min \big ( \frac14 , d(\lambda) \big ) - \varepsilon \Big ) \mathds{1} - \mathrm{c} \big ( \mathds{1} - \chi_{\Delta} (H) \big ) ( 1 + H^2 )^{\frac12} + K_0.
\end{equation*}

\item[$H_2(n):$] Let $\varepsilon > 0 $ and $\lambda \in [E,E+n m_1 )$. There exist a constant $\mathrm{c}$ and an interval $\Delta$ containing $\lambda$ such that
\begin{equation*}
H'_0 + g H'_I \ge \Big (\min \big ( \frac14 , \tilde{d}(\lambda) \big ) -\varepsilon \Big )\mathds{1} - \mathrm{c} \big ( \mathds{1} - \chi_{\Delta} (H) \big ) ( 1 + H^2 )^{\frac12}.
\end{equation*}

\item[$H_3(n):$] Let $\kappa > 0$, $\varepsilon_0>0$ and $\varepsilon > 0$. There exist a constant $\mathrm{c}$ and $\delta>0$ such that, for all $\lambda  \in [E,E+n m_1 - \varepsilon_0] $, one has
\begin{equation*}
H'_0 + g H'_I \ge \Big (\min \big ( \frac14 , \tilde{d}^{\kappa}(\lambda) \big ) -\varepsilon \Big ) \mathds{1} - \mathrm{c} \big ( \mathds{1} - \chi_{\Delta^{\delta}_{\lambda}} (H) \big )( 1 + H^2 )^{\frac12}.
\end{equation*}

\item[$S_1(n):$] $\tau$ is a closed countable set in $[E,E+n m_1 ]$.

\item[$S_2(n):$] for all $\lambda_1$, $\lambda_2$ such that $\lambda_1 < \lambda_2 \leq E + n m_1 $ and $[\lambda_1,\lambda_2]\cap \tau = \emptyset$, $H$ has finitely many eigenvalues, with finite multiplicities, in $[\lambda_1,\lambda_2]$. 
\end{itemize}
We claim that, for all $n \in \mathbb{N}^*$,
\begin{align}
& S_2(n-1) \Rightarrow S_1(n) \label{Statement4} \\
& S_1(n) ~~\text{and} ~~ H_3(n-1) \Rightarrow H_1(n). \label{Statement5} \\
& H_1(n) \Rightarrow H_2(n) \label{Statement1} \\
& H_2(n) \Rightarrow H_3(n) \label{Statement2} \\
& H_1(n) \Rightarrow S_2(n) \label{Statement3} 
\end{align} 
By definition, $\tau \cap [E , E+ m_1 ) = \emptyset$ and hence $S_1(1)$ is obviously satisfied. We refer to \cite{Mo80_01,AmBoGe96_01,DeGe97_01} for the proofs of \eqref{Statement4}, \eqref{Statement1}, \eqref{Statement2} and \eqref{Statement3}. Note that \eqref{Statement1} uses the compactness of $K_0$ and finite rank operator estimates, and that \eqref{Statement3} is a consequence of the virial theorem. It remains to prove $H_1(1)$ and \eqref{Statement5}.

Using that the commutators $[H_I^{(j)}(G) , \mathrm{i} A]$, $j=1,\dots,4$, are given by \eqref{commut_11}--\eqref{commut_41}, a direct application of the $N_\tau$ estimates of Glimm and Jaffe (see \cite[Proposition 1.2.3(b)]{GlJa77_01}) shows that there exist $\mathrm{c}_1>0$ and $\mathrm{c}_2>0$ such that
\begin{equation}\label{estimH'I}
\big|\langle \Psi , H'_I  \Psi \rangle \big | \leq \mathrm{c}_1 \langle \Psi , H'_0 \Psi \rangle + \mathrm{c}_2 \|\Psi\|^2,
\end{equation}
for all $\Psi \in \mathscr{D}( H'_0 )$. Here it should be noticed that $| \nabla \omega^{(3)}( p_3 ) |^2 = p_3^2 ( p_3^2 + m_W^2 )^{-1}$. Hence, according to \cite[Proposition 1.2.3(b)]{GlJa77_01}, the constant $\mathrm{c}_1$ can be chosen to be proportional to $\|| p_3 |^{-1} G \|_2 + \| G \|_2$, which is finite by \eqref{eq:hypclassC1_3_1_1_1}.

Recall that the operator $\check{\Gamma}(j^R) : \mathscr{H} \to \mathscr{H}^{\mathrm{ext}}$ has been defined in Section \ref{subsec:estimates} (see the paragraph after the statement of Lemma \ref{lemV2Numbestimates} and see also Appendix \ref{sec:standard_def}). Using that $\check{\Gamma}(j^R)$ is an isometry, we can write
\begin{align*}
 H'_0 + g H'_I &= \check{\Gamma}(j^R)^*\check{\Gamma}(j^R) (H'_0 + g H'_I ) \\
&= \check{\Gamma}(j^R)^* (H'_0\otimes\mathds{1}+\mathds{1}\otimes H'_0 + g H'_I \otimes \mathds{1} + \mathrm{Rem}_R ) \check{\Gamma}(j^R) ,
\end{align*}
with $\mathrm{Rem}_R ( N_0 + N_\infty + 1 )^{-1} = o( R^0 )$. The last equality can be proven in the same way as in the proofs of Lemmas \ref{numberestimates3}--\ref{numberestimates13}. We decompose
\begin{align}
[H, \mathrm{i} A] &= \check{\Gamma}(j^R)^* (H'_0\otimes\mathds{1}+\mathds{1}\otimes H'_0 + g H'_I \otimes \mathds{1} + \mathrm{Rem}_R ) \mathds{1}_{\{0\}}( N_{\mathrm{neut},\infty} ) \check{\Gamma}(j^R) \notag \\
&\quad + \check{\Gamma}(j^R)^* (H'_0\otimes\mathds{1}+\mathds{1}\otimes H'_0 + g H'_I \otimes \mathds{1} + \mathrm{Rem}_R ) \mathds{1}_{ [1,\infty) }( N_{\mathrm{neut},\infty} ) \check{\Gamma}(j^R) , \label{eq:decomp_1}
\end{align}
and estimate the two terms separately. For the second one, we notice that 
\begin{equation}
( \mathds{1}\otimes H'_0 ) \mathds{1}_{ [1,\infty) }( N_{\mathrm{neut},\infty} ) \ge \frac12 ( \mathds{1}\otimes N_{\mathrm{neut},\infty} ) \mathds{1}_{ [1,\infty) }( N_{\mathrm{neut},\infty} ) \ge \frac12 \mathds{1}_{ [1,\infty) }( N_{\mathrm{neut},\infty} ) . \label{eq:lsif1}
\end{equation}
By \eqref{estimH'I}, we have that
\begin{align*}
& \check{\Gamma}(j^R)^* (H'_0\otimes\mathds{1}+\mathds{1}\otimes H'_0 + g H'_I \otimes \mathds{1} + \mathrm{Rem}_R ) \mathds{1}_{ [1,\infty) }( N_{\mathrm{neut},\infty} ) \check{\Gamma}(j^R) \\
&\ge \check{\Gamma}(j^R)^* ((1 - \mathrm{c}_1 g) H'_0 \otimes\mathds{1} - \mathrm{c}_2 g + \mathds{1}\otimes H'_0 + \mathrm{Rem}_R ) \mathds{1}_{ [1,\infty) }( N_{\mathrm{neut},\infty} ) \check{\Gamma}(j^R) \\
&\ge \check{\Gamma}(j^R)^* ((1 - \mathrm{c}_1 g) H'_0 \otimes\mathds{1} - \mathrm{c}_2 g + \mathds{1}\otimes H'_0 \\
&\quad + \mathrm{Rem}_R (N_0+N_{\infty}+1)^{-1}(N_0+N_{\infty}+1)      ) \mathds{1}_{ [1,\infty) }( N_{\mathrm{neut},\infty} ) \check{\Gamma}(j^R) \\
&\ge \check{\Gamma}(j^R)^* ((1 - \mathrm{c}_1 g) H'_0 \otimes\mathds{1} - \mathrm{c}_2 g + \mathds{1}\otimes H'_0 + o(R^{0})(N_0+N_{\infty}+1)      ) \mathds{1}_{ [1,\infty) }( N_{\mathrm{neut},\infty} ) \check{\Gamma}(j^R). 
\end{align*}
For $g$ small enough and $R$ large enough, since $N_{\mathrm{lept}}$ and $N_W$ are relatively $H$-bounded, this yields
\begin{align}
& \check{\Gamma}(j^R)^* (H'_0\otimes\mathds{1}+\mathds{1}\otimes H'_0 + g H'_I \otimes \mathds{1} + \mathrm{Rem}_R ) \mathds{1}_{ [1,\infty) }( N_{\mathrm{neut},\infty} ) \check{\Gamma}(j^R) \notag \\
& \ge  \check{\Gamma}(j^R)^*((1-\mathrm{c}_1 g)H'_0\otimes\mathds{1}-\mathrm{c}_2g+\mathds{1}\otimes H'_0+o(R^0)(N_{\mathrm{neut},0}+N_{\mathrm{neut},\infty})) \mathds{1}_{ [1,\infty) }( N_{\mathrm{neut},\infty} ) \check{\Gamma}(j^R) \notag  \\
& +  o( R^0 ) \check{\Gamma}(j^R)^* \mathds{1}_{ [1,\infty) }( N_{\mathrm{neut},\infty} ) \check{\Gamma}(j^R) ( 1 + H^2 )^{\frac12} \notag \\
& \ge \frac14 \check{\Gamma}(j^R)^* \mathds{1}_{ [1,\infty) }( N_{\mathrm{neut},\infty} ) \check{\Gamma}(j^R) + o( R^0 ) \check{\Gamma}(j^R)^* \mathds{1}_{ [1,\infty) }( N_{\mathrm{neut},\infty} ) \check{\Gamma}(j^R) ( 1 + H^2 )^{\frac12}. \label{eq:comput_Mourre_1}
\end{align}
In the last inequality, we used \eqref{eq:lsif1} and the fact that $H'_0\ge0$.

Now, we consider the first term in \eqref{eq:decomp_1}. As in the proof of Theorem \ref{Thess2}, we use the fact that, for all bounded interval $I$ and $R>0$, $\check{\Gamma}(j^R)^* ( \mathds{1} \otimes \Pi_\Omega ) \check{\Gamma}(j^R) \chi_I(H)$ is compact. Thus, using that $N_{\mathrm{neut}}\le H'_0$, that $H'_I$ is relatively $(N_{\mathrm{lept}}+N_W)$-bounded and that $N_{\mathrm{lept}}$ and $N_W$ are relatively $H$-bounded, we can write
\begin{align}
& \check{\Gamma}(j^R)^* (H'_0\otimes\mathds{1}+\mathds{1}\otimes H'_0 + g H'_I \otimes \mathds{1} + \mathrm{Rem}_R ) \mathds{1}_{\{0\}}( N_{\mathrm{neut},\infty} ) \mathds{1}_{\{0\}}( N_{\mathrm{lept},\infty} + N_{W,\infty} )\check{\Gamma}(j^R)  \notag \\
&= \check{\Gamma}(j^R)^* (H'_0\otimes\mathds{1} + g H'_I \otimes \mathds{1} + \mathrm{Rem}_R ) (\mathds{1} \otimes \Pi_\Omega ) )\check{\Gamma}(j^R) \notag \\
&\ge \check{\Gamma}(j^R)^* (  H'_0 \otimes \mathds{1} + g H'_I \otimes \mathds{1} + o( R^0 )( N_{\mathrm{neut}} \otimes \mathds{1} + N_{\mathrm{lept}} \otimes \mathds{1} + N_W \otimes \mathds{1} ) ) (\mathds{1} \otimes \Pi_\Omega ) )\check{\Gamma}(j^R) \notag \\
&\ge \check{\Gamma}(j^R)^* ( (1+o( R^0 )) H'_0 \otimes \mathds{1} + g H'_I \otimes \mathds{1} + o( R^0 )( N_{\mathrm{lept}} \otimes \mathds{1} + N_W \otimes \mathds{1} ) ) (\mathds{1} \otimes \Pi_\Omega ) )\check{\Gamma}(j^R) \notag \\
&\ge \check{\Gamma}(j^R)^* ( g H'_I \otimes \mathds{1} + o( R^0 )( N_{\mathrm{lept}} \otimes \mathds{1} + N_W \otimes \mathds{1} ) ) (\mathds{1} \otimes \Pi_\Omega ) )\check{\Gamma}(j^R) \notag \\
& \ge K_{R,I} - \mathrm{c}_3 \check{\Gamma}(j^R)^* (\mathds{1} \otimes \Pi_\Omega ) \check{\Gamma}(j^R) ( \mathds{1} - \chi_I( H ) ) ( 1 + H^2 )^{\frac12} , \label{eq:comput_Mourre_2}
\end{align}
for any bounded interval $I$ and $R>0$ large enough, where $K_{R,I}$ is compact and $\mathrm{c}_3$ is a positive constant.

It remains to consider
\begin{align*}
& \check{\Gamma}(j^R)^* (H'_0\otimes\mathds{1}+\mathds{1}\otimes H'_0 + g H'_I \otimes \mathds{1} + \mathrm{Rem}_R ) \mathds{1}_{\{0\}}( N_{\mathrm{neut},\infty} ) \mathds{1}_{[1,\infty)}( N_{\mathrm{lept},\infty} + N_{W,\infty} )\check{\Gamma}(j^R)  \\
&\ge \check{\Gamma}(j^R)^* ( H'_0\otimes\mathds{1} + g H'_I \otimes \mathds{1} + o(R^0)(N_0+N_{\infty}+1) ) \mathds{1}_{\{0\}}( N_{\mathrm{neut},\infty} ) \mathds{1}_{[1,\infty)}( N_{\mathrm{lept},\infty} + N_{W,\infty} )\check{\Gamma}(j^R).
\end{align*}
 Here we used that $\mathds{1} \otimes H'_0 \ge 0$. We introduce $\mathds{1} = \chi_I(H) + ( \mathds{1} - \chi_I(H) )$ on the right. Using that $H'_0 \otimes \mathds{1} \ge 0$, we can write
\begin{align}
& \check{\Gamma}(j^R)^* (H'_0\otimes\mathds{1}+\mathds{1}\otimes H'_0 + g H'_I \otimes \mathds{1} + \mathrm{Rem}_R ) \mathds{1}_{\{0\}}( N_{\mathrm{neut},\infty} ) \mathds{1}_{[1,\infty)}( N_{\mathrm{lept},\infty} + N_{W,\infty} )\check{\Gamma}(j^R) \notag \\
&\ge \check{\Gamma}(j^R)^* (( 1 - o( R^0 ) ) ( H'_0 + g H'_I ) \otimes\mathds{1} ) \mathds{1}_{\{0\}}( N_{\mathrm{neut},\infty} ) \mathds{1}_{[1,\infty)}( N_{\mathrm{lept},\infty} + N_{W,\infty} )\check{\Gamma}(j^R) \chi_I(H) \notag \\
&\quad +  \check{\Gamma}(j^R)^*  (o(R^0)g H'_I \otimes \mathds{1} + o(R^0)(N_0+N_{\infty}+1) ) \mathds{1}_{\{0\}}( N_{\mathrm{neut},\infty} ) \mathds{1}_{[1,\infty)}( N_{\mathrm{lept},\infty} + N_{W,\infty} )\check{\Gamma}(j^R)\chi_I( H )   \notag \\
&\quad  +
\check{\Gamma}(j^R)^* ( H'_0\otimes\mathds{1} + g H'_I \otimes \mathds{1} + o(R^0)(N_0+N_{\infty}+1) ) \notag \\
& \qquad \qquad \qquad \qquad \qquad \mathds{1}_{\{0\}}( N_{\mathrm{neut},\infty} ) \mathds{1}_{[1,\infty)}( N_{\mathrm{lept},\infty} + N_{W,\infty} )\check{\Gamma}(j^R) ( \mathds{1} - \chi_I( H ) ) . \notag
\end{align} 
Hence, using again that $H'_I$ is relatively $(N_{\mathrm{lept}}+N_W)$-bounded and that $N_{\mathrm{lept}}$ and $N_W$ are relatively $H$-bounded, we obtain that
\begin{align}
& \check{\Gamma}(j^R)^* (H'_0\otimes\mathds{1}+\mathds{1}\otimes H'_0 + g H'_I \otimes \mathds{1} + \mathrm{Rem}_R ) \mathds{1}_{\{0\}}( N_{\mathrm{neut},\infty} ) \mathds{1}_{[1,\infty)}( N_{\mathrm{lept},\infty} + N_{W,\infty} )\check{\Gamma}(j^R) \notag \\
&\ge \check{\Gamma}(j^R)^* (( 1 - o( R^0 ) ) ( H'_0 + g H'_I ) \otimes\mathds{1} ) \mathds{1}_{\{0\}}( N_{\mathrm{neut},\infty} ) \mathds{1}_{[1,\infty)}( N_{\mathrm{lept},\infty} + N_{W,\infty} )\check{\Gamma}(j^R) \chi_I(H) \notag \\
&\quad + o(R^0) \check{\Gamma}(j^R)^*  \mathds{1}_{\{0\}}( N_{\mathrm{neut},\infty} ) \mathds{1}_{[1,\infty)}( N_{\mathrm{lept},\infty} + N_{W,\infty} )\check{\Gamma}(j^R) \chi_I(H)  \notag \\
&\quad - \mathrm{c}_4 \check{\Gamma}(j^R)^* \mathds{1}_{\{0\}}( N_{\mathrm{neut},\infty} ) \mathds{1}_{[1,\infty)}( N_{\mathrm{lept},\infty} + N_{W,\infty} ) \check{\Gamma}(j^R) ( \mathds{1} - \chi_I( H ) ) ( 1 + H^2 )^{\frac12} , \label{eq:comput_Mourre_9}
\end{align}
for any bounded interval $I$ and $R>0$, where $\mathrm{c}_4$ is a positive constant.

Now, we can prove $H_1(1)$ and \eqref{Statement5}. To prove $H_1(1)$, let $\varepsilon>0$, $\lambda \in [E,E+ m_1)$ and $\Delta$ be an interval containing $\lambda$ and supported in $[E,E+m_1)$. The function $\chi_\Delta \in \mathrm{C}_0^\infty( \mathbb{R} )$ is chosen such that $\mathrm{supp} (\chi_\Delta) \subset ( - \infty , E+m_1)$ and we consider in addition $\tilde\chi_\Delta \in \mathrm{C}_0^\infty( \mathbb{R} )$ such that $\mathrm{supp} (\tilde\chi_\Delta) \subset ( - \infty , E+m_1)$ and $\tilde\chi_\Delta \chi_\Delta = \chi_\Delta$. Then, using that $H'_0 \ge 0$, the fact that $H'_I$ is $(N_{\mathrm{lept}}+N_W)$-bounded as before, and Lemma \ref{numberestimates13}, the first term in the right-hand-side of \eqref{eq:comput_Mourre_9} (with $I=\Delta$) can be estimated, for $R$ large enough, in the following way:
\begin{align}
&\check{\Gamma}(j^R)^* (( 1 - o( R^0 ) ) ( H'_0 + g H'_I ) \otimes\mathds{1} ) \mathds{1}_{\{0\}}( N_{\mathrm{neut},\infty} ) \mathds{1}_{[1,\infty)}( N_{\mathrm{lept},\infty} + N_{W,\infty} )\check{\Gamma}(j^R) \chi_\Delta(H) \notag \\
&\ge\check{\Gamma}(j^R)^* (( 1 - o( R^0 ) ) H'_0 \otimes\mathds{1} ) \mathds{1}_{\{0\}}( N_{\mathrm{neut},\infty} ) \mathds{1}_{[1,\infty)}( N_{\mathrm{lept},\infty} + N_{W,\infty} )\check{\Gamma}(j^R) \chi_\Delta(H) \notag \\
&\quad - \mathrm{c}_5 \check{\Gamma}(j^R)^* \mathds{1}_{\{0\}}( N_{\mathrm{neut},\infty} ) \mathds{1}_{[1,\infty)}( N_{\mathrm{lept},\infty} + N_{W,\infty} )\check{\Gamma}(j^R) \chi_\Delta(H) \notag \\
&\ge\check{\Gamma}(j^R)^* (( 1 - o( R^0 ) ) H'_0 \otimes\mathds{1} ) \mathds{1}_{\{0\}}( N_{\mathrm{neut},\infty} ) \mathds{1}_{[1,\infty)}( N_{\mathrm{lept},\infty} + N_{W,\infty} )\check{\Gamma}(j^R) \chi_\Delta(H) \notag \\
& - \mathrm{c}_5 \check{\Gamma}(j^R)^* \mathds{1}_{\{0\}}( N_{\mathrm{neut},\infty} ) \mathds{1}_{[1,\infty)}( N_{\mathrm{lept},\infty} + N_{W,\infty} ) \tilde \chi_\Delta(H^{\mathrm{ext}})  \check{\Gamma}(j^R) \chi_\Delta(H) \notag \\
&\ge - \mathrm{c}_5 \check{\Gamma}(j^R)^* \mathds{1}_{\{0\}}( N_{\mathrm{neut},\infty} ) \mathds{1}_{[1,\infty)}( N_{\mathrm{lept},\infty} + N_{W,\infty} ) \tilde \chi_\Delta(H^{\mathrm{ext}}) \check{\Gamma}(j^R) \chi_\Delta(H) \notag \\
&=0. \label{eq:comput_Mourre_10}
\end{align} 
The last equality comes from the fact that $H^{\mathrm{ext}} = H \otimes \mathds{1}+\mathds{1} \otimes H_0 \ge E + m_1$ on the range of the operator $\mathds{1}_{[1,\infty)}( N_{\mathrm{lept},\infty} + N_{W,\infty} )$, and hence that $\tilde \chi_\Delta(H^{\mathrm{ext}})\mathds{1}_{[1,\infty)}( N_{\mathrm{lept},\infty} + N_{W,\infty} )=0$, since $\tilde\chi_\Delta$ is supported in $(-\infty,E+m_1)$. Since $d(\lambda)=0$ for $\lambda \in (-\infty,E+m_1)$, fixing $R$ large enough, $I=\Delta$ in \eqref{eq:comput_Mourre_2} and \eqref{eq:comput_Mourre_9}, and combining Equations \eqref{eq:decomp_1}--\eqref{eq:comput_Mourre_10}, we obtain $H_1(1)$.

To prove \eqref{Statement5}, let $\varepsilon>0$ and $\lambda \in [E,E+ n m_1 )$. We go back to \eqref{eq:comput_Mourre_9}, with $I=\Delta^{\delta}_\mu$, and consider again the first term in the right-hand-side. By $S_1(n)$, $\tau$, where $d$ vanishes, is a closed countable set in $[E,E+n m_1 ]$, which implies that there exists $\kappa$ such that $\tau \cap \Delta^{\kappa}_{\lambda} \neq \emptyset$ and then $d(\lambda) = \sup_{\kappa > 0} d^{\kappa}(\lambda)$. Hence there exists $\kappa > 0$ such that $d^{\kappa}(\lambda)> d(\lambda)-\varepsilon/3$. Let $\varepsilon_0>0$ be such that $\lambda \in [E,E+ n m_1 -\varepsilon_0 ]$. By $H_3(n-1)$, we know that there exist $\mathrm{c}_6\in\mathbb{R}$ and $\delta >0$ such that
\begin{equation}
H'_0 + gH'_I \geq ( \min ( \frac14 , \tilde{d}^{\kappa}(\mu) )- \frac{\varepsilon}{3}) \mathds{1} - \mathrm{c}_6 ( \mathds{1} - \tilde\chi_{\Delta^{\delta}_{\mu}} (H)  )( 1 + H^2 )^{\frac12} , \label{eq:H3(n)}
\end{equation}
for all $\mu \in [E , E+(n-1)m_1 - \varepsilon_0 ]$. Here $\tilde\chi_{\Delta^{\delta}_{\lambda}}$ is chosen such that $\tilde\chi_{\Delta^{\delta}_{\lambda}}\chi_{\Delta^{\delta}_{\lambda}}=\tilde\chi_{\Delta^{\delta}_{\lambda}}$, where $\chi_{\Delta^{\delta}_{\lambda}}$ is the function appearing in \eqref{eq:comput_Mourre_9}. We begin by estimating the first term in the right-hand-side of \eqref{eq:comput_Mourre_9} as
\begin{align}
& \check{\Gamma}(j^R)^* (( 1 - o( R^0 ) ) ( H'_0 + g H'_I ) \otimes\mathds{1} ) \mathds{1}_{\{0\}}( N_{\mathrm{neut},\infty} ) \mathds{1}_{[1,\infty)}( N_{\mathrm{lept},\infty} + N_{W,\infty} )\check{\Gamma}(j^R) \chi_{\Delta_\lambda^\delta}(H) \notag \\
&\geq \check{\Gamma}(j^R)^* (( 1 - o( R^0 ) ) ( \min ( \frac14 , \tilde{d}^{\kappa}(\mu) )- \frac{\varepsilon}{3}) \mathds{1} \otimes\mathds{1} ) \mathds{1}_{\{0\}}( N_{\mathrm{neut},\infty} ) \mathds{1}_{[1,\infty)}( N_{\mathrm{lept},\infty} + N_{W,\infty} )\check{\Gamma}(j^R) \chi_{\Delta_\lambda^\delta}(H) \notag \\
&  - \mathrm{c}_6 \check{\Gamma}(j^R)^* (( 1 - o( R^0 ) )  (( \mathds{1} - \tilde\chi_{\Delta^{\delta}_{\mu}} (H)  )( 1 + H^2 )^{\frac12} \otimes\mathds{1} ) \mathds{1}_{\{0\}}( N_{\mathrm{neut},\infty} ) \mathds{1}_{[1,\infty)}( N_{\mathrm{lept},\infty} + N_{W,\infty} )\check{\Gamma}(j^R) \chi_{\Delta_\lambda^\delta}(H) .\notag
\end{align}
On the range of $\mathds{1}_{[1,\infty)}( N_{\mathrm{lept},\infty} + N_{W,\infty} )$, we have that $\mathds{1} \otimes H_0 \ge m_1$. Therefore, by \eqref{eq:H3(n)} and the functional calculus, with $\mu = \lambda - 1 \otimes H_0$, we obtain that
\begin{align}
&  ( 1 - o(R^0)) \check{\Gamma}(j^R)^* ( \min ( \frac14 , \tilde{d}^\kappa( \lambda - \mathds{1} \otimes H_0 ) ) - \frac{\varepsilon}{3} ) \mathds{1}_{\{0\}}( N_{\mathrm{neut},\infty} ) \mathds{1}_{[1,\infty)}( N_{\mathrm{lept},\infty} + N_{W,\infty} )\check{\Gamma}(j^R) \chi_{\Delta_\lambda^\delta}(H)  \notag \\
& = ( 1 - o(R^0)) \check{\Gamma}(j^R)^* ( \min ( \frac14 , d^\kappa( \lambda ) ) - \frac{\varepsilon}{3} ) \mathds{1}_{\{0\}}( N_{\mathrm{neut},\infty} ) \mathds{1}_{[1,\infty)}( N_{\mathrm{lept},\infty} + N_{W,\infty} )\check{\Gamma}(j^R) \chi_{\Delta_\lambda^\delta}(H) \notag \\
& \ge ( 1 - o(R^0)) ( \min ( \frac14 , d(\lambda) ) - \frac{2\varepsilon}{3} ) \check{\Gamma}(j^R)^* \mathds{1}_{\{0\}}( N_{\mathrm{neut},\infty} ) \mathds{1}_{[1,\infty)}( N_{\mathrm{lept},\infty} + N_{W,\infty} )\check{\Gamma}(j^R) \chi_{\Delta_\lambda^\delta}(H) \notag \\
& \ge ( \min ( \frac14 , d(\lambda) ) - \varepsilon ) \check{\Gamma}(j^R)^* \mathds{1}_{\{0\}}( N_{\mathrm{neut},\infty} ) \mathds{1}_{[1,\infty)}( N_{\mathrm{lept},\infty} + N_{W,\infty} )\check{\Gamma}(j^R) \chi_{\Delta_\lambda^\delta}(H) , \label{eq:comput_Mourre_12}
\end{align}
where the equality comes from the definitions of $\tilde{d}^\kappa$ and $d^\kappa$ and the fact that $\tilde{d}^\kappa( \lambda - \mathds{1} \otimes H_0 )$ acts on the range of $\mathds{1}_{[1,\infty)}( N_{\mathrm{lept},\infty} + N_{W,\infty} )$. 
Moreover, let $\tilde{\tilde{\chi}}_{\Delta_\lambda^\delta} \in \mathrm{C}_0^{\infty}(\mathbb{R})$ be such that $\tilde{\tilde{\chi}}_{\Delta_\lambda^\delta}\chi_{\Delta_\lambda^\delta}= \chi_{\Delta_\lambda^\delta}$. Using Lemma \ref{numberestimates3}, we write
\begin{align}
&  \mathrm{c}_6 \check{\Gamma}(j^R)^* (( 1 - o( R^0 ) ) ( ( \mathds{1} - \tilde\chi_{\Delta^{\delta}_{\mu}} (H)  )( 1 + H^2 )^{\frac12} \otimes\mathds{1} ) \notag \\
&\qquad \qquad \mathds{1}_{\{0\}}( N_{\mathrm{neut},\infty} ) \mathds{1}_{[1,\infty)}( N_{\mathrm{lept},\infty} + N_{W,\infty} )\check{\Gamma}(j^R) \chi_{\Delta_\lambda^\delta}(H) \tilde{\tilde{\chi}}_{\Delta_\lambda^\delta}(H) \notag \\
& = \mathrm{c}_6(( 1 - o( R^0 ) )  \check{\Gamma}(j^R)^*  \mathds{1}_{\{0\}}( N_{\mathrm{neut},\infty} ) \mathds{1}_{[1,\infty)}( N_{\mathrm{lept},\infty} + N_{W,\infty} ) \notag \\
&\qquad \qquad ((( \mathds{1} - \tilde{\chi}_{\Delta^{\delta}_{\mu}} (H)  )( 1 + H^2 )^{\frac12} \otimes\mathds{1} )\chi_{\Delta_\lambda^\delta}(H^{\mathrm{ext}})\check{\Gamma}(j^R) \tilde{\tilde{\chi}}_{\Delta_\lambda^\delta}(H)  + o(R^0) \notag \\
& =  o(R^0). \label{eq:comput_Mourre_13}
\end{align}
The last equality comes from $\mu = \lambda - \mathds{1}\otimes H_0$ and $\tilde\chi_{\Delta^{\delta}_{\lambda}}\chi_{\Delta^{\delta}_{\lambda}}=\tilde\chi_{\Delta^{\delta}_{\lambda}}$. Fixing $R$ large enough and $I = \Delta_\lambda^\delta$, Equations \eqref{eq:decomp_1}, \eqref{eq:comput_Mourre_1}, \eqref{eq:comput_Mourre_2},  \eqref{eq:comput_Mourre_9}, \eqref{eq:comput_Mourre_12} and \eqref{eq:comput_Mourre_13} prove \eqref{Statement5}. This concludes the proof of \eqref{eq:Mourre_estimate_1_4}.

The fact that $\sigma_{\mathrm{sc}}(H) = \emptyset$ if \eqref{eq:hypclassC2_2_4} holds follows from Lemma \ref{lem:C8} and the abstract results recalled above.
\end{proof}
\section{Propagation estimates}\label{sec:propag}

In this section, we use the method of propagation observables that was developed in $N$-body scattering theory (see e.g. \cite{SiSo87_01,HuSiSo99_01,Gr90_01,De93_01,DeGe97_01} and references therein). This method was adapted to the context of Pauli-Fierz or $P(\varphi)_2$ Hamiltonians in several papers (see, in particular, \cite{DeGe99_01,DeGe00_01,FrGrSc02_01,FrGrSc04_01,Am04_01,FrGrSc07_01,BoFaSi12_01,FaSi14_01}). In the case where all particles are massive, the propagation estimates that we prove in this section are straightforwardly adapted from \cite{Gr90_01,DeGe99_01}. In the case where neutrinos are supposed to be massless, however, the propagation estimates have to be substantially modified. The main idea consists in replacing the usual one-particle position operator by a suitably modified, time-dependent one-particle ``position'' operator. A related trick was used in \cite{FaSi14_01} for Pauli-Fierz Hamiltonians, but the details of the analysis here are different, in particular because we rely heavily on the Mourre estimate proven in Theorem \ref{MourreI_3}.

The basic approach that we follow to prove our propagation estimates is the following: let $H$ be a self-adjoint operator on a Hilbert space $\mathscr{H}$ and let $\Phi(t)$ be a time-dependent family of self-adjoint operators. Suppose that for some $u \in \mathcal{H}$, 
\begin{align*}
& \big \langle e^{ - \mathrm{i} t H } u , \Phi( t ) e^{ - \mathrm{i} t H } u \big \rangle \le \mathrm{C}_u , \quad \text{ uniformly in } t \ge 1 ,  
\end{align*}
and that one of the following to conditions holds,
\begin{align}
& \partial_t \big \langle e^{ - \mathrm{i} t H } u , \Phi( t ) e^{ - \mathrm{i} t H } u \big \rangle \ge \big \langle e^{ - \mathrm{i} t H } u , \Psi(t) e^{ - \mathrm{i} t H } u \big \rangle - \sum_{j=1}^n \big \langle e^{ - \mathrm{i} t H } u , B^*_j(t) B_j(t) e^{ - \mathrm{i} t H } u \big \rangle, \label{eq:condPO} \\
& \partial_t \big \langle e^{ - \mathrm{i} t H } u , \Phi( t ) e^{ - \mathrm{i} t H } u \big \rangle \le - \big \langle e^{ - \mathrm{i} t H } u , \Psi(t) e^{ - \mathrm{i} t H } u \big \rangle + \sum_{j=1}^n \big \langle e^{ - \mathrm{i} t H } u , B^*_j(t) B_j(t) e^{ - \mathrm{i} t H } u \big \rangle, \label{eq:condPO_1}
\end{align}
where $\Psi(t)$ are positive operators and $B_i(t)$ are families of time-dependent operators such that
\begin{equation}
\int_1^\infty \big \| B_i(t) e^{-\mathrm{i}tH}u\big\|^2dt \le\mathrm{C}_u. \label{eq:cond2PO}
\end{equation}
Then, integrating with respect to $t$, one obtains that
\begin{align}
\int_1^\infty \big \langle e^{ - \mathrm{i} t H } u , \Psi( t ) e^{ - \mathrm{i} t H } u \big \rangle dt \lesssim  \mathrm{C}_u , \label{eq:concPO}
\end{align}
which is sometimes called a weak propagation estimate for the family of observables $\Psi(t)$. Observe that the left-hand-sides of \eqref{eq:condPO}--\eqref{eq:condPO_1} can be rewritten as
\begin{equation*}
\big \langle e^{ - \mathrm{i} t H } u , \mathbf{D} \Phi( t ) e^{ - \mathrm{i} t H } u \big \rangle ,
\end{equation*}
where $\mathbf{D}$ stands for the Heisenberg derivative $\mathbf{D} \Phi ( t ) = \partial_t \Phi( t ) + [ H , \mathrm{i} \Phi( t ) ]$. Therefore, to prove the propagation observable \eqref{eq:concPO}, it suffices to find a family of operators $\Phi(t)$ whose Heisenberg derivative ``dominates'' $\Psi(t)$, in the sense that $\mathbf{D} \Phi ( t ) \ge \Psi(t)$ or $\mathbf{D} \Phi ( t ) \le - \Psi(t)$, up to remainder terms that are integrable in the sense of \eqref{eq:cond2PO}. The strategy usually consists in comparing the time derivative $\partial_t \Phi( t )$ and the commutator $[ H , \mathrm{i} \Phi( t ) ]$, possibly by means of a ``commutator expansion" of $[ H , \mathrm{i} \Phi( t ) ]$. We refer the reader to e.g. \cite{DeGe97_01} for details on the method of propagation observables, see also \cite[Section 2]{FaSi14_01} for a description of the method of propagation observables for Pauli-Fierz Hamiltonians, closely related to the approach that we will follow here.

It is useful to introduce the following notations
\begin{align*}
\mathbf{d}^{(i)}_{0l} b(t) =  \frac{\partial b}{\partial t}(t) + [\omega^{(i)}_l(p_i), \mathrm{i} b(t) ] , \quad i=1,2 , \quad \mathbf{d}^{(3)}_{0} b(t) =  \frac{\partial b}{\partial t}(t) + [\omega^{(3)}(p_i), \mathrm{i} b(t) ] ,
\end{align*}  
if $b(t)$ is a family of operators, acting on $\mathfrak{h}_1$ if $i=1,2$, or acting on $\mathfrak{h}_2$ if $i=3$. Likewise, we set
\begin{align*}
\mathbf{D}_0 B(t) = \frac{\partial B}{\partial t}(t) + [H_0, \mathrm{i} B(t) ] , \quad \mathbf{D} B(t) = \frac{\partial B}{\partial t}(t) + [H, \mathrm{i} B(t) ] ,
\end{align*}
if $B(t)$ is a family of operators acting on $\mathscr{H}$.
Note that if $B(t) = (b_1(t) , \dots , b_{14}(t) )$, with the notations of Appendix \ref{sec:standard_def} (remembering that the total Hilbert space $\mathscr{H}$ is the tensor product of $14$ Fock spaces), then, as functions of $t$,
\begin{align*}
 \mathbf{D}_0 \mathrm{d}\Gamma(B) =: \mathrm{d}\Gamma( \mathbf{d}_0 b ) := \mathrm{d}\Gamma( & \mathbf{d}_{0,1} b_1 , \dots , \mathbf{d}_{0,14} b_{14} ) \\
:=  \mathrm{d}\Gamma\big( & \mathbf{d}^{(3)}_{0} b_1 , \mathbf{d}^{(3)}_{0} b_2 , \mathbf{d}^{(1)}_{01} b_3 , \mathbf{d}^{(1)}_{01} b_4 , \mathbf{d}^{(2)}_{01} b_5 , \mathbf{d}^{(2)}_{01} b_6 , \mathbf{d}^{(1)}_{02} b_7, \\
 &\mathbf{d}^{(1)}_{02} b_8 , \mathbf{d}^{(2)}_{02} b_9 , \mathbf{d}^{(2)}_{02} b_{10} , \mathbf{d}^{(1)}_{03} b_{11} ,\mathbf{d}^{(1)}_{03} b_{12} , \mathbf{d}^{(2)}_{01} b_{13} , \mathbf{d}^{(2)}_{01} b_{14} \big).
\end{align*}

In the remainder of this section, we prove the propagation estimates that will be used in Section \ref{sec:AC}. We begin with the case where the masses of the neutrinos are supposed to be positive in Section \ref{subsec:massive_propag}, next we turn to the much more difficult case where neutrinos are supposed to be massless (Section \ref{subsec:massless_propag}).

\subsection{Massive neutrinos}\label{subsec:massive_propag}
As mentioned before, the proofs of the propagation estimates of this section are almost straightforward adaptations of the ones in \cite{DeGe99_01,DeGe00_01}. We state the results and emphasize the differences with \cite{DeGe99_01,DeGe00_01} but the details of the proofs are left to the reader.

To shorten expressions below, we set $x_i = (x_i^{(1)},x_i^{(2)},x_i^{(3)}) = \mathrm{i} \nabla$, $i=1,\dots,14$, and, for $R=(R_1,\dots,R_{14})$ and $R'=(R'_1,\dots,R'_{14})$,
\begin{equation}\label{eq:1Rx}
\mathds{1}_{[R,R']}(|x|)=\big(\mathds{1}_{[R_1,R'_1]}(|x_1|),\dots,\mathds{1}_{[R_{14},R'_{14}]}(|x_{14}|)\big)=\big(\mathds{1}_{[R_1,R'_1]}(\sqrt{-\Delta}),\dots,\mathds{1}_{[R_{14},R'_{14}]}(\sqrt{-\Delta})\big),
\end{equation}
and likewise for other functions of $x$. The operators $\mathrm{d}\Gamma(\mathds{1}_{[R,R']}(|x|))$, $\Gamma(\mathds{1}_{[R,R']}(|x|))$ are then defined, as in the previous sections, following the conventions of Appendix \ref{sec:standard_def}. We also set 
\begin{equation}\label{eq:not_omega}
\omega := (\omega^{(3)} , \omega^{(3)} , \omega^{(1)}_1 , \omega^{(1)}_1 , \omega^{(2)}_1 , \omega^{(2)}_1 , \omega^{(1)}_2 , \omega^{(1)}_2 , \omega^{(2)}_2 , \omega^{(2)}_2 , \omega^{(1)}_3 , \omega^{(1)}_3 , \omega^{(2)}_3 , \omega^{(2)}_3 ). 
\end{equation}

Recall that, for $i=1,2$, the notation ``$a_{(i),\cdot} G \in L^2$'' means that, for all $j$, $l$ and $\epsilon$, $a_{(i),l} G_{l,\epsilon}^{(j)}$ is square integrable, and that ``$a_{(3),\cdot} G \in L^2$'' means that for all $j$, $l$ and $\epsilon$, $a_{(3)} G_{l,\epsilon}^{(j)}$ is square integrable.
\begin{Th}\label{th:propag_massive}
Suppose that the masses of the neutrinos $m_{\nu_e}$, $m_{\nu_{\mu}}$, $m_{\nu_{\tau}}$ are positive and consider the Hamiltonian \eqref{total_hamiltonian_g} with $H_I$ given by \eqref{interaction_term}. Assume that
\begin{equation*}
G \in L^2, \quad a_{(i),\cdot} G \in L^2, \quad i = 1,2,3. 
\end{equation*}
and that
\begin{equation*}
G \in \mathbb{H}^{1+\mu} \text{ for some } \mu > 0.
\end{equation*}
\begin{itemize}
\item[(i)] Let $\chi \in \mathrm{C}_0^{\infty}(\mathbb{R})$, $R=(R_1,\dots,R_{14})$ and $R'=(R'_1,\dots,R'_{14})$ be such that $R'_{i}>R_i>1$. There exists $\mathrm{C} > 0$ such that, for all $u \in \mathscr{H}$, 
\begin{equation*}
\int^{\infty}_{1} \Big \| \mathrm{d} \Gamma \Big ( \mathds{1}_{[R,R']} \Big (\frac{|x|}{t} \Big ) \Big)^{\frac{1}{2}}\chi(H) e^{-\mathrm{i}tH} u \Big \|^2 \frac{dt}{t} \leq \mathrm{C} \|u\|^{2}.
\end{equation*}
\item[(ii)] Let $0<\mathrm{v}_0<\mathrm{v}_1$ and $\chi \in \mathrm{C}_0^{\infty}(\mathbb{R})$. There exists $\mathrm{C}>0$ such that, for all $\ell \in \{ 1,2,3\} $ and $u \in \mathscr{H}$,
\begin{align*}
& \int^{\infty}_{1} \Big \| \mathrm{d}\Gamma \Big( \Big\langle \big( \frac{x}{t} - \nabla \omega \big ) , \mathds{1}_{ [ \mathrm{v}_0 , \mathrm{v}_1 ] }( \frac{x}{t} ) \big( \frac{x}{t} - \nabla \omega \big ) \Big\rangle \Big)^{\frac12} \chi(H)e^{-\mathrm{i} t H} u \Big\|^{2} \frac{dt}{t}  \le \mathrm{C} \| u \|^2.
\end{align*}
\item[(iii)] Let $0<\mathrm{v}_0<\mathrm{v}_1$, $J \in \mathrm{C}_0^{\infty}( \{ x \in \mathbb{R}^3 , \mathrm{v}_0 < |x| < \mathrm{v}_1\})$ and $\chi \in \mathrm{C}_0^{\infty}(\mathbb{R})$. There exists $\mathrm{C}>0$ such that, for all $\ell \in \{ 1,2,3\} $ and $u \in \mathscr{H}$,
\begin{equation*}
\int^{\infty}_{1} \Big \| \mathrm{d}\Gamma \Big( \Big| J\Big(\frac{x}{t}\Big) \Big( \frac{x^{(\ell)}}{t} - \partial_{(\ell)} \omega \Big) + \mathrm{h.c }\Big|\Big)^{\frac{1}{2}}\chi(H)e^{-\mathrm{i} t H} u \Big\|^{2} \frac{dt}{t} \le \mathrm{C} \| u \|^2.
\end{equation*}
\item[(iv)]  Let $\chi \in \mathrm{C}_{0}^{\infty}(\mathbb{R})$ be supported in $\mathbb{R}\backslash(\tau \cup \sigma_{\mathrm{pp}}(H))$. There exists $\epsilon>0$ and $\mathrm{C} > 0$ such that, for all $u \in \mathscr{H}$,
\begin{equation*}
\int^{\infty}_{1} \Big\| \Gamma \Big ( \mathds{1}_{[0,\epsilon]} \Big(\frac{|x|}{t}\Big)\Big) \chi(H) e^{-\mathrm{i} t H} u \Big\|^2 \frac{dt}{t} \leq \mathrm{C}\| u \|^2.
 \end{equation*}
\end{itemize}
\end{Th}
\begin{proof}
We give a sketch of the proof of (i), underlying the differences with \cite{DeGe99_01}. The proofs of (ii)--(iv) can be achieved by similar arguments, adapting \cite[Section 6]{DeGe99_01} in a straightforward way.

By definition of the operator $\mathrm{d}\Gamma$, we have that
\begin{align}
\int^{\infty}_{1} \Big \| \mathrm{d}\Gamma \Big( \mathds{1}_{[R,R']} \Big(\frac{|x|}{t}\Big)\Big)^{\frac{1}{2}}\chi(H)e^{-\mathrm{i}tH} u \Big\|^2 \frac{dt}{t} & = \int^{\infty}_{1} \Big\| \sum^{14}_{i=1} \mathrm{d}\Gamma_i \Big(\mathds{1}_{[R_i,{R}'_i]}\Big(\frac{\sqrt{-\Delta}}{t}\Big)\Big)^{\frac{1}{2}}\chi(H) e^{-\mathrm{i}tH} u \Big\|^2 \frac{dt}{t} \notag \\
& \lesssim \sum^{14}_{i=1} \int^{\infty}_{1} \Big\|\mathrm{d}\Gamma_i\Big(\mathds{1}_{[R_i,{R}'_i]}\Big(\frac{\sqrt{-\Delta}}{t}\Big)\Big)^{\frac{1}{2}}\chi(H) e^{-\mathrm{i}tH} u \Big\|^2 \frac{dt}{t}, \label{eq:hdt1}
\end{align}
where we have set
\begin{equation}\label{eq:defdgammaj}
\mathrm{d}\Gamma_{i}(B) = \underset{i-1}{\underbrace{\mathds{1}\otimes \dots \otimes \mathds{1}}}\otimes \mathrm{d}\Gamma(B) \otimes \underset{14-i}{\underbrace{\mathds{1}\otimes \dots \otimes\mathds{1}}}.
\end{equation}
Next we proceed as in \cite{DeGe99_01,Am04_01}. Let $F \in \mathrm{C}^{\infty}(\mathbb{R})$ be equal to 1 near $\infty$, to 0 on a compact set near the origin, and such that $F'(s) \geq \mathds{1}_{[R_i,R'_i]}(s) $. We define in addition
\begin{align*}
\Phi(t) & =  \chi(H) \mathrm{d}\Gamma_i\Big(F\Big(\frac{|x_i|}{t}\Big)\Big)\chi(H) , \qquad b(t)  = \mathbf{d}_{0, i} F\left(\frac{|x_i|}{t}\right).
\end{align*}
A direct computation then shows that
\begin{equation*}
\mathbf{D} \Phi(t) = \chi(H) b(t) \chi(H) + \chi(H) \, \mathrm{i} \Big[ H_I,\mathrm{d}\Gamma_i\Big(F\Big(\frac{|x_i|}{t}\Big)\Big) \Big ]\chi(H) .
\end{equation*}
Using pseudo-differential calculus (or a commutator expansion at second order, see, e.g., \cite{DeGe97_01}) gives
\begin{align*}
b(t)  \leq  \frac{ - \mathrm{C}_0}{t} \mathds{1}_{[R_i,R'_i]}(t) + \mathcal{O}(t^{-2}) .
\end{align*}
Besides, since $F_i$ vanishes near $0$ and since $G \in \mathbb{H}^{1+\mu}$, we have that
\begin{align*}
\Big\|F_i\Big(\frac{|x_i|}{t}\Big)G\Big\|_2=\mathcal{O}(t^{-1-\mu}).
\end{align*}
Therefore one can deduce from Lemma \ref{lemSCW_1} together with the fact that $N$ is relatively $H$-bounded that
\begin{align*}
\chi( H ) \Big [ H_I , \mathrm{d}\Gamma_i\Big(F\Big(\frac{|x_i|}{t}\Big)\Big) \Big] \chi( H ) =  \mathcal{O}(t^{-1-\mu}).
\end{align*} 
Hence the condition \eqref{eq:condPO} is satisfied and it suffices to apply the abstract method recalled at the beginning of Section \ref{sec:propag}.
\end{proof}
\subsection{Massless neutrinos}\label{subsec:massless_propag}

In this section we prove the counterpart of Theorem \ref{th:propag_massive} in the case where the masses of the neutrinos vanish. As mentioned at the beginning of Section \ref{sec:propag}, for massless neutrinos we cannot directly adapt \cite{DeGe99_01,DeGe00_01}, hence some parts of the proof require substantial modifications. A first issue comes from the fact that, in order to control remainder terms in some commutator expansions, one needs to control the second quantization of expressions of the form $[ \mathrm{i} \nabla_{p_2} , [ \mathrm{i} \nabla_{p_2} , \omega^{(2)}_l(p_2) ] ]$, where $\omega^{(2)}_l$ are the dispersion relations of neutrinos. In the massive case, such commutators are bounded. Their second quantizations are therefore relatively $N$-bounded and, hence, relatively $H$-bounded. In the massless case, however, $[ \mathrm{i} \nabla_{p_2} , [ \mathrm{i} \nabla_{p_2} , \omega^{(2)}_l(p_2) ] ]$ is  of order $\mathcal{O}(|p_2|^{-1})$ near $0$. To overcome this difficulty, the idea is to show that $\langle e^{ - \mathrm{i} t H } u , \mathrm{d} \Gamma( |p_2|^{-1} ) e^{ - \mathrm{i} t H } u \rangle \lesssim t^\gamma$ for some $\gamma<1$ and for suitable states $u \in \mathscr{H}$. This idea was used in \cite{BoFaSi12_01,FaSi14_01,FaSi14_02} for Pauli-Fierz Hamiltonians, modifying an argument of \cite{Ge02_01}. In Lemma \ref{lem:evolp2-1}, we adapt \cite{BoFaSi12_01,FaSi14_01,FaSi14_02} to our context.

A second issue that has to be dealt with, in the massless case, to prove a suitable minimal velocity estimate, is that we have to use the Mourre estimate stated in Theorem \ref{MourreI_3}, with a positive term proportional to the number operator $N_{\mathrm{neut}}$. This is again required in order to be able to control some remainder terms. Some care must also be taken because of the fact that the considered conjugate operator (see \eqref{eq:conjugate_1}) is not self-adjoint. Propagation estimates analogous to those of Theorem \ref{th:propag_massive} are established in Theorem \ref{th:propag_massless}. An important difference with the massive case is that the propagation estimates hold only for states in a dense subset of the total Hilbert space.

In Theorem \ref{th:propag_massless_2}, we prove a second version of propagation estimates involving a time-dependent modified one-particle ``position'' operator. Theorem \ref{th:propag_massless_2} will be a crucial input in our proof of asymptotic completeness of the wave operators in Section \ref{sec:AC}.

For shortness, we set $\mathrm{d} \Gamma ( | p_2 |^{-\alpha} ) = \mathrm{d} \Gamma ( h^\alpha )$ where $h^\alpha=(h^\alpha_1,\dots,h^\alpha_{14})$, according to the notations of Appendix \ref{sec:standard_def}, with $h^\alpha_i=0$ if $i\in\{1,2,3,4,7,8,11,12\}$ (corresponding to the label of a massive particle), and $h^\alpha_i$ is the operator of multiplication by $|p_2|^{-\alpha}$ if $i \in \{5,6,9,10,13,14\}$ (corresponding to the label of a neutrino). The first part of the following lemma (see \eqref{eq:estimp2-1}) is adapted from \cite[Lemma 4.1]{BoFaSi12_01}. For technical reasons, that will appear in the proofs of Theorems \ref{th:propag_massless} and \ref{th:propag_massless_2} below, we also need a new, related estimate, see \eqref{eq:estimp2-1_a}.
\begin{lem}\label{lem:evolp2-1}
Suppose that the masses of the neutrinos $m_{\nu_e}$, $m_{\nu_{\mu}}$, $m_{\nu_{\tau}}$ vanish and consider the Hamiltonian \eqref{total_hamiltonian_g} with $H_I$ given by \eqref{interaction_term}. Let $0 \le \alpha \le 1$ and assume that
\begin{equation*}
G \in L^2 , \quad |p_2|^{-1-\mu} G \in L^2 \text{ for some } \mu > 0.
\end{equation*}
Let $\chi \in \mathrm{C}_0^\infty ( \mathbb{R} )$. There exists $\mathrm{C}>0$ such that, for all $u \in \mathscr{D}( \mathrm{d}\Gamma( |p_2|^{-\alpha} )^{\frac12} )$, 
\begin{equation}
\Big \langle e^{ - \mathrm{i} t H } \chi ( H ) u , \mathrm{d} \Gamma ( | p_2 |^{-\alpha} ) e^{ - \mathrm{i} t H } \chi( H ) u \Big \rangle \le \mathrm{C} t^{\frac{\alpha}{1+\mu}} \big ( \| \mathrm{d} \Gamma( |p_2|^{-\alpha} )^{\frac12} u \|^2 + \| N_{\mathrm{neut}}^{\frac12} u \|^2 + \| u \|^2 \big ) , \label{eq:estimp2-1}
\end{equation}
and for all $u \in \mathscr{D}( \mathrm{d}\Gamma( |p_2|^{-\alpha} ) )$, 
\begin{equation}
\Big \| \mathrm{d} \Gamma ( | p_2 |^{-\alpha} ) e^{ - \mathrm{i} t H } \chi( H ) u \Big \| \le \mathrm{C} t^{\frac{\alpha}{1+\mu}} \big ( \| \mathrm{d} \Gamma( |p_2|^{-\alpha} ) u \| + \| N_{\mathrm{neut}} u \| + \| u \| \big ) .  \label{eq:estimp2-1_a}
\end{equation}
\end{lem}
\begin{proof}
To prove \eqref{eq:estimp2-1}, we consider a function $f_0 \in \mathrm{C}^\infty( [0,\infty) ; \mathbb{R} )$ such that $f_0$ is decreasing, $f_0(r) = 1$ on $[ 0 , 1 ]$ and $f_0(r) = 0$ on $[ 2 , \infty )$. Let $f_\infty = 1 - f$. For $\nu>0$, we decompose
\begin{equation} 
\mathrm{d} \Gamma ( | p_2 |^{- \alpha} ) = \mathrm{d} \Gamma \big( | p_2 |^{-\alpha} f_0( t^{\nu} | p_2 | ) \big) +  \mathrm{d} \Gamma \big( | p_2 |^{-\alpha} f_\infty (t^{\nu} | p_2 | ) \big) , \label{eq:decompp-1}
\end{equation}
and insert this into the left-hand-side of \eqref{eq:estimp2-1}. Since $f_\infty$ is supported in $[1,\infty)$, the second term can be estimated as
\begin{align} 
\Big \langle e^{ - \mathrm{i} t H } \chi ( H ) u , \mathrm{d} \Gamma \big( | p_2 |^{-\alpha} f_\infty (t^{\nu} | p_2 | ) \big) e^{ - \mathrm{i} t H } \chi ( H ) u \Big \rangle &\leq t^{\alpha\nu} \Big \langle e^{ - \mathrm{i} t H } \chi ( H ) u , \mathrm{d} \Gamma \big( f_\infty ( t^{\nu} | p_2 | ) \big) e^{ - \mathrm{i} t H } \chi ( H ) u \Big \rangle \notag \\
&\leq t^{\alpha\nu} \Big \langle e^{ - \mathrm{i} t H } \chi ( H ) u , N_{\mathrm{neut}} e^{ - \mathrm{i} t H } \chi ( H ) u \Big \rangle \notag \\
&\lesssim t^{\alpha\nu} \big( \| N_{\mathrm{neut}}^{\frac12} u \|^2 + \| u \|^2 \big) . \label{eq:1st_p2-1}
\end{align}
The second equality comes from the facts that $N_{\mathrm{neut}}-N_{\mathrm{lept}}$ commutes with $H$ and that $N_{\mathrm{lept}}$ is relatively $H$-bounded. To estimate the evolution of the first term in \eqref{eq:decompp-1}, we differentiate
\begin{align*}
& \partial_t \Big \langle e^{ - \mathrm{i} t H } \chi ( H ) u , \mathrm{d} \Gamma \big( | p_2 |^{-\alpha} f_0 (t^{\nu} | p_2 | ) \big) e^{ - \mathrm{i} t H } \chi ( H ) u \Big \rangle \\
&= \nu t^{\nu-1} \Big \langle e^{ - \mathrm{i} t H } \chi ( H ) u , \mathrm{d} \Gamma \big( | p_2 |^{1-\alpha} f'_0 (t^{\nu} | p_2 | ) \big) e^{ - \mathrm{i} t H } \chi ( H ) u \Big \rangle \\
&\quad + \Big \langle e^{ - \mathrm{i} t H } \chi ( H ) u , \big [ H , \mathrm{i} \mathrm{d} \Gamma \big( | p_2 |^{-\alpha} f_0 (t^{\nu} | p_2 | ) \big) \big ] e^{ - \mathrm{i} t H } \chi ( H ) u \Big \rangle .
\end{align*}
Since $f_0' \le 0$ and since $H_0$ commutes with $\mathrm{d} \Gamma \big( | p_2 |^{-\alpha} f_0 (t^{\nu} | p_2 | ) \big)$, this implies that
\begin{align}
& \partial_t \Big \langle e^{ - \mathrm{i} t H } \chi ( H ) u , \mathrm{d} \Gamma \big( | p_2 |^{-\alpha} f_0 (t^{\nu} | p_2 | ) \big) e^{ - \mathrm{i} t H } \chi ( H ) u \Big \rangle \notag \\
&\le \Big \langle e^{ - \mathrm{i} t H } \chi ( H ) u , \big [ H_I(G) , \mathrm{i} \mathrm{d} \Gamma \big( | p_2 |^{-\alpha} f_0 (t^{\nu} | p_2 | ) \big) \big ] e^{ - \mathrm{i} t H } \chi ( H ) u \Big \rangle . \label{eq:estimp2-1_1}
\end{align}
Similarly as in \eqref{commut_1}--\eqref{commut_4}, a direct computation gives
\begin{align*}
& \big [H_I^{(1)}(G) , \mathrm{i} \mathrm{d} \Gamma \big( | p_2 |^{-\alpha} f_0 (t^{\nu} | p_2 | ) \big ) \big ] = - H_I^{(1)}( \mathrm{i}  | p_2 |^{-\alpha} f_0 (t^{\nu} | p_2 | ) G ),  \\
& \big [H_I^{(2)}(G) , \mathrm{i} \mathrm{d} \Gamma \big( | p_2 |^{-\alpha} f_0 (t^{\nu} | p_2 | ) \big ) \big ] = - H_I^{(2)}( \mathrm{i} | p_2 |^{-\alpha} f_0 (t^{\nu} | p_2 | ) G ),  \\
& \big [H_I^{(3)}(G) , \mathrm{i} \mathrm{d} \Gamma \big( | p_2 |^{-\alpha} f_0 (t^{\nu} | p_2 | ) \big ) \big ] = H_I^{(3)}( \mathrm{i} | p_2 |^{-\alpha} f_0 (t^{\nu} | p_2 | ) G ),  \\
& \big [H_I^{(4)}(G) , \mathrm{i} \mathrm{d} \Gamma \big( | p_2 |^{-\alpha} f_0 (t^{\nu} | p_2 | ) \big ) \big ] = H_I^{(4)}( \mathrm{i} | p_2 |^{-\alpha} f_0 (t^{\nu} | p_2 | ) G ). 
\end{align*}
By Lemma \ref{lemSCW_1}, the support of $f_0$ and the assumption that $|p_2|^{-1-\mu} G \in L^2$, we obtain that
\begin{align}
\big \| \big [ H_I(G) , \mathrm{i} \mathrm{d} \Gamma \big( | p_2 |^{-\alpha} f_0 (t^{\nu} | p_2 | ) \big) \big ] (N_{\mathrm{lept}} + N_W + 1)^{-1} \big \| & \lesssim \| (1 + | p_2 |^{-\alpha} f_0 (t^{\nu} | p_2 | ) ) G \|_2 \notag \\
& \lesssim \| (1 + | p_2 |^{1+\mu-\alpha} f_0 (t^{\nu} | p_2 | )  |p_2|^{-1-\mu}) G \|_2 \notag \\
& \lesssim t^{-(1+\mu-\alpha) \nu } \| (1 + |p_2|^{-(1+\mu)} ) G \|_2. \label{eq:ruqb1}
\end{align}
Since $N_{\mathrm{lept}} + N_W$ is relatively $H$-bounded, integrating \eqref{eq:estimp2-1_1} shows that
\begin{align}
& \Big \langle e^{ - \mathrm{i} t H } \chi ( H ) u , \mathrm{d} \Gamma \big( | p_2 |^{-\alpha} f_0 (t^{\nu} | p_2 | ) \big) e^{ - \mathrm{i} t H } \chi ( H ) u \Big \rangle \lesssim t^{1-(1+\mu-\alpha)\nu} \big ( \| \mathrm{d} \Gamma( |p_2|^{-\alpha} )^{\frac12} u \|^2 + \| u \|^2 \big ). \label{eq:1st_p2-2}
\end{align}
Equations \eqref{eq:1st_p2-1} and \eqref{eq:1st_p2-2} yield
\begin{align*}
& \Big \langle e^{ - \mathrm{i} t H } \chi ( H ) u , \mathrm{d} \Gamma \big( | p_2 |^{-\alpha} \big) e^{ - \mathrm{i} t H } \chi ( H ) u \Big \rangle \lesssim \max( t^{1-(1+\mu-\alpha)\nu} , t^{\alpha \nu} ) \big ( \| \mathrm{d} \Gamma( |p_2|^{-1} )^{\frac12} u \|^2 + \| N_{\mathrm{neut}}^{\frac12} u \|^2 + \| u \|^2 \big ). 
\end{align*}
Choosing $\nu =( 1+\mu)^{-1}$ concludes the proof of \eqref{eq:estimp2-1}.

To establish \eqref{eq:estimp2-1_a}, we modify the proof as follows. On one hand, we have that
\begin{equation}
\Big\| \mathrm{d} \Gamma \big( | p_2 |^{-\alpha} f_\infty (t^{\nu} | p_2 | ) \big) e^{ - \mathrm{i} t H } \chi ( H ) u  \Big\| \leq t^{-\nu\alpha} \Big\| N_{\mathrm{neut}} e^{ - \mathrm{i} t H } \chi ( H ) u  \Big\| \lesssim  t^{-\nu\alpha} \Big( \big\| N_{\mathrm{neut}}  u  \big\| +\| u  \|\Big) , \label{eq:jske1}
\end{equation}
and on the other hand
\begin{align}
& \Big\| \mathrm{d} \Gamma \big( | p_2 |^{-\alpha} f_0 (t^{\nu} | p_2 | ) \big) e^{ - \mathrm{i} t H } \chi ( H ) u  \Big\| \notag \\
&=\Big\| e^{\mathrm{i}t H } \mathrm{d} \Gamma \big( | p_2 |^{-\alpha} f_0 (t^{\nu} | p_2 | ) \big) e^{ - \mathrm{i} t H } \chi ( H ) u  \Big\| \notag \\
&\le \Big\| \mathrm{d} \Gamma \big( | p_2 |^{-\alpha} f_0 ( | p_2 | ) \big) e^{ - \mathrm{i} H } \chi ( H ) u  \Big\| + \int_1^t \Big\| e^{\mathrm{i}s H } \Big [ \mathrm{d} \Gamma \big( | p_2 |^{-\alpha} f_0 (s^{\nu} | p_2 | ) \big) , \mathrm{i} H \Big ] e^{ - \mathrm{i} s H } \chi ( H ) u  \Big\| ds \notag \\
& \lesssim \big\| \mathrm{d} \Gamma \big( | p_2 |^{-\alpha} \big) u \big\|+\|u\| + \int_1^t s^{-(1+\mu-\alpha) \nu } \big \| (1 + |p_2|^{-(1+\mu)} ) G \big \|_2 \| u \| ds , \label{eq:jske2}
\end{align}
where we used \eqref{eq:ruqb1} and the fact that $N_{\mathrm{lept}} + N_W$ is relatively $H$-bounded in the last inequality. Combining \eqref{eq:jske1} and \eqref{eq:jske2}, we deduce \eqref{eq:estimp2-1_a} similarly as above.
\end{proof}
Now we are ready to prove a first version of the propagation estimates in the massless case.  The general strategy is the same as in Theorem \ref{th:propag_massive}, with some important technical modifications. In particular, we have to use Lemma \ref{lem:evolp2-1}, the particular form of the Mourre estimate stated in Theorem \ref{MourreI_3} and the fact that $N_{\mathrm{lept}} - N_{\mathrm{neut}}$ commutes with $H$.

Recall that the notation $\mathds{1}_{[R,R']}(|x|)$ has been introduced in \eqref{eq:1Rx}.
\begin{Th}\label{th:propag_massless}
Suppose that the masses of the neutrinos $m_{\nu_e}$, $m_{\nu_{\mu}}$, $m_{\nu_{\tau}}$ vanish and consider the Hamiltonian $H = H_0 + g ( H_I^{(1)} + H_I^{(2)} )$ with $H_I^{(1)}$ and $H_I^{(2)}$ given by \eqref{interaction_term}. Assume that
\begin{equation*}
G \in L^2, \quad a_{(i),\cdot} G \in L^2, \quad |p_3|^{-1} a_{(i),\cdot} G \in L^2, \quad i = 1,2,3 ,
\end{equation*}
and that
\begin{equation*}
G \in \mathbb{H}^{1+\mu} \text{ for some } \mu > 0.
\end{equation*}
\begin{itemize}
\item[(i)] Let $\chi \in \mathrm{C}_0^{\infty}(\mathbb{R})$, $R=(R_1,\dots,R_{14})$ and $R'=(R'_1,\dots,R'_{14})$ be such that $R'_{i}>R_i>1$. There exists $\mathrm{C} > 0$ such that, for all $u \in \mathscr{D}(N_{\mathrm{neut}}^{\frac12})\cap\mathscr{D}( \mathrm{d}\Gamma( |p_2|^{-1} )^{\frac12} )$, 
\begin{align*}
& \int^{\infty}_{1} \Big \| \mathrm{d} \Gamma \Big ( \mathds{1}_{[R,R']} \Big (\frac{|x|}{t} \Big ) \Big)^{\frac{1}{2}}\chi(H) e^{-\mathrm{i}tH} u \Big \|^2 \frac{dt}{t} \leq \mathrm{C} \big ( \| N_{\mathrm{neut}}^{\frac12} u \|^2 + \| \mathrm{d}\Gamma( |p_2|^{-1} )^{\frac12} u \|^2 + \|u\|^{2} \big ).
\end{align*}
\item[(ii)] Let $0<\mathrm{v}_0<\mathrm{v}_1$ and $\chi \in \mathrm{C}_0^{\infty}(\mathbb{R})$. There exists $\mathrm{C}>0$ such that, for all $u \in \mathscr{D}(N_{\mathrm{neut}}^{\frac12})\cap\mathscr{D}( \mathrm{d}\Gamma( |p_2|^{-1} )^{\frac12} )$,
\begin{align*}
& \int^{\infty}_{1} \Big \| \mathrm{d}\Gamma \Big( \Big\langle \big( \frac{x}{t} - \nabla \omega \big ) , \mathds{1}_{ [ \mathrm{v}_0 , \mathrm{v}_1 ] }( \frac{x}{t} ) \big( \frac{x}{t} - \nabla \omega \big ) \Big\rangle \Big)^{\frac12} \chi(H)e^{-\mathrm{i} t H} u \Big\|^{2} \frac{dt}{t} \\
& \le \mathrm{C} \big ( \| N_{\mathrm{neut}}^{\frac12} u \|^2 + \| \mathrm{d}\Gamma( |p_2|^{-1} )^{\frac12} u \|^2 + \|u\|^{2} \big ).
\end{align*}
\item[(iii)] Let $0< \mathrm{v}_0<\mathrm{v}_1$, $J \in \mathrm{C}_0^{\infty}( \{ x \in \mathbb{R}^3 , \mathrm{v}_0 < |x| < \mathrm{v}_1\})$ and $\chi \in \mathrm{C}_0^{\infty}(\mathbb{R})$. There exists $\mathrm{C}>0$ such that, for all $\ell \in \{ 1,2,3\} $ and $u \in \mathscr{D}(N_{\mathrm{neut}}^{\frac12})\cap\mathscr{D}( \mathrm{d}\Gamma( |p_2|^{-1} )^{\frac12} )$,
\begin{align*}
& \int^{\infty}_{1} \Big \| \mathrm{d}\Gamma \Big( \Big| J\Big(\frac{x}{t}\Big) \Big( \frac{x^{(\ell)}}{t} - \partial_{(\ell)} \omega \Big) + \mathrm{h.c }\Big|\Big)^{\frac{1}{2}}\chi(H)e^{-\mathrm{i} t H} u \Big\|^{2} \frac{dt}{t} \\
& \le \mathrm{C} \big ( \| N_{\mathrm{neut}}^{\frac12} u \|^2 + \| \mathrm{d}\Gamma( |p_2|^{-1} )^{\frac12} u \|^2 + \|u\|^{2} \big ).
\end{align*}
\item[(iv)] There exists $g_0>0$ such that, for all $|g|\le g_0$, the following holds: let $\chi \in \mathrm{C}_{0}^{\infty}(\mathbb{R})$ be supported in $\mathbb{R}\backslash(\tau \cup \sigma_{\mathrm{pp}}(H))$. There exist $\delta>0$ and $\mathrm{C} > 0$ such that
\begin{equation*}
\int^{\infty}_{1} \Big\| \Gamma \Big ( \mathds{1}_{[0,\delta]} \Big(\frac{|x|}{t}\Big)\Big) \chi(H) e^{-\mathrm{i} t H} u \Big\|^2 \frac{dt}{t} \le \mathrm{C} \big\| ( \mathrm{d}\Gamma(|p_2|^{-1}) + N_{\mathrm{neut}} + 1 )^{\frac12} (N_{\mathrm{neut}}+1)^{\frac32} u \big\|^2 ,
 \end{equation*}
 for all $u \in \mathscr{D}( ( \mathrm{d}\Gamma(|p_2|^{-1}) + N_{\mathrm{neut}} + 1 )^{\frac12} (N_{\mathrm{neut}}+1)^{\frac32} )$.
\end{itemize}
\end{Th}
\begin{proof}
We prove (i) and (iv), underlying the differences with \cite{DeGe99_01}. The proofs of (ii) and (iii) can be deduced by adapting \cite[Section 6]{DeGe99_01}, using furthermore arguments similar to those used to prove (i).

(i) Let $F_i \in \mathrm{C}^\infty ( \mathbb{R} ; \mathbb{R} )$ be a non-decreasing function such that $F_i=0$ near $0$, $F_i=\mathrm{const}$ near $+\infty$ and $F'_i \ge \mathds{1}_{[R_i,R'_i]}$. As in the proof of Theorem \ref{th:propag_massive}, \eqref{eq:hdt1} holds and therefore it suffices to prove that
\begin{equation*}
\int^{\infty}_{1} \Big\|\mathrm{d}\Gamma_i\Big(\mathds{1}_{[R_i,R'_i]}\Big(\frac{ |x_i| }{t}\Big)\Big)^{\frac{1}{2}}\chi(H) e^{-\mathrm{i}tH} u \Big\|^2 \frac{dt}{t} \lesssim \big ( \| N_{\mathrm{neut}}^{\frac12} u \|^2 + \| \mathrm{d}\Gamma( |p_2|^{-1} )^{\frac12} u \|^2 + \|u\|^{2} \big )
\end{equation*}
for all $u \in \mathscr{D}(N_{\mathrm{neut}}^{\frac12})\cap\mathscr{D}( \mathrm{d}\Gamma( |p_2|^{-1} )^{\frac12} )$ and $i \in \{1,\dots,14\}$, where $\mathrm{d}\Gamma_i$ is defined in \eqref{eq:defdgammaj}.
Let
\begin{align}
f_i(t) := \Big \langle e^{ - \mathrm{i} t H } \chi ( H ) u , \mathrm{d}\Gamma_{i} \Big( F_i \Big(\frac{ |x_i| }{t}\Big)\Big) e^{ - \mathrm{i} t H } \chi(H) u \Big \rangle . \label{eq:def_fj}
\end{align}
Note that, since $\mathrm{d}\Gamma_{i} ( F_i (\frac{ |x_i| }{t})) \le \mathrm{C} N$, for some positive constant $\mathrm{C}$, since $N_{\mathrm{lept}}$ and $N_W$ are relatively $H$-bounded and since $N_{\mathrm{lept}} - N_{\mathrm{neut}}$ commutes with $H$, we can write
\begin{align*}
f_i(t) & \le \Big \langle e^{ - \mathrm{i} t H } \chi ( H ) u , \mathrm{C} N e^{ - \mathrm{i} t H } \chi(H) u \Big \rangle \\
& = \mathrm{C} \Big \langle e^{ - \mathrm{i} t H } \chi ( H ) u , \big ( 2 N_{\mathrm{lept}} + N_W + N_{\mathrm{neut}} - N_{\mathrm{lept}} \big ) e^{ - \mathrm{i} t H } \chi(H) u \Big \rangle \\
& \lesssim \| u \|^2 +  \Big \langle \big ( N_{\mathrm{neut}} - N_{\mathrm{lept}} \big )^{\frac12} u , \chi(H)^2 \big ( N_{\mathrm{neut}} - N_{\mathrm{lept}} \big )^{\frac12} u \Big \rangle \\
& \lesssim \| N_{\mathrm{neut}}^{\frac12} u \|^2 + \| u \|^2.
\end{align*}
Hence $f_i(t)$ is uniformly bounded.

Differentiating $f_i$ gives
\begin{align}
\partial_t f_i(t) &= \Big \langle e^{ - \mathrm{i} t H } \chi ( H ) u , \mathbf{D} \mathrm{d}\Gamma_{i} \Big( F_i \Big( \frac{ |x_i| }{t}\Big)\Big) e^{ - \mathrm{i} t H } \chi(H) u \Big \rangle \notag \\
&= \Big \langle e^{ - \mathrm{i} t H } \chi ( H ) u , \mathrm{d}\Gamma_{i} \Big( \mathbf{d}_0 F_i \Big(\frac{ |x_i| }{t}\Big)\Big) e^{ - \mathrm{i} t H } \chi(H) u \Big \rangle \notag \\
&\quad + \Big \langle e^{ - \mathrm{i} t H } \chi ( H ) u , \Big [ H_I^{(1)}(G) + H_I^{(2)}(G) , \mathrm{i} \mathrm{d}\Gamma_{i} \Big( F_i \Big(\frac{|x_i| }{t}\Big) \Big) \Big ] e^{ - \mathrm{i} t H } \chi(H) u \Big \rangle , \label{eq:propagestim1}
\end{align}
where $\mathbf{D}$ and $\mathbf{d}_0$ denotes the Heisenberg derivatives defined at the beginning of Section \ref{sec:propag}.

Using that $F'_i \ge \mathds{1}_{[R_i,R'_i]}$ and $R_i>1$, a commutator expansion at second order (see, e.g., \cite{DeGe97_01} and \cite[Lemma 5.2]{BoFaSi12_01}) gives
\begin{align*}
\mathbf{d}_0 F_i \Big(\frac{|x_i|}{t}\Big) \le  \frac{-\mathrm{C}_0}{t}\mathds{1}_{[R_i,R'_i]} \Big(\frac{|x_i|}{t}\Big) + \mathcal{O}(t^{-2}) , 
\end{align*}
if $i\in\{1,2,3,4,7,8,11,12\}$ (corresponding to the label of a massive particle), and
\begin{align*}
\mathbf{d}_0 F_i \Big(\frac{|x_i|}{t}\Big) \le  \frac{-\mathrm{C}_0}{t} \mathds{1}_{[R_i,R'_i]} \Big(\frac{|x_i|}{t}\Big) + \mathcal{O}(t^{-2}) |p_2|^{-1} , 
\end{align*}
if $i \in \{5,6,9,10,13,14\}$ (corresponding to the label of a neutrino). Here $\mathrm{C}_0>0$. In the first case, using as above that $N_{\mathrm{lept}}-N_{\mathrm{neut}}$ commutes with $H$, this gives
\begin{align}
& \Big\langle e^{-\mathrm{i}tH}\chi(H)u,\mathrm{d}\Gamma_{i}\Big(\mathbf{d}_0F_i\Big(\frac{|x_i|}{t}\Big)\Big)e^{-\mathrm{i}tH}\chi(H)u\Big\rangle \notag \\
&\le - \frac{ \mathrm{C}_0 }{ t } \Big \langle e^{ - \mathrm{i} t H } \chi ( H ) u , \mathrm{d}\Gamma_{i} \Big ( \mathds{1}_{[R_i,R'_i]} \Big(\frac{|x_i|}{t}\Big) \Big ) e^{ - \mathrm{i} t H } \chi(H) u \Big \rangle + \mathcal{O}(t^{-2}) \Big \langle e^{ - \mathrm{i} t H } \chi ( H ) u , N e^{ - \mathrm{i} t H } \chi(H) u \Big \rangle \notag \\
&\le - \frac{ \mathrm{C}_0 }{ t } \Big \langle e^{ - \mathrm{i} t H } \chi ( H ) u , \mathrm{d}\Gamma_{i} \Big ( \mathds{1}_{[R_i,R'_i]} \Big(\frac{|x_i|}{t}\Big) \Big ) e^{ - \mathrm{i} t H } \chi(H) u \Big \rangle + \mathcal{O}(t^{-2}) \big ( \| N_{\mathrm{neut}}^{\frac12} u \|^2 + \| u \|^2 \big ) , \label{eq:propagestim2}
\end{align}
if $i\in\{1,2,3,4,7,8,11,12\}$. In the second case (if $i$ corresponds to the label of a neutrino), we use in addition Lemma \ref{lem:evolp2-1}, yielding
\begin{align}
& \Big\langle e^{-\mathrm{i}tH}\chi(H)u,\mathrm{d}\Gamma_{i}\Big(\mathbf{d}_0F_i\Big(\frac{|x_i|}{t}\Big)\Big)e^{-\mathrm{i}tH}\chi(H)u\Big\rangle \notag \\
&\le - \frac{ \mathrm{C}_0 }{ t } \Big \langle e^{ - \mathrm{i} t H } \chi ( H ) u , \mathrm{d}\Gamma_{i} \Big ( \mathds{1}_{[R_i,R'_i]} \Big(\frac{|x_i|}{t}\Big) \Big ) e^{ - \mathrm{i} t H } \chi(H) u \Big \rangle + \mathcal{O}(t^{-2}) \Big \langle e^{ - \mathrm{i} t H } \chi ( H ) u , \mathrm{d}\Gamma(|p_2|^{-1}) e^{ - \mathrm{i} t H } \chi(H) u \Big \rangle \notag \\
&\le - \frac{ \mathrm{C}_0 }{ t } \Big \langle e^{ - \mathrm{i} t H } \chi ( H ) u , \mathrm{d}\Gamma_{i} \Big ( \mathds{1}_{[R_i,R'_i]} \Big(\frac{|x_i|}{t}\Big) \Big ) e^{ - \mathrm{i} t H } \chi(H) u \Big \rangle + \mathcal{O}(t^{-2 + \frac{2}{2+\mu} }) \big (  \| \mathrm{d} \Gamma( |p_2|^{-1} )^{\frac12} u \|^2 + \| u \|^2 \big ) . \label{eq:propagestim3}
\end{align}
Note that Hardy's inequality implies that, if $G \in \mathbb{H}^{1+\mu}$ for some $\mu>0$, then $|p_2|^{-1-\mu}G \in L^2$ (if $\mu<1/2$), and hence Lemma \ref{lem:evolp2-1} can indeed by applied.

The commutators $ [ H_I^{(j)}(G) , \mathrm{d}\Gamma_{i} ( F_i (\frac{|x_i|}{t}) )  ]$, $j=1,2$, are given by expressions similar to \eqref{commut_11}--\eqref{commut_21}, with the operator $F_i (\frac{|x_i|}{t})$ instead of $a_i$. Since $F_i$ vanishes near $0$ and $G \in \mathbb{H}^{1+\mu}$, we have that
\begin{align*}
\Big\|F_i\Big(\frac{|x_i|}{t}\Big)G\Big\|_2=\mathcal{O}(t^{-1-\mu}).
\end{align*}
Therefore we deduce from Lemma \ref{lemSCW_1} that
\begin{align}
& \Big \langle e^{ - \mathrm{i} t H } \chi ( H ) u , \Big [ H_I^{(1)}(G) + H_I^{(2)}(G) , \mathrm{i} \mathrm{d}\Gamma_{i} \Big( F_i \Big(\frac{|x_i|}{t}\Big) \Big) \Big ] e^{ - \mathrm{i} t H } \chi(H) u \Big \rangle \notag \\
& \le \mathrm{C} t^{-1-\mu} \Big \langle e^{ - \mathrm{i} t H } \chi ( H ) u , ( N_W + N_{\mathrm{lept}} ) e^{ - \mathrm{i} t H } \chi(H) u \Big \rangle = \mathcal{O}( t^{-1-\mu} ) \| u \|^2. \label{eq:propagestim4}
\end{align}

Integrating \eqref{eq:propagestim1} over $[1,\infty)$, using that $t^{-2+\frac{2}{2+\mu}}$ and $t^{-1-\mu}$ are integrable, we obtain the statement of (i) by combining \eqref{eq:propagestim2}--\eqref{eq:propagestim3} and \eqref{eq:propagestim4}.

\vspace{0,2cm}

(iv) Let $\lambda \in \mathbb{R}\setminus\{ \tau \cup \sigma_{\mathrm{pp}}(H) \}$. Clearly, since $N_{\mathrm{lept}}-N_{\mathrm{neut}}$ commutes with $H$, it suffices to prove (iv) for all $u \in \mathrm{Ran}( \mathds{1}_{\{n\}}(N_{\mathrm{lept}}-N_{\mathrm{neut}}))$, $n\in\mathbb{Z}$. Hence we fix $n\in\mathbb{Z}$ and  $u \in \mathrm{Ran}( \mathds{1}_{\{n\}}(N_{\mathrm{lept}}-N_{\mathrm{neut}}))$. Recall that, by Theorem \ref{MourreI_3}, there exist $\varepsilon>0$, $\mathrm{c}_0>0$ and $\mathrm{C}_\lambda > 0$ such that
\begin{equation*}
[H,\mathrm{i} A ] \geq \mathrm{c}_0 ( N_{\mathrm{neut}} + \mathds{1} ) - \mathrm{C}_\lambda ( \mathds{1} - \mathds{1}_{[\lambda-\varepsilon,\lambda+\varepsilon]}(H) ) ( 1 + H^2 )^{\frac12} . 
\end{equation*}
For $\tilde\chi \in \mathrm{C}_0^\infty( ( \lambda-\varepsilon,\lambda+\varepsilon) )$, this implies that
\begin{equation}
[H,\mathrm{i} A ] \geq \mathrm{c}_0 ( N_{\mathrm{neut}} + \mathds{1} ) - \mathrm{C}_\lambda ( \mathds{1} - \tilde{\chi}^2(H) ) ( 1 + H^2 )^{\frac12} . \label{eq:Mourre_for_propag}
\end{equation}

Let $\chi \in \mathrm{C}_0^\infty( ( \lambda-\varepsilon,\lambda+\varepsilon) )$ be such that $\chi \tilde\chi = \chi$. Let $\delta > 0 $ and $q_i$, $i \in \{1, \dots, 14\}$ be functions in $\mathrm{C}_0^{\infty} (\{ x \in \mathbb{R}^3 , |x| \le 2 \delta \})$ such that $0 \le q_i \le 1$, $q_i=1$ on $\{ x \in \mathbb{R}^3 , |x| \leq \delta \}$ and let $q^t=(q_1(\frac{\mathrm{i}\nabla}{t}),\dots,q_{14}(\frac{\mathrm{i}\nabla}{t}))$. Let
\begin{equation*}
h(t):=\Big\langle e^{-\mathrm{i}tH}\chi(H)u,\Gamma(q^t)\frac{A}{t}\Gamma(q^t)e^{-\mathrm{i}tH}\chi(H)u\Big\rangle.
\end{equation*}
Since $\nabla \omega^{(i)}_l$, $i=1,2$, $l=1,2,3$ and $\nabla \omega^{(3)}$ are bounded, it is not difficult to observe that
\begin{equation}
\label{Step1}
\pm \Gamma(q^t) \frac{A}{t} \Gamma(q^t) \le \mathrm{C} \delta (N+1) \quad \text{ and } \quad \big \| \frac{A}{t} \Gamma(q^t) ( N+1 )^{-1}\big\|\le \mathrm{C} \delta.
\end{equation}
Here it should be noticed that $\Gamma(q^t)$ maps $\mathscr{H}$ to $\mathscr{D}(A)$. Indeed a vector in $\mathrm{Ran}(\Gamma(q^t))$ belongs to the Sobolev space $\mathbb{H}^s(\mathbb{R}^3)$, for any $s>0$, as a function of any of the momentum variable $p_i$. It is then regular in a neighborhood of zero and using, in particular, \eqref{eq:domain_conjugate}, it is not difficult to verify that $\mathrm{Ran}( \Gamma(q^t) ) \subset \mathscr{D}(A)$. In the second estimate of \eqref{Step1}, Hardy's inequality in $\mathbb{R}^3$, $\| |p_2|^{-1} u \|_2 \lesssim \| \nabla_{p_2} u \|_2$, is used. 

From \eqref{Step1} and using as before that $N_{\mathrm{lept}} - N_{\mathrm{neut}}$ commutes with $H$, we obtain that
\begin{align}
| h(t) | \le  \Big\langle e^{-\mathrm{i}tH}\chi(H)u,\Gamma(q^t)\frac{A}{t}\Gamma(q^t)e^{-\mathrm{i}tH}\chi(H)u\Big\rangle \lesssim \|N_{\mathrm{neut}}^{\frac12}u\|^2 + \|u\|^2. \label{eq:estimh(t)}
\end{align}
In particular, $h(t)$ is uniformly bounded. Furthermore, one can compute
\begin{align}
\partial_t h(t) &= \Big\langle e^{-\mathrm{i}tH} u, \Big ( \chi(H) \mathrm{d}\Gamma(q^t, \mathbf{d}_0 q^t) \frac{A}{t}\Gamma(q^t) \chi(H) + \mathrm{h.c.} \Big ) e^{-\mathrm{i}tH}u\Big\rangle \notag \\
& \quad + \Big\langle e^{-\mathrm{i}tH}u, \Big ( \chi(H) [ H_I^{(1)} (G) + H_I^{(2)} (G)  , \mathrm{i} \Gamma(q^t) ] \frac{A}{t}\Gamma(q^t) \chi(H) + \mathrm{h.c.} \Big ) e^{-\mathrm{i}tH}u \Big\rangle \notag \\
& \quad + t^{-1} \Big\langle e^{-\mathrm{i}tH}u, \chi(H) \Gamma(q^t) [ H , \mathrm{i}A ] \Gamma(q^t) \chi(H) e^{-\mathrm{i}tH}u\Big\rangle \notag \\ 
& \quad - t^{-1} \Big\langle e^{-\mathrm{i}tH}u, \chi(H) \Gamma(q^t)\frac{A}{t}\Gamma(q^t) \chi(H) e^{-\mathrm{i}tH}u\Big\rangle \notag \\
& =: R_1(t) + R_2(t) + R_3(t) + R_4(t). \label{eq:R1-R4}
\end{align}
Observe that $A$ being symmetric, $A^*$ is an extension of $A$ so that the second terms in $R_2(t)$ and $R_3(t)$ are indeed the hermitian conjugates of the first ones. In what follows we consider each term $R_\sharp(t)$ separately.

We begin with $R_2(t)$. From the assumption that $G\in\mathbb{H}^{1+\mu}$, the commutation relations of Appendix \ref{sec:standard_def} and Lemma \ref{lemSCW_1}, one can show that
\begin{align*}
\big \| [H_I^{(1)} (G) + H_I^{(2)} (G) , \Gamma(q^t)] (N_{\mathrm{lept}}+N_W)^{-1} \big \|  =  \mathcal{O}(t^{-1-\mu}) .
\end{align*}
Together with \eqref{Step1} this implies that
\begin{equation*}
|R_2(t)|\lesssim \mathcal{O}(t^{-1-\mu})\big\|(N+1)\chi(H)e^{-\mathrm{i}tH}u\big\|^2,
\end{equation*}
since $N_{\mathrm{lept}}+N_W$ is relatively $H$-bounded. Since $N_{\mathrm{lept}}-N_{\mathrm{neut}}$ commutes with $H$, we deduce as before that
\begin{equation}
|R_2(t)|\lesssim \mathcal{O}(t^{-1-\mu})\|(N_{\mathrm{neut}}+1)u\|^2. \label{eq:estim_R2}
\end{equation}

 To evaluate $R_1(t)$, we compute, by means of a commutator expansion (see \cite{DeGe97_01} and \cite[Section 5]{BoFaSi12_01}),
 \begin{align}
\mathbf{d}_0 q_i^t(x_i) =  -\frac{1}{2 t} \Big \langle \frac{x_i}{t} - \nabla\omega_i ,  \nabla q_i \big( \frac{x_i}{t} \big) \Big \rangle + \mathrm{h.c.}   + \mathcal{O}(t^{-2}) , \label{eq:comput_rem_0}
\end{align}
if $i\in\{1,2,3,4,7,8,11,12\}$ (corresponding to the label of a massive particle), and
\begin{align}
\mathbf{d}_0 q_i^t(x_i) =  -\frac{1}{2 t} \Big \langle \frac{x_i}{t} - \nabla\omega_i ,  \nabla q_i \big( \frac{x_i}{t} \big) \Big \rangle + \mathrm{h.c.}   + \mathcal{O}(t^{-\frac32}) |p_2|^{-\frac12} , \label{eq:comput_rem}
\end{align}
if $i \in \{5,6,9,10,13,14\}$ (corresponding to the label of a neutrino). Here it should be noted that that the remainder term in \eqref{eq:comput_rem} can be computed to be of order $|p_2|^{-1+\gamma} \mathcal{O}(t^{-2+\gamma})$ for any $0\le\gamma\le1$ (see \cite[Lemma 5.2]{BoFaSi12_01}.  Here we choose $\gamma=\frac12$ for convenience, see below.

Let $g^t_i :=  -\frac{1}{2} \langle \frac{x_i}{t} - \nabla\omega_i ,  \nabla q_i \big( \frac{x_i}{t} \big)  \rangle + \mathrm{h.c.}$ and let $r^t_i$ be the remainder in \eqref{eq:comput_rem_0} or \eqref{eq:comput_rem}. For $i\in\{1,2,3,4,7,8,11,12\}$, by \eqref{Step1}, we deduce that
\begin{align}
\Big|\Big\langle e^{-\mathrm{i}tH} u, \Big ( \chi(H) \mathrm{d}\Gamma(q^t_i, r^t_i) \frac{A}{t}\Gamma(q^t) \chi(H) + \mathrm{h.c.} \Big ) e^{-\mathrm{i}tH}u\Big\rangle\Big| & \lesssim \mathcal{O}(t^{-2})\big\|N\chi(H)e^{-\mathrm{i}tH}u\big\|^2 \notag \\
& \lesssim \mathcal{O}(t^{-2}) \|(N_{\mathrm{neut}}+1)u\|^2 , \label{eq:dd0}
\end{align}
where we used again that $N_{\mathrm{lept}}-N_{\mathrm{neut}}$ commutes with $H$ in the second inequality. The case $i \in \{5,6,9,10,13,14\}$ corresponding to massless particles is more difficult. Using \eqref{eq:comput_rem} and Lemma \ref{estimationN2}, we write
\begin{align*}
& \Big|\Big\langle e^{-\mathrm{i}tH} u, \Big ( \chi(H) \mathrm{d}\Gamma(q^t_i, r^t_i) \frac{A}{t}\Gamma(q^t) \chi(H) + \mathrm{h.c.} \Big ) e^{-\mathrm{i}tH}u\Big\rangle\Big| \\
& \Big|\Big\langle e^{-\mathrm{i}tH} u, \Big ( \chi(H) \mathrm{d}\Gamma(q^t_i , r^t_i) (N+1)^{-\frac12}(N+1)^{\frac12}\frac{A}{t}\Gamma(q^t) \chi(H) + \mathrm{h.c.} \Big ) e^{-\mathrm{i}tH}u\Big\rangle\Big| \\
& \lesssim \mathcal{O}(t^{-\frac32}) \Big \| \mathrm{d}\Gamma(|p_2|^{-1})^{\frac12} \chi(H) e^{-\mathrm{i}tH} u \Big \| \big \| (N+1)^{\frac32} \chi(H) e^{-\mathrm{i}tH} u \big \| .
\end{align*}
Applying Lemma \ref{lem:evolp2-1}, and using that $N_{\mathrm{lept}}-N_{\mathrm{neut}}$ commutes with $H$, it thus follows that
\begin{align}
& \Big|\Big\langle e^{-\mathrm{i}tH} u, \Big ( \chi(H) \mathrm{d}\Gamma(q^t_i , r^t_i) \frac{A}{t}\Gamma(q^t) \chi(H) + \mathrm{h.c.} \Big ) e^{-\mathrm{i}tH}u\Big\rangle\Big| \notag \\
& \lesssim \mathcal{O}(t^{-\frac32+\frac{1}{2+\mu}}) \big ( \| \mathrm{d}\Gamma(|p_2|^{-1})^{\frac12} u \| + \| N_{\mathrm{neut}}^{\frac12}u \| + \|u\| \big ) \big ( \big \| (N+1)^{\frac32} u \big \| +  \| u \| \big ). \label{eq:dd1}
\end{align}
Summing over $i$, we obtain from \eqref{eq:dd0} and \eqref{eq:dd1} that
\begin{align}
& \Big|\Big\langle e^{-\mathrm{i}tH} u, \Big ( \chi(H) \mathrm{d}\Gamma(q^t , r^t) \frac{A}{t}\Gamma(q^t) \chi(H) + \mathrm{h.c.} \Big ) e^{-\mathrm{i}tH}u\Big\rangle\Big| \notag \notag \\
& \lesssim \mathcal{O}(t^{-2})\|(N_{\mathrm{neut}}+1)u\|^2 +\mathcal{O}(t^{-\frac32+\frac{1}{2+\mu}}) \big ( \| \mathrm{d}\Gamma(|p_2|^{-1})^{\frac12} u \| + \|u\| \big ) \big ( \big \| (N+1)^{\frac32} u \big \| + \| u \| \big ). \label{eq:estim_rt}
\end{align}

To estimate the term corresponding to $\frac1t\mathrm{d}\Gamma(q^t, g^t)$, we proceed similarly but use (ii) instead of Lemma \ref{lem:evolp2-1}. We introduce $\tilde{q}^t$, defined as $q^t$, such that $\tilde{q}^t q^t = q^t$ and hence $\Gamma(q^t)=\Gamma(\tilde{q}^t)\Gamma(q^t)$. This yields
\begin{align*}
& \frac{1}{t} \Big|\Big\langle e^{-\mathrm{i}tH} u, \Big ( \chi(H) \mathrm{d}\Gamma(q^t_i , g^t_i ) \frac{A}{t}\Gamma(q^t) \chi(H) + \mathrm{h.c.} \Big ) e^{-\mathrm{i}tH}u\Big\rangle\Big| \\
& = \frac{1}{t} \Big|\Big\langle e^{-\mathrm{i}tH} u, \Big ( \chi(H) \mathrm{d}\Gamma(q^t_i , g^t_i) (N+1)^{-\frac12}(N+1)^{\frac12}\frac{A}{t}\Gamma(\tilde{q}^t)\Gamma(q^t) \chi(H) + \mathrm{h.c.} \Big ) e^{-\mathrm{i}tH}u\Big\rangle\Big| \\
& \lesssim \frac{1}{t} \Big \| \mathrm{d}\Gamma((g^t_i)^*g^t_i)^{\frac12} \chi(H) e^{-\mathrm{i}tH} u \Big \| \big \| (N+1)^{\frac32} \Gamma(q^t) \chi(H) e^{-\mathrm{i}tH} u \big \| \\
& \lesssim \frac{\alpha(|n|+1)^3}{t} \Big \| \mathrm{d}\Gamma((g^t_i)^*g^t_i)^{\frac12} \chi(H) e^{-\mathrm{i}tH} u \Big \|^2 + \frac{1}{\alpha(|n|+1)^3t} \big \| (N+1)^{\frac32} \Gamma(q^t) \chi(H) e^{-\mathrm{i}tH} u \big \|^2 ,
\end{align*}
where we recall that $n\in\mathbb{Z}$ has been fixed such that $u \in \mathrm{Ran}( \mathds{1}_{\{n\}}(N_{\mathrm{lept}}-N_{\mathrm{neut}}))$. The parameter $\alpha>0$ will be determined later. Summing over $i$, since $N_{\mathrm{lept}}-N_{\mathrm{neut}}$ commutes with $H$, this implies that
\begin{align}
& \frac{1}{t} \Big|\Big\langle e^{-\mathrm{i}tH} u, \Big ( \chi(H) \mathrm{d}\Gamma(q^t, g^t) \frac{A}{t}\Gamma(q^t) \chi(H) + \mathrm{h.c.} \Big ) e^{-\mathrm{i}tH}u\Big\rangle\Big| \notag \\
& \lesssim \frac{\alpha}{t} \Big \| \mathrm{d}\Gamma((g^t)^*g^t)^{\frac12} \chi(H) e^{-\mathrm{i}tH} (|N_{\mathrm{neut}}-N_{\mathrm{lept}}|+1)^{\frac32} u \Big \|^2 + \frac{1}{\alpha t} \big \| \Gamma(q^t) \chi(H) e^{-\mathrm{i}tH} u \big \|^2 \notag \\
& \qquad + \frac{1}{\alpha t} \big \| (N_{\mathrm{lept}} + N_W +1)^{\frac32} \Gamma(q^t) \chi(H) e^{-\mathrm{i}tH} u \big \|^2 . \label{eq:dd2}
\end{align}
Remembering that $\tilde \chi \chi = \chi$, the last term of \eqref{eq:dd2} can be decomposed into
\begin{align*}
& \frac{1}{\alpha t} \big \| (N_{\mathrm{lept}} + N_W +1)^{\frac32} \Gamma(q^t) \chi(H) e^{-\mathrm{i}tH} u \big \|^2 \\
& = \frac{1}{\alpha t} \big \| (N_{\mathrm{lept}} + N_W +1)^{\frac32} \big [ \Gamma(q^t) , \tilde \chi(H) \big ] \chi ( H ) e^{-\mathrm{i}tH} u \big \|^2 + \frac{1}{\alpha t} \big \| (N_{\mathrm{lept}} + N_W +1)^{\frac32} \tilde \chi ( H ) \Gamma(q^t) \chi( H ) e^{-\mathrm{i}tH} u \big \|^2.
\end{align*}
Using the Helffer-Sj{\"o}strand functional calculus and similar arguments as before, one verifies that
\begin{align}
\big\|\big[\Gamma(q^t),\chi(H)\big](N_{\mathrm{neut}}+1)^{-1}\big\|=\mathcal{O}(t^{-1}). \label{eq:kha_1}
\end{align}
By (ii) in Lemma \ref{lemV1NumbestimatesV2}, it is not difficult to deduce from \eqref{eq:kha_1} and the previous equality that
\begin{align}
& \frac{1}{\alpha t} \big \| (N_{\mathrm{lept}} + N_W +1)^{\frac32} \Gamma(q^t) \chi(H) e^{-\mathrm{i}tH} u \big \|^2 \lesssim \frac{1}{\alpha t^2} \big \| ( N_{\mathrm{neut}}^2 + 1) u \big \|^2 + \frac{1}{\alpha t} \big \| \Gamma(q^t) \chi( H ) e^{-\mathrm{i}tH} u \big \|^2. \label{eq:lart_1}
\end{align}

Combining \eqref{eq:dd1}, \eqref{eq:dd2} and \eqref{eq:lart_1} gives
\begin{align}
 |R_1(t)| & \lesssim \frac{\alpha}{t} \Big \| \mathrm{d}\Gamma((g^t)^*g^t)^{\frac12} \chi(H) e^{-\mathrm{i}tH} ( |N_{\mathrm{neut}}-N_{\mathrm{lept}}|+1)^{\frac32} u \Big \|^2 + \frac{1}{\alpha t} \big \| \Gamma(q^t) \chi(H) e^{-\mathrm{i}tH} u \big \|^2 \notag \\
&\quad + \mathcal{O}(t^{-2}) \|(N_{\mathrm{neut}}^2 + N_{\mathrm{neut}}+1)u\|^2 \notag \\
&\quad+ \mathcal{O}(t^{-\frac32+\frac{1}{2+\mu}}) \big ( \| \mathrm{d}\Gamma(|p_2|^{-1})^{\frac12} u \| +\| N_{\mathrm{neut}}^{\frac12}u \|+ \|u\| \big ) \big ( \big \| (N_{\mathrm{neut}}+1)^{\frac32} u \big \| + \| u \| \big ). \label{eq:estim_R1}
\end{align}

Next we consider $R_3(t)$. It follows from \eqref{eq:Mourre_for_propag} that 
\begin{align*}
R_3(t)  \ge \frac{1}{t} \Big\langle e^{-\mathrm{i}tH}u, \chi(H) \Gamma(q^t) \big ( \mathrm{c}_0 ( N_{\mathrm{neut}} + \mathds{1} ) - \mathrm{C}_\lambda ( \mathds{1} - \tilde{\chi}^2(H) ) \big ) \Gamma(q^t) \chi(H) e^{-\mathrm{i}tH}u\Big\rangle.
\end{align*}
Using \eqref{eq:kha_1} together with the previous equation, the facts that $N_{\mathrm{lept}}-N_{\mathrm{neut}}$ commutes with $H$, that $N_W$ and $N_{\mathrm{lept}}$ are relatively $H$-bounded and that $\chi \tilde\chi = \chi$, this gives
\begin{align}
R_3(t) \ge \frac{\mathrm{c}_0}{t} \Big\langle e^{-\mathrm{i}tH}u, \chi(H) \Gamma(q^t) ( N_{\mathrm{neut}} + \mathds{1} ) \Gamma(q^t) \chi(H) e^{-\mathrm{i}tH}u\Big\rangle + \mathcal{O}(t^{-2})\|(N_{\mathrm{neut}}+1)u\|^2. \label{eq:estim_R3}
\end{align}

To estimate $R_4(t)$, it suffices to apply \eqref{Step1} together with the fact that $N_W$ and $N_{\mathrm{lept}}$ are relatively $H$-bounded. This yields
\begin{align}
| R_4(t) | \lesssim \frac{\delta}{t} \Big\langle e^{-\mathrm{i}tH}u, \chi(H) \Gamma(q^t) (N_{\mathrm{neut}}+1 )\Gamma(q^t) \chi(H) e^{-\mathrm{i}tH}u\Big\rangle . \label{eq:estim_R4}
\end{align}

Putting together \eqref{eq:estim_R2}, \eqref{eq:estim_R1}, \eqref{eq:estim_R3} and \eqref{eq:estim_R4}, we finally arrive at
\begin{align*}
\partial_t h(t) \ge & \Big ( \frac{\mathrm{c}_0}{t} - \frac{1}{\alpha t} - \frac{ \delta }{ t } \Big ) \Big\langle e^{-\mathrm{i}tH}u, \chi(H) \Gamma(q^t) ( N_{\mathrm{neut}} + \mathds{1} ) \Gamma(q^t) \chi(H) e^{-\mathrm{i}tH}u\Big\rangle \\
& + \frac{\alpha}{t} \Big \| \mathrm{d}\Gamma(g_t^*g_t)^{\frac12} \chi(H) e^{-\mathrm{i}tH} (|N_{\mathrm{neut}}-N_{\mathrm{lept}}|+1)^{\frac32} u \Big \|^2  + \mathcal{O}(t^{-\min(2,1+\mu)})\|(N_{\mathrm{neut}}^2+N_{\mathrm{neut}}+1)u\|^2 \\
& + \mathcal{O}(t^{-\frac32+\frac{1}{2+\mu}}) \big ( \| \mathrm{d}\Gamma(|p_2|^{-1})^{\frac12} u \| + \|u\| \big ) \big ( \big \| (N_{\mathrm{neut}}+1)^{\frac32} u \big \| + \| u \| \big ).
\end{align*}
Fixing $\alpha$ large enough and $\delta$ small enough and integrating over $t$ from $1$ to $\infty$, we obtain from \eqref{eq:estimh(t)} and (ii) that
\begin{align*}
& \int_1^\infty \Big\langle e^{-\mathrm{i}tH}u, \chi(H) \Gamma(q^t) ( N_{\mathrm{neut}} + 1 ) \Gamma(q^t) \chi(H) e^{-\mathrm{i}tH}u\Big\rangle \frac{dt}{t} \\
&\lesssim \big\| ( \mathrm{d}\Gamma(|p_2|^{-1}) + N_{\mathrm{neut}} + 1 )^{\frac12} (N_{\mathrm{neut}}+1)^{\frac32} u \big\|^2.
\end{align*}
This proves (iv) for any $\chi \in \mathrm{C}_0^\infty( ( \lambda-\varepsilon,\lambda+\varepsilon) )$. The extension of the result to any $\chi \in \mathrm{C}_0^\infty( \mathbb{R} \setminus (\sigma_{\mathrm{pp}}(H) \cup \tau ) )$ follows from standard arguments (see, e.g., \cite[Proposition 4.4.7]{DeGe97_01}).
\end{proof}
As mentioned before, the remainder of this section is devoted to the proof of a different version of propagation estimates, involving a time-dependent modified position operator. In accordance with the notations previously introduced in this section, we set
\begin{equation*}
x_{t,\rho} := \Big ( \frac{ |p_2| }{ |p_2| + t^{-\rho} } \Big )^{\frac12} x \Big ( \frac{ |p_2| }{ |p_2| + t^{-\rho} } \Big )^{\frac12} ,
\end{equation*}
meaning that $x_{t,\rho,i} = x_i = \mathrm{i}\nabla$ if $i\in\{1,2,3,4,7,8,11,12\}$ (corresponding to the label of a massive particle), and
\begin{equation*}
x_{t,\rho,i} = \Big ( \frac{ |p_2| }{ |p_2| + t^{-\rho} } \Big )^{\frac12} x_i \Big ( \frac{ |p_2| }{ |p_2| + t^{-\rho} } \Big )^{\frac12} = \Big ( \frac{ |p_2| }{ |p_2| + t^{-\rho} } \Big )^{\frac12} (\mathrm{i}\nabla_{p_2}) \Big ( \frac{ |p_2| }{ |p_2| + t^{-\rho} } \Big )^{\frac12} ,
\end{equation*}
if $i \in \{5,6,9,10,13,14\}$ (corresponding to the label of a neutrino). Likewise, 
\begin{equation*}
\omega_{t,\rho} := \Big ( \frac{ |p_2| }{ |p_2| + t^{-\rho} } \Big )^{\frac12} \omega \Big ( \frac{ |p_2| }{ |p_2| + t^{-\rho} } \Big )^{\frac12} ,
\end{equation*}
where the notation $\omega$ has been introduced in \eqref{eq:not_omega}. Similarly as above, for $R=(R_1,\dots,R_{14})$ and $R'=(R'_1,\dots,R'_{14})$, we have that
\begin{equation*}
\mathds{1}_{[R,R']}(|x_{t,\rho}|)=\big(\mathds{1}_{[R_1,R'_1]}( | x_{t,\rho,1} | ),\dots,\mathds{1}_{[R_{14},R'_{14}]}( | x_{t,\rho,14} | )\big).
\end{equation*}
\begin{Th}\label{th:propag_massless_2}
Suppose that the masses of the neutrinos $m_{\nu_e}$, $m_{\nu_{\mu}}$, $m_{\nu_{\tau}}$ vanish and consider the Hamiltonian $H = H_0 + g ( H_I^{(1)} + H_I^{(2)} )$ with $H_I^{(1)}$ and $H_I^{(2)}$ given by \eqref{interaction_term}. Assume that
\begin{equation*}
G \in L^2, \quad a_{(i),\cdot} G \in L^2, \quad |p_3|^{-1} a_{(i),\cdot} G \in L^2, \quad i = 1,2,3 ,
\end{equation*}
and that
\begin{equation*}
G \in \mathbb{H}^{1+\mu} \text{ for some } \mu > 0.
\end{equation*}
\begin{itemize}
\item[(i)] Let $\rho>0$ be such that $(1+\mu)^{-1} < \rho < 1$ and let $c>\rho^{-1}$. Let $\chi \in \mathrm{C}_0^{\infty}(\mathbb{R})$, $R=(R_1,\dots,R_{14})$ and $R'=(R'_1,\dots,R'_{14})$ be such that $R'_{i}>R_i>c$. There exists $\mathrm{C} > 0$ such that, for all $u \in \mathscr{D}(N_{\mathrm{neut}}^{\frac12})$, 
\begin{align*}
& \int^{\infty}_{1} \Big \| \mathrm{d} \Gamma \Big ( \mathds{1}_{[R,R']} \Big ( \frac{ | x_{t,\rho} | }{ t } \Big ) \Big)^{\frac{1}{2}}\chi(H) e^{-\mathrm{i}tH} u \Big \|^2 \frac{dt}{t} \leq \mathrm{C} \big ( \| N_{\mathrm{neut}}^{\frac12} u \|^2 + \|u\|^{2} \big ).
\end{align*}
\item[(ii)] Let $\rho>0$ be such that $(1+\mu)^{-1} < \rho < 1$. Let $0<\mathrm{v}_0<\mathrm{v}_1$ and $\chi \in \mathrm{C}_0^{\infty}(\mathbb{R})$. There exists $\mathrm{C}>0$ such that, for all $u \in \mathscr{D}(N_{\mathrm{neut}}^{\frac12})$,
\begin{align*}
& \int^{\infty}_{1} \Big \| \mathrm{d}\Gamma \Big( \Big\langle \big( \frac{x_{t,\rho}}{t} - \nabla \omega_{t,\rho} \big ) , \mathds{1}_{ [ \mathrm{v}_0 , \mathrm{v}_1 ] }( \frac{x_{t,\rho}}{t} ) \big( \frac{x_{t,\rho}}{t} - \nabla \omega_{t,\rho} \big ) \Big\rangle \Big)^{\frac12} \chi(H)e^{-\mathrm{i} t H} u \Big\|^{2} \frac{dt}{t} \\
& \le \mathrm{C} \big ( \| N_{\mathrm{neut}}^{\frac12} u \|^2  + \|u\|^{2} \big ).
\end{align*}
\item[(iii)] Let $\rho>0$ be such that $(1+\mu)^{-1} < \rho < 1$. Let $0< \mathrm{v}_0<\mathrm{v}_1$, $J \in \mathrm{C}_0^{\infty}( \{ x \in \mathbb{R}^3 , \mathrm{v}_0 < |x| < \mathrm{v}_1\})$ and $\chi \in \mathrm{C}_0^{\infty}(\mathbb{R})$. There exists $\mathrm{C}>0$ such that, for all $\ell \in \{ 1,2,3\} $ and $u \in \mathscr{D}(N_{\mathrm{neut}}^{\frac12})$,
\begin{align*}
& \int^{\infty}_{1} \Big \| \mathrm{d}\Gamma \Big( \Big| J\Big(\frac{x_{t,\rho}}{t}\Big) \Big( \frac{x_{t,\rho}^{(\ell)}}{t} - \partial_{(\ell)} \omega_{t,\rho} \Big) + \mathrm{h.c }\Big|\Big)^{\frac{1}{2}}\chi(H)e^{-\mathrm{i} t H} u \Big\|^{2} \frac{dt}{t} \le \mathrm{C} \big ( \| N_{\mathrm{neut}}^{\frac12} u \|^2 + \|u\|^{2} \big ).
\end{align*}
\item[(iv)] There exists $g_0>0$ such that, for all $|g|\le g_0$, the following holds: let $\rho>0$ be such that $(1+\mu)^{-1} < \rho < 1$. Let $\chi \in \mathrm{C}_{0}^{\infty}(\mathbb{R})$ be supported in $\mathbb{R}\backslash(\tau \cup \sigma_{\mathrm{pp}}(H))$. There exist $\delta>0$ and $\mathrm{C} > 0$ such that, for all $u \in \mathscr{D}( ( \mathrm{d}\Gamma(|p_2|^{-1}) + N_{\mathrm{neut}} + 1 ) (N_{\mathrm{neut}}+1)^{\frac32} )$,
\begin{equation*}
\int^{\infty}_{1} \Big\| \Gamma \Big ( \mathds{1}_{[0,\delta]} \Big(\frac{|x_{t,\rho}|}{t}\Big)\Big) \chi(H) e^{-\mathrm{i} t H} u \Big\|^2 \frac{dt}{t} \le \mathrm{C} \big\| ( \mathrm{d}\Gamma(|p_2|^{-1}) + N_{\mathrm{neut}} + 1 ) (N_{\mathrm{neut}}+1)^{\frac32} u \big\|^2.
 \end{equation*}
\end{itemize}
\end{Th}
\begin{proof}
(i) The proof is similar to that of Theorem \ref{th:propag_massless} (i), with the following differences. Let $F_i: \mathbb{R} \to \mathbb{R}$ be a bounded non-decreasing function supported in $(1,\infty)$. Instead of the function $f_i$ in \eqref{eq:def_fj}, we consider
\begin{align*}
\tilde f_i(t) := \Big \langle e^{ - \mathrm{i} t H } \chi ( H ) u , \mathrm{d}\Gamma_{i} \Big( F_i \Big(\frac{ | x_{t,\rho,i} | }{ct}\Big) \Big) e^{ - \mathrm{i} t H } \chi(H) u \Big \rangle ,
\end{align*}
with $c>1$. Obviously, in the massive case ($i\in\{1,2,3,4,7,8,11,12\}$) the proof is the same than in Theorem \ref{th:propag_massless}. Let $i \in \{5,6,9,10,13,14\}$. In the same way as for $f_i$, since $F_i$ is bounded, we have that
\begin{equation}
\tilde f_i(t) \lesssim \| N_{\mathrm{neut}}^{\frac12} u \|^2 + \| u \|^2. \label{eq:new_propag1}
\end{equation}
Write $F_i ( r ) = \tilde F_i( r^2 )$. Note that
\begin{align*}
x_{t,\rho,i}^2 = \sum_{\ell=1}^3 \Big ( \frac{ |p_2| }{ |p_2| + t^{-\rho} } \Big )^{\frac12} x_i^{(\ell)} \Big ( \frac{ |p_2| }{ |p_2| + t^{-\rho} } \Big ) x_i^{(\ell)} \Big ( \frac{ |p_2| }{ |p_2| + t^{-\rho} } \Big )^{\frac12}.
\end{align*}
We compute
\begin{align}
& \mathbf{d}_0 \tilde F_i \Big ( \frac{ x_{t,\rho,i}^2 }{ c^2 t^2 } \Big ) \le - 2 \frac{1-\rho}{t} \tilde F_i' \Big ( \frac{ x_{t,\rho,i}^2 }{ c^2 t^2 } \Big )  - 2 \frac{ \rho - c^{-1} }{ t } \Big ( \frac{ |p_2| }{ |p_2| + t^{-\rho} } \Big )^{\frac12} \tilde F_i' \Big ( \frac{ x_{t,\rho}^2 }{ c^2 t^2 } \Big ) \Big ( \frac{ |p_2| }{ |p_2| + t^{-\rho} } \Big )^{\frac12} + \mathcal{O}( t^{-2+\rho} ) . \label{eq:new_propag2}
\end{align}
Details of the estimate \eqref{eq:new_propag2} are provided in Appendix \ref{app:technical}. The first two terms are non-positive since $\rho < 1$ and $c > \rho^{-1}$.

Moreover, the commutators 
\begin{equation*}
\Big [ H_I^{(j)}(G) , \mathrm{d}\Gamma_{i} \Big ( F_i \Big(\frac{ | x_{t,\rho,i} | }{ ct } \Big) \Big)  \Big ], \quad j=1,2,
\end{equation*}
are given by expressions similar to \eqref{commut_11}--\eqref{commut_21}, with the operator $F_i (\frac{ | x_{t,\rho,i} | }{ ct })$ instead of $a_i$. Since $F_i$ vanishes near $0$ and $G \in \mathbb{H}^{1+\mu}$, one can verify using interpolation that $G$ belongs to the domain of $| x_{t,\rho,i} |^{1+\mu}$, and hence
\begin{align*}
\Big\|F_i\Big(\frac{ | x_{t,\rho,i} | }{ ct }\Big)G\Big\|_2=\mathcal{O}(t^{-1-\mu}).
\end{align*}
Therefore, in the same way as for Equation \eqref{eq:propagestim4}, Lemma \ref{lemSCW_1} yields
\begin{align}
& \Big \langle e^{ - \mathrm{i} t H } \chi ( H ) u , \Big [ H_I^{(1)}(G) + H_I^{(2)}(G) , \mathrm{i} \mathrm{d}\Gamma_{j} \Big( F_i \Big( \frac{ | x_{t,\rho,i} | }{ ct }\Big) \Big) \Big ] e^{ - \mathrm{i} t H } \chi(H) u \Big \rangle = \mathcal{O}( t^{-1-\mu} ) \| u \|^2. \label{eq:new_propag3}
\end{align}

Using \eqref{eq:new_propag1}, \eqref{eq:new_propag2} and \eqref{eq:new_propag3}, one can conclude that (i) holds by arguing in the same way as in the proof of Theorem \ref{th:propag_massless} (i).

\vspace{0,2cm}

(iv) Again, the proof resembles that of Theorem \ref{th:propag_massless} (iv), we focus on the differences. We consider $\chi \in \mathrm{C}_0^\infty( ( \lambda-\varepsilon,\lambda+\varepsilon) )$, $\delta > 0 $ and $q_i$, $i \in \{1, \dots, 14\}$ as in the proof of Theorem \ref{th:propag_massless}. Let $u \in \mathrm{Ran}( \mathds{1}_{\{n\}}(N_{\mathrm{lept}}-N_{\mathrm{neut}}))$ for some $n \in \mathbb{Z}$. Let $\tilde{q}^t=(q_1(\frac{ x_{t,\rho,1} }{t}),\dots,q_{14}(\frac{ x_{t,\rho,14} }{t}))$ and
\begin{equation*}
\tilde h(t):=\Big\langle e^{-\mathrm{i}tH}\chi(H)u,\Gamma(\tilde{q}^t)\frac{A}{t}\Gamma(\tilde{q}^t)e^{-\mathrm{i}tH}\chi(H)u\Big\rangle.
\end{equation*}
To verify that $\tilde h(t)$ is uniformly bounded, we modify \eqref{Step1} as follows. We can decompose 
\begin{equation*}
A=\sum_{i=1}^{14} \mathrm{d}\Gamma_i( a_i ).
\end{equation*}
Clearly the massive cases ($i\in\{1,2,3,4,7,8,11,12\}$) can be handled as in \eqref{Step1}. For $i \in \{5,6,9,10,13,14\}$, we recall that
\begin{equation*}
a_i = \frac{\mathrm{i}}{2} \Big ( \frac{p_2}{|p_2|} \cdot  \nabla + \nabla \cdot \frac{p_2}{|p_2|} \Big ) ,
\end{equation*}
and write
\begin{align}
 \Big\langle e^{-\mathrm{i}tH}\chi(H)u,\Gamma(\tilde{q}^t) \frac{\mathrm{d}\Gamma_i(a_i)}{t} \Gamma(\tilde{q}^t) e^{-\mathrm{i}tH}\chi(H)u\Big\rangle & = \Big\langle e^{-\mathrm{i}tH}\chi(H)u,\Gamma(\tilde{q}^t) \Big [ \frac{\mathrm{d}\Gamma_i(a_i)}{t} , \Gamma(\tilde{q}^t) \Big ] e^{-\mathrm{i}tH}\chi(H)u\Big\rangle \notag \\
&\quad + \Big\langle e^{-\mathrm{i}tH}\chi(H)u,\Gamma(\tilde{q}^t)^2 \frac{\mathrm{d}\Gamma_i(a_i)}{t} e^{-\mathrm{i}tH}\chi(H)u\Big\rangle. \label{eq:ddd0}
\end{align}
We estimate each term separately. First, we compute, using \cite[Lemma 5.2]{BoFaSi12_01},
\begin{equation*}
\Big [ \frac{ a_i }{ t } , q_i \Big ( \frac{ x_{t,\rho,i} }{ t } \Big ) \Big ] = \mathcal{O}( t^{-1} ) |p_2|^{-1}.
\end{equation*}
This implies that
\begin{align}
\Big | \Big\langle e^{-\mathrm{i}tH}\chi(H)u,\Gamma(\tilde{q}^t) \Big [ \frac{\mathrm{d}\Gamma_i(a_i)}{t} , \Gamma(\tilde{q}^t) \Big ] e^{-\mathrm{i}tH}\chi(H)u\Big\rangle \Big | &\lesssim t^{-1} \| u \| \big \| \mathrm{d}\Gamma( |p_2|^{-1} ) e^{-\mathrm{i}tH}\chi(H)u \big \| \notag \\
& \lesssim \| u \| \big ( \| u\|+ \|N_{\mathrm{neut}} u\|+\big\| \mathrm{d}\Gamma( |p_2|^{-1} )u \big \| \big ) , \label{eq:ddd1}
\end{align}
where we used Lemma \ref{lem:evolp2-1} in the second inequality. Next we have that
\begin{align}\label{eq:rewrite_ai}
a_i &= \sum_{\ell=1}^3 x_{t,\rho,i}^{(\ell)}\big ( 1 + \mathcal{O}(t^{-\rho} ) |p_2|^{-1} \big ) +  \mathcal{O}( 1 ) |p_2|^{-1} ,
\end{align}
see Appendix \ref{app:technical}. Using that $|x_{t,\rho,i}|\le\delta t$ on the support of $q^t_i$, we deduce from the previous equation that
\begin{align}
& \Big | \Big\langle e^{-\mathrm{i}tH}\chi(H)u,\Gamma( \tilde{q}^t )^2  \frac{\mathrm{d}\Gamma_i(a_i)}{t} e^{-\mathrm{i}tH}\chi(H)u\Big\rangle \Big | \notag \\
&\lesssim \delta \| u \| \big ( \big \| N_{\mathrm{neut}} e^{-\mathrm{i}tH}\chi(H)u \big \| + \mathcal{O}(t^{-\rho}) \big \| \mathrm{d}\Gamma( |p_2|^{-1} ) e^{-\mathrm{i}tH}\chi(H)u \big \| \big )  + \mathcal{O}(t^{-1}) \| u \| \big \| \mathrm{d}\Gamma( |p_2|^{-1} ) e^{-\mathrm{i}tH}\chi(H)u \big \|  \notag \\
&\lesssim \| u \| \big ( \big \| N_{\mathrm{neut}} u \big \| + \big \| \mathrm{d}\Gamma( |p_2|^{-1} ) u \big \| + \|u \| \big ) , \label{eq:ddd2}
\end{align}
where we used Lemma \ref{lem:evolp2-1} and the fact that $\rho > (1+\mu)^{-1}$ in the second inequality. We deduce from \eqref{eq:ddd1} and \eqref{eq:ddd2} that \eqref{eq:ddd0}, and therefore $\tilde h(t)$, are uniformly bounded in $t$.

Next, we can decompose the derivative $\partial_t \tilde h ( t ) = \tilde R_1(t) + \tilde R_2(t) + \tilde R_3(t) + \tilde R_4(t)$ analogously as $h(t)$ in \eqref{eq:R1-R4}. The term $\tilde R_2(t)$ can be estimated in the same way as $R_2(t)$, using estimates similar to \eqref{eq:ddd1} and \eqref{eq:ddd2} instead of \eqref{Step1}. This yields
\begin{equation}
|\tilde R_2(t)|\lesssim \mathcal{O}(t^{-1-\mu})\|u\|\big(\|(N_{\mathrm{neut}} u\| + \big \| \mathrm{d}\Gamma( |p_2|^{-1} ) u \big\| + \|u\| \big). \label{eq:estim_tilde_R2}
\end{equation}
The estimate for $\tilde R_3(t)$ is identical to that of $R_3(t)$, yielding
\begin{align}
\tilde R_3(t) \ge \frac{\mathrm{c}_0}{t} \Big\langle e^{-\mathrm{i}tH}u, \chi(H) \Gamma(\tilde{q}^t) ( N_{\mathrm{neut}} + \mathds{1} ) \Gamma(\tilde{q}^t) \chi(H) e^{-\mathrm{i}tH}u\Big\rangle + \mathcal{O}(t^{-2})\|(N_{\mathrm{neut}}+1)u\|^2. \label{eq:estim_tilde_R3}
\end{align}
To estimate $\tilde R_4(t)$, we proceed as in \eqref{eq:ddd1}--\eqref{eq:ddd2}. This gives
\begin{align}
| \tilde R_4(t) | = \frac{1}{t} | \tilde h ( t ) | & \lesssim  \frac{\delta}{t} \Big\langle e^{-\mathrm{i}tH}u, \chi(H) \Gamma(\tilde{q}^t) (N_{\mathrm{neut}}+1 )\Gamma(\tilde{q}^t) \chi(H) e^{-\mathrm{i}tH}u\Big\rangle \notag \\
&\quad + t^{-1-\rho} \| u \| \big \| \mathrm{d}\Gamma( |p_2|^{-1} ) e^{-\mathrm{i}tH}\chi(H)u \big \| \notag \\
& \lesssim  \frac{\delta}{t} \Big\langle e^{-\mathrm{i}tH}u, \chi(H) \Gamma(\tilde{q}^t) (N_{\mathrm{neut}}+1 )\Gamma(\tilde{q}^t) \chi(H) e^{-\mathrm{i}tH}u\Big\rangle \notag \\
&\quad + t^{-1-\rho+\alpha(1+\mu)^{-1}} \| u \| \big(\big\|\mathrm{d}\Gamma( |p_2|^{-1} ) u\big\|+\|u\|\big) , \label{eq:estim_tilde_R4}
\end{align}
the second inequality being a consequence of Lemma \ref{lem:evolp2-1}.

Next we consider $\tilde R_1(t)$. As in the proof of Theorem \ref{th:propag_massless}, we compute
\begin{align*}
\mathbf{d}_0 \tilde{q}_i^t(x_{t,\rho,i}) =  -\frac{1}{2 t} \Big \langle \frac{x_{t,\rho,i}}{t} - \nabla\omega_{t,\rho,i} ,  \nabla \tilde{q}_t \big( \frac{x_{t,\rho,i}}{t} \big) \Big \rangle + \mathrm{h.c.}   + \mathcal{O}(t^{-2+\rho}) ,
\end{align*}
and set $\tilde g^t_i :=  -\frac{1}{2} \langle \frac{x_{t,\rho,i}}{t} - \nabla\omega_{t,\rho,i} ,  \nabla \tilde q_j \big( \frac{x_{t,\rho,i}}{t} \big) \rangle + \mathrm{h.c.}$. Let also $\tilde r^t_i$ be the remainder term in the previous equality. We have that  
\begin{align*}
& \Big|\Big\langle e^{-\mathrm{i}tH} u, \Big ( \chi(H) \mathrm{d}\Gamma( \tilde{q}^t_i , \tilde{r}^t_i) \frac{\mathrm{d}\Gamma_i(a_i)}{t}\Gamma(\tilde{q}^t) \chi(H) + \mathrm{h.c.} \Big ) e^{-\mathrm{i}tH}u\Big\rangle\Big| \\
& \lesssim \mathcal{O}(t^{-2+\rho}) \big \| N_{\mathrm{neut}} \chi(H) e^{-\mathrm{i}tH} u \big \| \Big \| \frac{\mathrm{d}\Gamma_i(a_i)}{t}\Gamma(\tilde{q}^t) \chi(H) e^{-\mathrm{i}tH}u \Big \|  \\
& \lesssim \mathcal{O}(t^{-2+\rho}) \big \| N_{\mathrm{neut}} \chi(H) e^{-\mathrm{i}tH} u \big \| \Big ( \big \| N_{\mathrm{neut}} \chi(H) e^{-\mathrm{i}tH}u \big \| + \mathcal{O}(t^{-\rho}) \big \| \mathrm{d}\Gamma( |p_2|^{-1} ) \chi(H) e^{-\mathrm{i}tH}u \big \| \Big ) ,
\end{align*}
where we used \eqref{eq:ddd1}--\eqref{eq:ddd2} in the second inequality. By Lemma \ref{lem:evolp2-1}, and since $N_{\mathrm{neut}}-N_{\mathrm{lept}}$ commutes with $H$ and $N_{\mathrm{lept}}$ is relatively $H$-bounded, we obtain that
\begin{align*}
& \Big|\Big\langle e^{-\mathrm{i}tH} u, \Big ( \chi(H) \mathrm{d}\Gamma_i( \tilde{q}^t_i , \tilde{r}^t_i ) \frac{\mathrm{d}\Gamma_i(a_i)}{t}\Gamma(\tilde{q}^t) \chi(H) + \mathrm{h.c.} \Big ) e^{-\mathrm{i}tH}u\Big\rangle\Big| \\
& \lesssim \mathcal{O}(t^{-2+\rho}) \big ( \big \| N_{\mathrm{neut}} u \big \| + \| u \| \big ) \big ( \big \| N_{\mathrm{neut}} u \big \| +  \big \| \mathrm{d}\Gamma(|p_2|^{-1})u\big\|+\|u\|\big).
\end{align*}

To estimate the term corresponding to $\frac1t\mathrm{d}\Gamma(\tilde{q}^t_i , \tilde{g}^t_i)$, we write
\begin{align*}
& \frac{1}{t} \Big|\Big\langle e^{-\mathrm{i}tH} u, \Big ( \chi(H) \mathrm{d}\Gamma(\tilde{q}^t_i , \tilde{g}^t_i ) \frac{\mathrm{d}\Gamma_i(a_i)}{t}\Gamma(\tilde{q}^t) \chi(H) + \mathrm{h.c.} \Big ) e^{-\mathrm{i}tH}u\Big\rangle\Big| \\
& = \frac{1}{t} \Big|\Big\langle e^{-\mathrm{i}tH} u, \Big ( \chi(H) \mathrm{d}\Gamma(\tilde{q}^t_i , \tilde{g}^t_i ) (N_{\mathrm{neut}}+1)^{-\frac12}(N_{\mathrm{neut}}+1)^{\frac12}\frac{\mathrm{d}\Gamma_i(a_i)}{t}\Gamma(\tilde{q}^t) \chi(H) + \mathrm{h.c.} \Big ) e^{-\mathrm{i}tH}u\Big\rangle\Big| \\
& \lesssim \frac{1}{t} \Big \| \mathrm{d}\Gamma((\tilde{g}^t_i)^*\tilde{g}^t_i)^{\frac12} \chi(H) e^{-\mathrm{i}tH} u \Big \| \Big ( \big \| ( N_{\mathrm{neut}} + 1 )^{\frac32} \Gamma(\tilde{q}^t) \chi(H) e^{-\mathrm{i}tH}u \big \| \\
& \qquad \qquad \qquad \qquad \qquad \qquad \qquad\qquad + \mathcal{O}(t^{-\rho}) \big \| ( N_{\mathrm{neut}} + 1 )^{\frac12} \mathrm{d}\Gamma( |p_2|^{-1} ) \chi(H) e^{-\mathrm{i}tH}u \big \| \Big ) ,
\end{align*}
where we used again \eqref{eq:ddd1}--\eqref{eq:ddd2}. We expand the last expression and estimate the two terms separately. The first one is estimated exactly as in the proof of Theorem \ref{th:propag_massless}, yielding
\begin{align*}
& \frac{1}{t} \Big \| \mathrm{d}\Gamma((\tilde{g}^t_i)^*\tilde{g}^t_i)^{\frac12} \chi(H) e^{-\mathrm{i}tH} u \Big \| \big \| ( N_{\mathrm{neut}} + 1 )^{\frac32}\Gamma(\tilde{q}^t) \chi(H) e^{-\mathrm{i}tH}u \big \| \\
& \lesssim \frac{\alpha}{t} \Big \| \mathrm{d}\Gamma((\tilde{g}^t)^*\tilde{g}^t)^{\frac12} \chi(H) e^{-\mathrm{i}tH} (|N_{\mathrm{neut}}-N_{\mathrm{lept}}|+1)^{\frac32} u \Big \|^2 + \frac{1}{\alpha t} \big \| \Gamma(\tilde{q}^t) \chi(H) e^{-\mathrm{i}tH} u \big \|^2 + \frac{1}{\alpha t^2} \big \| ( N_{\mathrm{neut}}^2 + 1) u \big \|^2,
\end{align*}
with $\alpha > 0$. For the second term, we have that
\begin{align*}
& \frac{1}{t^{1+\rho}} \Big \| \mathrm{d}\Gamma((\tilde{g}^t_i)^*\tilde{g}^t_i)^{\frac12} \chi(H) e^{-\mathrm{i}tH} u \Big \| \big \| ( N_{\mathrm{neut}} + 1 )^{\frac12} \mathrm{d}\Gamma( |p_2|^{-1} ) \chi(H) e^{-\mathrm{i}tH}u \big \| \\
& \lesssim \frac{1}{t^{1+\rho}} \Big \| \mathrm{d}\Gamma((\tilde{g}^t_i)^*\tilde{g}^t_i)^{\frac12} \chi(H) e^{-\mathrm{i}tH} u \Big \|^2 + \frac{1}{t^{1+\rho}} \big \| ( N_{\mathrm{neut}} + 1 )^{\frac12} \mathrm{d}\Gamma( |p_2|^{-1} ) \chi(H) e^{-\mathrm{i}tH}u \big \|^2,
\end{align*}
and, since $N_{\mathrm{lept}}$ is relatively $H$-bounded,
\begin{align*}
&\big \| ( N_{\mathrm{neut}} + 1 )^{\frac12} \mathrm{d}\Gamma( |p_2|^{-1} ) \chi(H) e^{-\mathrm{i}tH}u \big \|^2 \\
&=  \big \| \mathrm{d}\Gamma( |p_2|^{-1} ) \chi(H) e^{-\mathrm{i}tH} ( | N_{\mathrm{neut}} - N_{\mathrm{lept}} | + 1 )^{\frac12}u \big \|^2 + \big \| N_{\mathrm{lept}}^{\frac12} \mathrm{d}\Gamma( |p_2|^{-1} ) \chi(H) e^{-\mathrm{i}tH}u \big \|^2 \\
&\le  \big \| \mathrm{d}\Gamma( |p_2|^{-1} ) \chi(H) e^{-\mathrm{i}tH} ( | N_{\mathrm{neut}} - N_{\mathrm{lept}} | + 1 )^{\frac12}u \big \|^2 + \big \| (c_1H+c_2)^{\frac12} \mathrm{d}\Gamma( |p_2|^{-1} ) \chi(H) e^{-\mathrm{i}tH}u \big \|^2 ,
\end{align*}
where $c_1$ and $c_2$ are real numbers. Since one easily verifies that $[H,\mathrm{d}\Gamma(|p_2|^{-1})]$ is relatively $H$-bounded, Lemma \ref{lem:evolp2-1} implies that
\begin{align*}
& \big \| ( N_{\mathrm{neut}} + 1 )^{\frac12} \mathrm{d}\Gamma( |p_2|^{-1} ) \chi(H) e^{-\mathrm{i}tH}u \big \|^2 \lesssim t^{(1+\mu)^{-1}} \big \| ( \mathrm{d}\Gamma( |p_2|^{-1} ) + 1 ) ( N_{\mathrm{neut}} + 1 )^{\frac12} u \big \|^2 .
\end{align*}

Putting together the previous estimates, we obtain that
\begin{align}
& |\tilde R_1(t)| \lesssim \frac{\alpha+1}{t} \Big \| \mathrm{d}\Gamma(\tilde{g}_t^*\tilde{g}_t)^{\frac12} \chi(H) e^{-\mathrm{i}tH} (|N_{\mathrm{neut}}-N_{\mathrm{lept}}|+1)^{\frac32} u \Big \|^2 + \frac{1}{\alpha t} \big \| \Gamma(\tilde{q}^t) \chi(H) e^{-\mathrm{i}tH} u \big \|^2 \notag \\
&\quad + \mathcal{O}(t^{-2+\rho}) \big ( \big \| N_{\mathrm{neut}} u \big \| + \| u \| \big ) \big ( \big \| N_{\mathrm{neut}} u \big \| +  \big \| \mathrm{d}\Gamma(|p_2|^{-1})u\big\|+\|u\|\big) \notag \\
&\quad + \mathcal{O}(t^{-1-\rho+(1+\mu)^{-1}} ) \big \| ( \mathrm{d}\Gamma( |p_2|^{-1} ) + 1 ) ( N_{\mathrm{neut}} + 1 )^{\frac12} u \big \|^2 . \label{eq:estim_tilde_R1}
\end{align}
From \eqref{eq:estim_tilde_R2}--\eqref{eq:estim_tilde_R1}, we conclude the proof in the same way as for Theorem \ref{th:propag_massless} (iv).
\end{proof}
\section{Asymptotic completeness}\label{sec:AC}

In this section, we prove Theorem \ref{thm:main}. We begin by recalling the definitions and basic properties of the asymptotic spaces and of the wave operators in Subsection \ref{subsec:wave}. Subsection \ref{subsec:inverse_wave} is devoted to the proof of an important ingredient of the proof of Theorem \ref{thm:main}, namely the existence of inverse wave operators. Finally, in Subsection \ref{subsec:AC}, we establish asymptotic completeness of the wave operators.

\subsection{The asymptotic space and the wave operators}\label{subsec:wave}

Most of the results of this section are straightforward adaptations of corresponding results established in \cite{DeGe99_01,Am04_01} for Pauli-Fierz Hamiltonians. Therefore we do not give the details of the proofs but refer the reader to \cite{DeGe99_01,Am04_01}. The only exception is Theorem \ref{th:unitary} where we establish unitarity of the wave operators in the case where the masses of the neutrinos vanish. Indeed, when neutrinos are supposed to be massless, Theorem \ref{th:unitary} cannot be proven as in \cite{DeGe99_01,Am04_01}. We rely instead on an elegant argument due to \cite{DeGe00_01}.

Recall that the asymptotic creation and annihilation operators for bosons are formally defined by \eqref{eq:def_asympt_a}--\eqref{eq:def_a}, and by analogous formulas for leptons and neutrinos. See Appendix \ref{sec:asympt_a} for precise definitions and basic properties. In particular, the fermionic asymptotic creation or annihilation operators $b^{\pm,\sharp}_{l,\epsilon}(h)$ (for leptons) and $c^{\pm,\sharp}_{l,\epsilon}(h)$ (for neutrinos), with $h \in \mathfrak{h}_2$, are bounded. The bosonic operators $a^{\pm,\sharp}_{\epsilon}(h)$, with $h \in \mathfrak{h}_1$, are closed but unbounded.

The space of asymptotic vacua is defined by
\begin{equation*}
\mathscr{K}^{\pm} = \left\{u \in \mathscr{D} ( N_W^{\frac12} ) , \, \forall h_1 \in \mathfrak{h}_2,  h_2,h_3 \in \mathfrak{h}_1, l , \epsilon, a^{\pm}_\epsilon(h_1) u = b^{\pm}_{l,\epsilon}(h_2) u = c^{\pm}_{l,\epsilon}(h_3) u = 0 \right\}.
\end{equation*}
The asymptotic space is 
\begin{equation*}
\mathscr{H}^{\pm} = \mathscr{K}^{\pm} \otimes \mathscr{H}.
\end{equation*}
The following proposition can be proven in the same way as \cite[Proposition 5.5]{DeGe99_01}.
\begin{Prop}\label{prop:basic_prop_omega}
Consider the Hamiltonian \eqref{total_hamiltonian_g} with $H_I$ given by \eqref{interaction_term} and suppose that
\begin{equation*}
G \in \mathbb{H}^{1+\mu} \text{ for some } \mu > 0.
\end{equation*}
Then
\begin{itemize}
\item[(i)] $\mathscr{K}^{\pm}$ is closed and $H$-invariant.
\item[(ii)] For all $n \in \mathbb{N}$ and $h_1, \dots, h_n$, $\mathscr{K}^{\pm}$ is contained in the domain of  $d^{\pm,*}(h_1)\dots d^{\pm,*}(h_n)$, where $d^{\pm,*}(h_i)$ stands for any of the operators $a^{\pm,*}_\epsilon(h_i)$, with $h_i \in \mathfrak{h}_1$, or $b^{\pm,*}_{l,\epsilon}(h_i)$ or $c^{\pm,*}_{l,\epsilon}(h_i)$, with $h_i \in \mathfrak{h}_2$.
\item[(iii)] $\mathscr{H}_{\mathrm{pp}}(H) \subset \mathscr{K}^{\pm}$.
\end{itemize}
\end{Prop}

Recall that the wave operators $\Omega^\pm:\mathscr{H}^\pm\to\mathscr{H}$ are defined by \eqref{eq:def_Omega_1}--\eqref{eq:def_Omega_2} and that $H^{\mathrm{ext}}=H\otimes\mathds{1}+\mathds{1}\otimes H_0$. The following properties are standard consequences of the definitions.
\begin{Prop}\label{prop:intert}
Consider the Hamiltonian \eqref{total_hamiltonian_g} with $H_I$ given by \eqref{interaction_term} and suppose that
\begin{equation*}
G \in \mathbb{H}^{1+\mu} \text{ for some } \mu > 0.
\end{equation*}
Then $\Omega^\pm$ are isometric and we have that  
\begin{equation*}
H \Omega^\pm = \Omega^\pm H^{\mathrm{ext}}.
\end{equation*}
Moreover, 
\begin{equation*}
d^{\pm,\sharp}(h) \Omega^\pm = \Omega^\pm (\mathds{1} \otimes d^{\sharp}(h)),
\end{equation*}
where $d^{\pm,\sharp}(h)$ stands for any of the operators $a^{\pm,\sharp}_\epsilon(h)$, with $h \in \mathfrak{h}_1$, or $b^{\pm,\sharp}_{l,\epsilon}(h)$ or $c^{\pm,\sharp}_{l,\epsilon}(h)$, with $h \in \mathfrak{h}_2$. 
\end{Prop}
If the masses of the neutrinos are supposed to be positive, then one can show that the $\Omega^\pm$ are unitary by adapting an argument of \cite{Ho69_02} (see also \cite[Theorem 5.6]{DeGe99_01}). In the massless case, this argument fails, but one can follow the approach of \cite{DeGe00_01}. We give a sketch of the proof.
\begin{Th}\label{th:unitary}
Suppose that the masses of the neutrinos $m_{\nu_e}$, $m_{\nu_{\mu}}$, $m_{\nu_{\tau}}$ are positive and consider the Hamiltonian \eqref{total_hamiltonian_g} with $H_I$ given by \eqref{interaction_term}. Suppose that
\begin{equation*}
G \in \mathbb{H}^{1+\mu} \text{ for some } \mu > 0.
\end{equation*}
Then $\Omega^{\pm}$ are unitary maps form $\mathscr{H}^\pm$ to $\mathscr{H}$.

The same holds if the masses of the neutrinos vanish and if one considers instead the Hamiltonian $H = H_0 + g ( H_I^{(1)} + H_I^{(2)} )$ with $H_I^{(1)}$ and $H_I^{(2)}$ given by \eqref{interaction_term}.
\end{Th}
\begin{proof}
As mentioned above, the proof in the massive case is a straightforward adaptation of arguments used in \cite{Ho69_02}. We consider the case where the masses of the neutrinos vanish. By the previous proposition, $\Omega^\pm$ are isometric. It remains to verify that they are onto.

Let $\mathfrak{f}_1 \subset \mathfrak{h}_1, \mathfrak{f}_2 \subset \mathfrak{h}_2$ be two subspaces of finite dimensions. For $u \in \mathscr{H}$, let
\begin{equation*}
n^\pm_{\mathfrak{f}_1, \mathfrak{f}_2}(u) = \sum_{\epsilon=\pm}\sum^{\mathrm{dim} \mathfrak{f}_1}_{i=1} \|a^\pm_\epsilon(h_i)u\|^2 + \sum_{\epsilon=\pm,l=1,2,3}\sum^{\mathrm{dim} \mathfrak{f}_2}_{j=1} \big( \|b^\pm_{l,\epsilon}(g_j)u\|^2 + \|c^\pm_{j,\epsilon}(g_k)u\|^2 \big),
\end{equation*}
where $\{h_i\}$ and $\{g_j\}$ are orthonormal bases of $\mathfrak{f}_1$ and $\mathfrak{f}_2$, respectively. Note that if $u \notin \mathscr{D}(a^\pm_{l,\epsilon}(h_{i}))$ for some $i$, then $n^\pm_{\mathfrak{f}_1, \mathfrak{f}_2}(u) = \infty$.  Clearly,
\begin{align*}
n^\pm_{\mathfrak{f}_1, \mathfrak{f}_2}(u) & =  \lim_{t \to \infty} \sum_{\epsilon=\pm}\sum^{\mathrm{dim} \mathfrak{f}_1}_{i=1} \|a_{\epsilon}(h_{i,\pm t}) e^{\mp \mathrm{i} t H}u\|^2 + \sum_{\epsilon=\pm,l=1,2,3}\sum^{\mathrm{dim} \mathfrak{f}_2}_{j=1} \big( \|b_{l,\epsilon}(g_{j,\pm t}) e^{\mp \mathrm{i} t H}u\|^2 + \|c_{l,\epsilon}(g_{j,\pm t}) e^{\mp \mathrm{i} t H}u\|^2 \big) \\
& \le \big\langle e^{\mp\mathrm{i}tH}u , N e^{ \mp \mathrm{i} t H } u \big\rangle.
\end{align*}
Decomposing $N = N_W + 2 N_{\mathrm{lept}} + ( N_{\mathrm{neut}} - N_{\mathrm{lept}} )$ and using that $N_W$ and $N_{\mathrm{lept}}$ are relatively $H$-bounded and that $N_{\mathrm{neut}} - N_{\mathrm{lept}}$ commutes with $H$, we deduce that
\begin{equation}
n^\pm_{\mathfrak{f}_1, \mathfrak{f}_2}(u) \lesssim \langle u,|H|u\rangle + \langle u , N_{\mathrm{neut}} u \rangle + \| u \|^2. \label{eq:npm}
\end{equation}

Now, as in \cite[Theorem 4.3]{DeGe00_01}, one can verify that
\begin{equation}
\mathrm{Ran} \, \Omega^\pm = \overline{\mathscr{D}(n^\pm)} , \label{eq:npm_2}
\end{equation}
where
\begin{equation*}
n^\pm(u) = \sup_{\mathfrak{f}_1 \subset \mathfrak{h}_1,\mathfrak{f}_2 \subset \mathfrak{h}_2} n^\pm_{\mathfrak{f}_1, \mathfrak{f}_2}(u).
\end{equation*}
and $\mathscr{D}(n^\pm) = \{ u \in \mathscr{H} , n^\pm(u) < \infty \}$. In the previous equation, the supremum is taken over all finite dimensional subspaces $\mathfrak{f}_1$, $\mathfrak{f}_2$. From \eqref{eq:npm} we deduce that $\mathscr{D}(n^\pm)$ contains the dense subset $\mathscr{D}( |H|^{\frac12} ) \cap \mathscr{D}( N_{\mathrm{neut}}^{\frac12})$. By \eqref{eq:npm_2}, this shows that $\Omega^\pm$ are onto.
\end{proof}
In the remainder of this subsection we introduce the ``extended wave operators'', defined as in \cite{DeGe99_01}, and state some of their properties. Recall that  $\mathscr{H}^{\mathrm{ext}} = \mathscr{H} \otimes \mathscr{H}$ and  $H^{\mathrm{ext}} =H \otimes \mathds{1} + \mathds{1} \otimes H_0$.
We set
\begin{equation*}
\mathscr{D} (\Omega^{\mathrm{ext},\pm}) = \bigotimes^{\infty}_{n=0} \mathscr{D}((|H|^{\frac{n}{2}})\otimes \bigoplus_{n_1 + \dots + n_{14} = n} \otimes^{n_1}_{s} \mathfrak{h}_2 \otimes \otimes^{n_2}_{s} \mathfrak{h}_2 \otimes  \otimes^{n_3}_{a} \mathfrak{h}_1 \otimes \dots \otimes \otimes^{n_{14}}_{a} \mathfrak{h}_1,
\end{equation*}
and 
\begin{align}
&  \Omega^{\mathrm{ext},\pm} : \mathscr{D}(\Omega^{\mathrm{ext},\pm}) \to \mathscr{H} \notag \\
& \Omega^{\mathrm{ext},\pm} \psi \otimes d^*(h_1)\dots d^*(h_n) \Omega = d^{\pm,*}(h_1)\dots d^{\pm,*}(h_n)\psi , \label{eq:defomegaext}
\end{align} 
where, as above, $d^{\pm,*}(h_i)$ stands for any of the operators $a^{\pm,*}_\epsilon(h_i)$, with $h_i \in \mathfrak{h}_1$, or $b^{\pm,*}_{l,\epsilon}(h_i)$ or $c^{\pm,*}_{l,\epsilon}(h_i)$, with $h_i \in \mathfrak{h}_2$, and likewise for $d^*(h_i)$. The definition \eqref{eq:defomegaext} extends to any vector in $\mathscr{D}(\Omega^{\mathrm{ext},+})$ by linearity. The fact that $\Omega^{\mathrm{ext},\pm}$ are well-defined follows from the properties of the asymptotic creation operators (see Appendix \ref{sec:asympt_a}).

It follows directly from the definitions of $\Omega^\pm$ and $\Omega^{\mathrm{ext},\pm}$ that
\begin{align*}
\Omega^{\mathrm{ext}, \pm} |_{\mathscr{H}^{\pm}}  =  \Omega^{\pm} .
\end{align*}
Moreover, in the same way as in \cite[Theorem 5.7]{DeGe99_01}, it is not difficult to verify that, for all $u \in \mathscr{D}(\Omega^{\mathrm{ext},\pm})$,
\begin{equation*}
\underset{t \to \pm \infty}{\slim} e^{\mathrm{i} t H} I e^{- \mathrm{i} t H^{\mathrm{ext}}}u = \Omega^{\mathrm{ext},\pm}u  .
\end{equation*}
where the scattering identification operator $I$ is defined in \eqref{SIO}.

\subsection{The geometric inverse wave operators}\label{subsec:inverse_wave}

In this section we establish the existence of two asymptotic observables using the propagation estimates of Section \ref{sec:propag}. Compared to similar results proven in \cite{DeGe99_01,FrGrSc02_01,Am04_01}, the main difficulty we encounter comes from the fact that the propagation observables of Section \ref{sec:propag} only hold for a dense set of states, and for suitable norms. For this reason, the results of this section are not straightforward modifications of previous papers.

We begin with the following important proposition.
\begin{Prop}
\label{Thvelocity}
Suppose that the masses of the neutrinos vanish and consider the Hamiltonian $H = H_0 + g ( H_I^{(1)} + H_I^{(2)} )$ with $H_I^{(1)}$ and $H_I^{(2)}$ given by \eqref{interaction_term}. Suppose that
\begin{equation*}
G \in L^2, \quad a_{(i),\cdot} G \in L^2, \quad |p_3|^{-1} a_{(i),\cdot} G \in L^2, \quad i = 1,2,3 ,
\end{equation*}
and that
\begin{equation*}
G \in \mathbb{H}^{1+\mu} \text{ for some } \mu > 0.
\end{equation*}
Let $\delta > 0 $ and $q_i$, $i \in \{1, \dots, 14\}$ be functions in $\mathrm{C}_0^{\infty} (\{ x \in \mathbb{R}^3 , |x| \le 2 \delta \})$ such that $0 \le q_i \le 1$, $q_i=1$ on $\{ x \in \mathbb{R}^3 , |x| \leq \delta \}$ and let $\tilde{q}^t=(q_1(\frac{x_{t,\rho}}{t}),\dots,q_{14}(\frac{x_{t,\rho}}{t}))$. 
The following limits exist
\begin{equation*}
\Gamma^{\pm}(q) := \underset{t\to \pm \infty}{\slim} e^{\mathrm{i} t H} \Gamma(\tilde q^t) e^{-\mathrm{i} t H}.
\end{equation*}
Moreover, for all $\chi \in \mathrm{C}_{0}^{\infty}(\mathbb{R})$ supported in $\mathbb{R}\backslash(\tau \cup \sigma_{\mathrm{pp}}(H))$, there exists $\delta>0$ such that
\begin{equation*}
\Gamma^{\pm}(q) \chi( H ) = 0.
\end{equation*}

The same holds if the masses of the neutrinos $m_{\nu_e}$, $m_{\nu_{\mu}}$, $m_{\nu_{\tau}}$ are positive and if one considers the Hamiltonian \eqref{total_hamiltonian_g} with $H_I$ given by \eqref{interaction_term}.
\end{Prop}
\begin{proof}
We consider the more difficult case of $H=H_0 + g(H_I^{(1)}+H_I^{(2)})$ with the masses of the neutrinos equal to $0$. The proof can easily be adapted in the case of $H=H_0 + g ( H_I^{(1)}+H_I^{(2)}+H_I^{(3)}+H_I^{(4)} )$, if the masses of the neutrinos are positive.

It suffices to prove the existence of
\begin{equation*}
\lim_{t\to \pm \infty} e^{\pm \mathrm{i} t H} \Gamma(\tilde{q}^t) e^{\mp \mathrm{i} t H} u,
\end{equation*}
for $u$ in a dense subset of $\mathscr{H}$. We consider 
\begin{equation*}
\mathscr{E}:=\big\{ u \in \mathscr{H} , \exists \chi \in \mathrm{C}_0^{\infty}(\mathbb{R}) , n \in \mathbb{N} , u = \chi(H) \mathds{1}_{[-n,n]}(N_{\mathrm{neut}}-N_{\mathrm{lept}})u \big\}.
\end{equation*}
Let $u \in\mathscr{E}$ and let $\tilde \chi \in \mathrm{C}_0^\infty( \mathbb{R} )$ be such that $\tilde \chi \chi = \chi$. As in the proof of Theorem \ref{th:propag_massless_2} (iv), the Helffer-Sj{\"o}strand functional shows that
\begin{align*}
\big\|\big[\Gamma(\tilde{q}^t),\tilde{\chi}(H)\big](N_{\mathrm{neut}}+1)^{-1}\big\|=\mathcal{O}(t^{-1}).
\end{align*}
Since  $N_{\mathrm{lept}}-N_{\mathrm{neut}}$ commutes with $H$ and $\Gamma(\tilde{q}^t)$ and since $N_{\mathrm{lept}}$ is relatively $H$-bounded, we deduce that
\begin{align*}
& e^{\pm \mathrm{i} t H} \Gamma(\tilde{q}^t) e^{\mp \mathrm{i} t H} \chi(H) \mathds{1}_{[-n,n]}(N_{\mathrm{neut}}-N_{\mathrm{lept}})u \\
&= \tilde \chi( H ) e^{\pm \mathrm{i} t H} \Gamma(\tilde{q}^t) e^{\mp \mathrm{i} t H} \chi(H) \mathds{1}_{[-n,n]}(N_{\mathrm{neut}}-N_{\mathrm{lept}})u + \mathcal{O}(t^{-1}) \\
&= \mathds{1}_{[-n,n]}(N_{\mathrm{neut}}-N_{\mathrm{lept}}) \tilde \chi( H ) e^{\pm \mathrm{i} t H} \Gamma(\tilde{q}^t) e^{\mp \mathrm{i} t H} \chi(H) \mathds{1}_{[-n,n]}(N_{\mathrm{neut}}-N_{\mathrm{lept}})u + \mathcal{O}(t^{-1}).
\end{align*}
To shorten notations, let $\chi_{(n)}(H) := \chi(H) \mathds{1}_{[-n,n]}(N_{\mathrm{neut}}-N_{\mathrm{lept}})$, $\tilde\chi_{(n)}(H) := \tilde\chi(H) \mathds{1}_{[-n,n]}(N_{\mathrm{neut}}-N_{\mathrm{lept}})$. By the previous equality, it now suffices to prove the existence of 
\begin{equation*}
\lim_{t\to \pm \infty} \tilde \chi_{(n)}(H) e^{\pm \mathrm{i} t H} \Gamma(\tilde{q}^t) e^{\mp \mathrm{i} t H} \chi_{(n)}(H) u.
\end{equation*}
Set $W(t) := \tilde \chi_{(n)}(H) e^{\pm \mathrm{i} t H} \Gamma(\tilde{q}^t) e^{\mp \mathrm{i} t H} \chi_{(n)}(H)$ and write, for $t'>t\ge1$,
\begin{align}
\big\|W(t')u-W(t)u\big\|&=\Big\|\int_{t}^{t'}\partial_sW(s)uds\Big\| \le \sup_{v\in\mathscr{H},\|v\|=1} \int_{t}^{t'} | \langle v , \partial_s W(s) u \rangle | ds . \label{eq:W(t)-}
\end{align}

We compute
 \begin{align}
\langle v , \partial_t W(t) u \rangle = & \langle v , \partial_t  \tilde\chi_{(n)}(H) e^{ \pm \mathrm{i} t H }\Gamma(\tilde{q}^t) e^{ \mp \mathrm{i} t H } \chi_{(n)}(H) u \rangle \notag \\
& = \pm \langle v , \tilde\chi_{(n)}(H) e^{ \pm \mathrm{i} t H } \big( \mathbf{D}_0 \Gamma(\tilde{q}^t) + \mathrm{i} g [H_I^{(1)}+H_I^{(2)} , \Gamma(\tilde{q}^t) ] \big) e^{\mp \mathrm{i} t H } \chi_{(n)}(H) u \rangle \notag \\ 
& = \pm \langle v , \tilde\chi_{(n)}(H) e^{ \pm \mathrm{i} t H } \big( \mathrm{d}\Gamma(\tilde{q}^t, \mathbf{d}_0 \tilde{q}^t) + \mathrm{i} g [H_I^{(1)}+H_I^{(2)}, \Gamma(\tilde{q}^t) ] \big) e^{\mp \mathrm{i} t H } \chi_{(n)}(H) u \rangle. \label{eq:gamma+_1}
\end{align}
We will show that the right-hand-side is integrable in $t$ on $[1,\infty)$.

We invoke arguments closely related to those used in the proof of Theorem \ref{th:propag_massless_2} (iv). First, the assumption that $G\in\mathbb{H}^{1+\mu}$, the commutation relations of Appendix \ref{sec:standard_def} and Lemma \ref{lemSCW_1} imply  that
\begin{align*}
\big \| [H_I^{(1)} (G) + H_I^{(2)} (G) , \Gamma(\tilde{q}^t)] (N_{\mathrm{lept}}+N_W)^{-1} \big \|  =  \mathcal{O}(t^{-1-\mu}) .
\end{align*}
Since $N_W$ and $N_{\mathrm{lept}}$ are relatively $H$-bounded, this yields
\begin{align*}
\big\| [H_I^{(1)}+H_I^{(2)}, \Gamma(\tilde{q}^t) ] \chi(H) \big\| \lesssim t^{-1-\mu}.
\end{align*}

Now we consider the term involving $\mathrm{d}\Gamma(\tilde{q}^t, \mathbf{d}_0 \tilde{q}^t)$ in \eqref{eq:gamma+_1}. As in the proof of Theorem \ref{th:propag_massless_2} (iv), we have that
 \begin{align*}
\mathbf{d}_0 \tilde{q}_i^t(x_{t,\rho,i}) =  -\frac{1}{2 t} \Big \langle \frac{x_{t,\rho,i}}{t} - \nabla\omega_{t,\rho,i} ,  \nabla q_i \big( \frac{x_{t,\rho,i}}{t} \big) \Big \rangle + \mathrm{h.c.} + \mathcal{O}(t^{-2}) , 
\end{align*}
if $i\in\{1,2,3,4,7,8,11,12\}$ (corresponding to the label of a massive particle), and
\begin{align*}
\mathbf{d}_0 \tilde{q}_i^t(x_{t,\rho,i}) =  -\frac{1}{2 t} \Big \langle \frac{x_{t,\rho,i}}{t} - \nabla\omega_{t,\rho,i} ,  \nabla q_i \big( \frac{x_{t,\rho,i}}{t} \big) \Big \rangle + \mathrm{h.c.} + \mathcal{O}(t^{-2+\rho}),
\end{align*}
if $i \in \{5,6,9,10,13,14\}$ (corresponding to the label of a neutrino). We treat the second case, namely $i \in \{5,6,9,10,13,14\}$, the case of $i\in\{1,2,3,4,7,8,11,12\}$ being easier.

Let $\tilde{g}^t_i :=  -\frac{1}{2} \langle \frac{x_{t,\rho,i}}{t} - \nabla\omega_{t,\rho,i}(k) ,  \nabla q_i \big( \frac{x_{t,\rho,i}}{t} \big) \rangle + \mathrm{h.c.}$ and let $\tilde{r}^t_i = \mathbf{d}_0 \tilde{q}_i^t( x_{t,\rho,i}) - \frac{1}{t} \tilde{g}^t_i = \mathcal{O}(t^{-2+\rho})$. For the term corresponding to $\tilde{r}^t_i$, we have that
\begin{align*}
& \big\| e^{ \pm \mathrm{i} t H } \tilde\chi_{(n)}(H) \mathrm{d} \Gamma(\tilde{q}^t_i , \tilde{r}^t_i ) e^{ \mp \mathrm{i} t H } \chi_{(n)}(H) u \big\| \lesssim \mathcal{O}(t^{-2+\rho}) \big\| N e^{- \mathrm{i} t H } \chi_{(n)}(H) u \big\| = \mathcal{O}(t^{-2+\rho}).
\end{align*}
The equality comes from the facts that $N_{\mathrm{lept}}-N_{\mathrm{neut}}$ commutes with $H$, $N_W$ and $N_{\mathrm{lept}}$ are relatively $H$-bounded and that $(N+1) \tilde\chi_{(n)}(H)$ is bounded.

To estimate the term corresponding to $\frac1t\mathrm{d}\Gamma(\tilde{q}^t_i , \tilde{g}^t_i)$, we use \eqref{eq:lau_1}, yielding
\begin{align*}
&  \frac1t \big|\big\langle e^{ \mp \mathrm{i} t H } \tilde\chi_{(n)}(H) v , \mathrm{d} \Gamma(\tilde{q}^t_i , \tilde{g}^t_i) e^{ \mp \mathrm{i} t H } \chi_{(n)}(H) u \big\rangle\big| \\
& \le \frac1t \big\| \mathrm{d}\Gamma( | \tilde{g}^t_i | )^{\frac12} e^{ \mp \mathrm{i} t H } \tilde\chi_{(n)}(H) v \big \| \big \| \mathrm{d}\Gamma( | \tilde{g}^t_i | )^{\frac12} e^{ \mp \mathrm{i} t H } \chi_{(n)}(H) u \big\| .
\end{align*}
By (iii) of Theorem \ref{th:propag_massless_2},
\begin{align*}
&\int_{1}^\infty \frac1t \big\| \mathrm{d}\Gamma( | \tilde{g}^t_i | )^{\frac12} e^{ \mp \mathrm{i} t H } \tilde\chi_{(n)}(H) v \big \| \big \| \mathrm{d}\Gamma( | \tilde{g}^t_i | )^{\frac12} e^{ \mp \mathrm{i} t H } \chi_{(n)}(H) u \big\| dt \lesssim \|u\|\|v\|.
\end{align*}

From \eqref{eq:W(t)-} and the previous computations, we easily deduce that for any $\varepsilon >0$,
\begin{equation*}
\big\|W(t')u-W(t)u\big\| \le \varepsilon,
\end{equation*}
for $t$ and $t'$ large enough. This proves that the limits $\Gamma^\pm(\tilde{q})$ exist.

The fact that $\Gamma^\pm(\tilde{q})\chi(H)=0$ for all $\chi \in \mathrm{C}_{0}^{\infty}(\mathbb{R})$ supported in $\mathbb{R}\backslash(\tau \cup \sigma_{\mathrm{pp}}(H))$ is a consequence of Theorem \ref{th:propag_massless_2} (iv). Indeed, Theorem \ref{th:propag_massless_2} (iv) shows that $\Gamma^\pm(\tilde{q})\chi(H)u=0$ for all $u$ in a dense subset of $\mathscr{H}$. Since $\Gamma^\pm(\tilde{q})\chi(H)$ is bounded, the statement follows.
\end{proof}
We introduce the following notations that will be used in the proof of the next theorem:
\begin{align*}
& \check{\mathbf{d}}^{(i)}_{0l} b(t) =  \frac{\partial b}{\partial t}(t) + \mathrm{i} \big( \omega^{(i)}_l(p_i) \oplus \omega^{(i)}_l(p_i) b(t) - b(t) \omega^{(i)}_l(p_i)\big), \quad i=1,2 , \\
& \check{\mathbf{d}}^{(3)}_{0} b(t) =  \frac{\partial b}{\partial t}(t) + \mathrm{i} \big( \omega^{(3)}(p_3) \oplus \omega^{(3)}(p_3) b(t) - b(t) \omega^{(3)}(p_3)\big),
\end{align*}  
if $b(t)$ is a family of operators from $\mathfrak{h}_1$ to $\mathfrak{h}_1\oplus\mathfrak{h}_1$, if $i=1,2$, or from $\mathfrak{h}_2$ to $\mathfrak{h}_2\oplus\mathfrak{h}_2$ if $j=3$. Likewise we set
\begin{align*}
\check{\mathbf{D}}_0 B(t) & =  \frac{\partial B}{\partial t}(t) + \mathrm{i} ( H_0\otimes \mathds{1} + \mathds{1} \otimes H_0) B(t) - \mathrm{i} B(t) H_0 ,
\end{align*}  
for any family of operators $B(t) : \mathscr{H} \to \mathscr{H}^{\mathrm{ext}}$.
Note that if $B(t) = (b_1(t) , \dots , b_{14}(t) )$ then, as functions of $t$,
\begin{align*}
\check{\mathbf{D}}_0 \mathrm{d}\Gamma(B) =: \mathrm{d}\Gamma( \check{\mathbf{d}}_0 b ) := \mathrm{d}\Gamma( & \check{\mathbf{d}}_{0,1} b_1 , \dots , \check{\mathbf{d}}_{0,14} b_{14} ) \\
:=  \mathrm{d}\Gamma\big( & \check{\mathbf{d}}^{(3)}_{0} b_1 , \check{\mathbf{d}}^{(3)}_{0} b_2 , \check{\mathbf{d}}^{(1)}_{01} b_3 , \check{\mathbf{d}}^{(1)}_{01} b_4 , \check{\mathbf{d}}^{(2)}_{01} b_5 , \check{\mathbf{d}}^{(2)}_{01} b_6 , \check{\mathbf{d}}^{(1)}_{02} b_7, \\
 &\check{\mathbf{d}}^{(1)}_{02} b_8 , \check{\mathbf{d}}^{(2)}_{02} b_9 , \check{\mathbf{d}}^{(2)}_{02} b_{10} , \check{\mathbf{d}}^{(1)}_{03} b_{11} ,\check{\mathbf{d}}^{(1)}_{03} b_{12} , \check{\mathbf{d}}^{(2)}_{01} b_{13} , \check{\mathbf{d}}^{(2)}_{01} b_{14} \big).
\end{align*}

The main result of this subsection is stated in the following theorem. It shows the existence of inverse wave operators.
\begin{Th}\label{th:inverse_wave}
Suppose that the masses of the neutrinos vanish and consider the Hamiltonian $H = H_0 + g ( H_I^{(1)} + H_I^{(2)} )$ with $H_I^{(1)}$ and $H_I^{(2)}$ given by \eqref{interaction_term}. Suppose that
\begin{equation*}
G \in L^2, \quad a_{(i),\cdot} G \in L^2, \quad |p_3|^{-1} a_{(i),\cdot} G \in L^2, \quad i = 1,2,3 ,
\end{equation*}
and that
\begin{equation*}
G \in \mathbb{H}^{1+\mu} \text{ for some } \mu > 0.
\end{equation*}
Let $\delta > 0 $ and $j_{0,i}$, $i \in \{1, \dots, 14\}$ be functions in $\mathrm{C}_0^{\infty} (\{ x \in \mathbb{R}^3 , |x| \le 2 \delta \})$ such that $0 \le j_{0,i} \le 1$, $j_{0,i}=1$ on $\{ x \in \mathbb{R}^3 , |x| \leq \delta \}$ and let $j_{\infty,i} = 1 - j_{0,i}$, $j_i = ( j_{0,i} , j_{\infty,i} )$. Let $\tilde{J}^t = (j_1(\frac{x_{t,\rho,1}}{t}),\dots,j_{14}(\frac{x_{t,\rho,14}}{t}))$.
\begin{itemize}
\item[(i)] The following limits exist
\begin{align*}
W^\pm(J) := \underset{t \to +\infty}{\slim} e^{\pm \mathrm{i} t H^{\mathrm{ext}}} \check{\Gamma} \big( \tilde{J}^t \big) e^{\mp \mathrm{i} t H}.
\end{align*}
\item[(ii)] For all $\chi \in \mathrm{C}_0^\infty( \mathbb{R} )$, we have that
\begin{align*}
W^\pm(J) \chi(H) = \chi(H^{\mathrm{ext}}) W^{\pm}(J) .
\end{align*} 
\item[(iii)] Let $q = (q_1,\dots,q_{14})$ be such that $q_i j_{i,0} = j_{i,0}$. Then
\begin{align*}
\big ( \Gamma^{\pm}(q) \otimes \mathds{1} \big ) W^{\pm}(J)=W^{\pm}(J).
\end{align*}
\item[(iv)]  For all $\chi \in \mathrm{C}_0^\infty( \mathbb{R} )$, we have that
\begin{align*}
\Omega^{\mathrm{ext},\pm} \chi(H^{\mathrm{ext}}) W^{\pm}(J) = \chi(H) .
\end{align*}
\end{itemize}

The same holds if the masses of the neutrinos $m_{\nu_e}$, $m_{\nu_{\mu}}$, $m_{\nu_{\tau}}$ are positive and if one considers the Hamiltonian \eqref{total_hamiltonian_g} with $H_I$ given by \eqref{interaction_term}.
\end{Th}
\begin{proof}
(i) As in the proof of Proposition \ref{Thvelocity}, it suffices to prove the existence of
\begin{equation*}
\lim_{t\to \pm \infty} e^{\pm \mathrm{i} t H^{\mathrm{ext}}} \check{\Gamma}(\tilde{J}^t) e^{\mp \mathrm{i} t H} u,
\end{equation*}
for $u$ in a dense subset of $\mathscr{H}$. We consider again
\begin{equation*}
\mathscr{E}=\big\{ u \in \mathscr{H} , \exists \chi \in \mathrm{C}_0^{\infty}(\mathbb{R}) , n \in \mathbb{N} , u = \chi(H) \mathds{1}_{[-n,n]}(N_{\mathrm{neut}}-N_{\mathrm{lept}})u \big\} ,
\end{equation*}
and fix $u \in\mathscr{E}$. Let $\tilde \chi \in \mathrm{C}_0^\infty( \mathbb{R} )$ be such that $\tilde \chi \chi = \chi$. In the same way as in Lemma \ref{numberestimates3}, one verifies that
\begin{align*}
\big\| \big ( \check{\Gamma}(\tilde{J}^t) \tilde{\chi}(H) - \chi(H^{\mathrm{ext}}) \check{\Gamma}(\tilde{J}^t) \big ) (N_{\mathrm{neut}}+1)^{-1}\big\|=\mathcal{O}(t^{-1}).
\end{align*}
Using that  $N_{\mathrm{lept}}-N_{\mathrm{neut}}$ commutes with $H$, $\check{\Gamma}(\tilde{J}^t) \mathds{1}_{[-n,n]}(N_{\mathrm{neut}}-N_{\mathrm{lept}}) = \mathds{1}_{[-n,n]}(N_{\mathrm{neut}}^{\mathrm{ext}}-N_{\mathrm{lept}}^{\mathrm{ext}}) \check{\Gamma}(\tilde{J}^t)$, and that $N_{\mathrm{lept}}$ is relatively $H$-bounded, we deduce that
\begin{align*}
& e^{\pm \mathrm{i} t H^{\mathrm{ext}}} \check{\Gamma}(\tilde{J}^t) e^{\mp \mathrm{i} t H} \chi(H) \mathds{1}_{[-n,n]}(N_{\mathrm{neut}}-N_{\mathrm{lept}})u \\
&= \tilde \chi( H^{\mathrm{ext}} ) e^{\pm \mathrm{i} t H^{\mathrm{ext}} } \check{\Gamma}(\tilde{J}^t) e^{\mp \mathrm{i} t H} \chi(H) \mathds{1}_{[-n,n]}(N_{\mathrm{neut}}-N_{\mathrm{lept}})u + \mathcal{O}(t^{-1}) \\
&= \mathds{1}_{[-n,n]}(N_{\mathrm{neut}}^{\mathrm{ext}}-N_{\mathrm{lept}}^{\mathrm{ext}}) \tilde \chi( H^{\mathrm{ext}} ) e^{\pm \mathrm{i} t H^{\mathrm{ext}}} \check{\Gamma}(\tilde{J}^t) e^{\mp \mathrm{i} t H} \chi(H) \mathds{1}_{[-n,n]}(N_{\mathrm{neut}}-N_{\mathrm{lept}})u + \mathcal{O}(t^{-1}).
\end{align*}
Similarly as in the proof of Proposition \ref{Thvelocity}, to shorten notation, we set $\chi_{(n)}(H) := \chi(H) \mathds{1}_{[-n,n]}(N_{\mathrm{neut}}-N_{\mathrm{lept}})$ and $\tilde\chi_{(n)}(H^{\mathrm{ext}}) := \tilde\chi(H^{\mathrm{ext}}) \mathds{1}_{[-n,n]}(N_{\mathrm{neut}}^{\mathrm{ext}}-N_{\mathrm{lept}}^{\mathrm{ext}})$. By the previous equality, it now suffices to prove the existence of 
\begin{equation*}
\lim_{t\to \pm \infty} \tilde \chi_{(n)}(H^{\mathrm{ext}}) e^{\pm \mathrm{i} t H^{\mathrm{ext}}} \check{\Gamma}(\tilde{J}^t) e^{\mp \mathrm{i} t H} \chi_{(n)}(H) u.
\end{equation*}
Set $\check{W}(t) := \tilde \chi_{(n)}(H^{\mathrm{ext}}) e^{\pm \mathrm{i} t H^{\mathrm{ext}}} \check{\Gamma}(\tilde{J}^t) e^{\mp \mathrm{i} t H} \chi_{(n)}(H)$ and write, for $t'>t\ge1$,
\begin{align}
\big\|\check{W}(t')u-\check{W}(t)u\big\|&=\Big\|\int_{t}^{t'}\partial_s\check{W}(s)uds\Big\| \le \sup_{v\in\mathscr{H}^{\mathrm{ext}},\|v\|=1} \int_{t}^{t'} | \langle v , \partial_s W(s) u \rangle | ds . \label{eq:checkW(t)-}
\end{align}

We compute
 \begin{align}
&\langle v , \partial_t \check{W}(t) u \rangle = \langle v , \partial_t  \tilde\chi_{(n)}(H^{\mathrm{ext}}) e^{ \pm \mathrm{i} t H^{\mathrm{ext}} }\check{\Gamma}(\tilde{J}^t) e^{ \mp \mathrm{i} t H } \chi_{(n)}(H) u \rangle \notag \\
& = \pm \langle v , \tilde\chi_{(n)}(H^{\mathrm{ext}}) e^{ \pm \mathrm{i} t H^{\mathrm{ext}} } \big\{ \check{\mathbf{D}}_0 \check{\Gamma}(\tilde{J}^t) + \mathrm{i} g \big( (( H_I^{(1)}+H_I^{(2)} ) \otimes \mathds{1} ) \check{\Gamma}(\tilde{J}^t) - \check{\Gamma}(\tilde{J}^t) ( H_I^{(1)}+H_I^{(2)} ) \big) \big\} e^{\mp \mathrm{i} t H } \chi_{(n)}(H) u \rangle \notag \\ 
& = \pm \langle v , \tilde\chi_{(n)}(H^{\mathrm{ext}}) e^{ \pm \mathrm{i} t H^{\mathrm{ext}} } \big\{ \mathrm{d}\check{\Gamma}(\tilde{J}^t, \check{\mathbf{d}}_0 \tilde{J}^t) \notag \\
&\qquad\qquad\qquad\qquad\qquad\qquad + \mathrm{i} g \big( (( H_I^{(1)}+H_I^{(2)} ) \otimes \mathds{1} ) \check{\Gamma}(\tilde{J}^t) - \check{\Gamma}(\tilde{J}^t) ( H_I^{(1)}+H_I^{(2)} ) \big) \big\} e^{\mp \mathrm{i} t H } \chi_{(n)}(H) u \rangle. \notag 
\end{align}
As in the proof of Lemma \ref{numberestimates3}, we have that
\begin{equation*}
\tilde\chi_{(n)}(H^{\mathrm{ext}}) \big( (( H_I^{(1)}+H_I^{(2)} ) \otimes \mathds{1} ) \check{\Gamma}(\tilde{J}^t) - \check{\Gamma}(\tilde{J}^t) ( H_I^{(1)}+H_I^{(2)} ) \big) \chi_{(n)}(H) = \mathcal{O}(t^{-1-\mu}).
\end{equation*}
Moreover, similarly as in the proof of Proposition \ref{Thvelocity}, we decompose 
\begin{align*}
& \check{\mathbf{d}}_0 \tilde{J}^t = \frac{1}{t} \tilde{G}^t + \tilde{R}^t, \quad \tilde{G}^t = (\tilde{g}_0^t, \tilde{g}^t_{\infty}), \quad \tilde{g}_{\sharp}^t  = - \frac{1}{2} \Big( \Big( \frac{x_{t,\rho}}{t} - \nabla \omega_{t,\rho} \Big) \nabla \tilde{J}_{\sharp} \Big( \frac{x_{t,\rho}}{t}\Big) + \mathrm{h.c.}  \Big), \quad \tilde{R}^t = \mathcal{O}(t^{-2+\rho}),
\end{align*}
and, for all $i \in \{1,\dots,14\}$, we have that
\begin{align*}
& \big\| e^{ \pm \mathrm{i} t H^{\mathrm{ext}} } \tilde\chi_{(n)}(H^{\mathrm{ext}}) \mathrm{d} \check{\Gamma}(\tilde{J}^t_i, \tilde{R}^t_i) e^{ \mp \mathrm{i} t H } \chi_{(n)}(H) u \big\| \lesssim \mathcal{O}(t^{-2+\rho}) \big\| N e^{\mp \mathrm{i} t H } \chi_{(n)}(H) u \big\| = \mathcal{O}(t^{-2+\rho}).
\end{align*}
The term corresponding to $\frac1t\mathrm{d}\check{\Gamma}(\tilde{J}^t_i, \tilde{G}^t_i)$ is estimated as
\begin{align*}
&  \frac1t \big|\big\langle e^{ \mp \mathrm{i} t H^{\mathrm{ext}} } \tilde\chi_{(n)}(H^{\mathrm{ext}}) v , \mathrm{d} \check{\Gamma}(\tilde{J}^t_i , \tilde{G}^t_i) e^{ \mp \mathrm{i} t H } \chi_{(n)}(H) u \big\rangle\big| \\
& \le \frac1t \big\| ( \mathrm{d}\Gamma( | \tilde{g}^t_{0,i} | )^{\frac12} \otimes \mathds{1} ) e^{ \mp \mathrm{i} t H^{\mathrm{ext}} } \tilde\chi_{(n)}(H^{\mathrm{ext}}) v \big \| \big \| \mathrm{d}\Gamma( | \tilde{g}^t_{0,i} | )^{\frac12} e^{ \mp \mathrm{i} t H } \chi_{(n)}(H) u \big\|  \\
& \quad + \frac1t \big\| ( \mathds{1} \otimes \mathrm{d}\Gamma( | \tilde{g}^t_{\infty,i} | )^{\frac12} ) e^{ \mp \mathrm{i} t H^{\mathrm{ext}} } \tilde\chi_{(n)}(H^{\mathrm{ext}}) v \big \| \big \| \mathrm{d}\Gamma( | \tilde{g}^t_{\infty,i} | )^{\frac12} e^{ \mp \mathrm{i} t H } \chi_{(n)}(H) u \big\| .
\end{align*}
By (iii) of Theorem \ref{th:propag_massless_2},
\begin{align*}
&\int_{1}^\infty \frac1t \big\| ( \mathrm{d}\Gamma( | \tilde{g}^t_{0,i} | )^{\frac12} \otimes \mathds{1} ) e^{ \mp \mathrm{i} t H^{\mathrm{ext}} } \tilde\chi_{(n)}(H^{\mathrm{ext}}) v \big \| \big \| \mathrm{d}\Gamma( | \tilde{g}^t_{0,i} | )^{\frac12} e^{ \mp \mathrm{i} t H } \chi_{(n)}(H) u \big\| dt \lesssim \|u\|\|v\|,
\end{align*}
and likewise for the second term in the right-hand-side of the previous inequality. Eq. \eqref{eq:checkW(t)-} and the previous estimates imply that, for any $\varepsilon >0$,
\begin{equation*}
\big\|W(t')u-W(t)u\big\| \le \varepsilon,
\end{equation*}
for $t$ and $t'$ large enough, which proves that the limits $W^\pm(\tilde{J})$ exist.

\vspace{0,2cm}

(ii) This is a standard intertwining property.

\vspace{0,2cm}

(iii) It suffices to write
\begin{align*}
\big ( \Gamma^{\pm}(q) \otimes \mathds{1} \big ) W^{\pm}(J) u &= \big ( e^{ \pm \mathrm{i} t H} \Gamma(q^t) e^{\mp\mathrm{i} t H} \otimes \mathds{1} \big ) e^{\pm \mathrm{i} t H^{\mathrm{ext}}} \check{\Gamma} \big( \tilde{J}^t \big) e^{\mp \mathrm{i} t H} u + o(1) \\
&= \big ( e^{\pm\mathrm{i} t H^{\mathrm{ext}}} ( \Gamma(q^t)  \otimes \mathds{1} \big ) \check{\Gamma} \big( \tilde{J}^t \big) e^{\mp \mathrm{i} t H} u + o(1) \\
&=  e^{\pm\mathrm{i} t H^{\mathrm{ext}}} \check{\Gamma} \big( \tilde{J}^t \big) e^{\mp \mathrm{i} t H} u + o(1) \\
&= W^{\pm}(J) u +o(1),
\end{align*}
where we used that $( \Gamma(q^t)  \otimes \mathds{1} \big ) \check{\Gamma} \big( \tilde{J}^t \big) = \check{\Gamma} \big( \tilde{J}^t \big)$ because $q_i j_{i,0} = j_{i,0}$ in the third equality.

\vspace{0,2cm}

(iv) This is again standard intertwining property.
\end{proof} 
\subsection{Asymptotic completeness}\label{subsec:AC}
	
We are now ready to conclude the proof of Theorem \ref{thm:main}. We begin by showing that the pure point spectral subspace of $H$, $\mathscr{H}_{\mathrm{pp}}(H)$, and the spaces of asymptotic vacua $\mathscr{K}^{\pm}$ coincide. Note that our reasoning process is slightly different from that of \cite{DeGe99_01,Am04_01}.
\begin{Th}\label{th:AC}
Suppose that the masses of the neutrinos vanish and consider the Hamiltonian $H = H_0 + g ( H_I^{(1)} + H_I^{(2)} )$ with $H_I^{(1)}$ and $H_I^{(2)}$ given by \eqref{interaction_term}. Suppose that
\begin{equation*}
G \in L^2, \quad a_{(i),\cdot} G \in L^2, \quad |p_3|^{-1} a_{(i),\cdot} G \in L^2, \quad i = 1,2,3 ,
\end{equation*}
and that
\begin{equation*}
G \in \mathbb{H}^{1+\mu} \text{ for some } \mu > 0.
\end{equation*}
There exists $g_0>0$ such that, for all $|g| \le g_0$,
\begin{equation*}
\mathscr{H}_{\mathrm{pp}}(H) = \mathscr{K}^{\pm}.
\end{equation*}
The same holds, for all $g\in\mathbb{R}$, if the masses of the neutrinos $m_{\nu_e}$, $m_{\nu_{\mu}}$, $m_{\nu_{\tau}}$ are positive and if one considers the Hamiltonian \eqref{total_hamiltonian_g} with $H_I$ given by \eqref{interaction_term}.
\end{Th}
\begin{proof}
By Proposition \ref{prop:basic_prop_omega}, we know that
\begin{equation*}
\mathscr{H}_{\mathrm{pp}}(H) \subset \mathscr{K}^{\pm}.
\end{equation*}
Since in addition $\mathscr{H}_{\mathrm{pp}}(H)$ and $\mathscr{K}^{\pm}$ are closed, it remains to establish that $\mathscr{H}_{\mathrm{pp}}(H)^\perp \subset (\mathscr{K}^{\pm})^\perp$. In turn, since $\sigma_{\mathrm{pp}}(H)$ can only accumulate at the closed countable set $\tau$, it suffices to prove that $\mathrm{Ran}( \chi( H ) ) \subset (\mathscr{K}^{\pm})^\perp$ for all $\chi \in \mathrm{C}_0^{\infty}(\mathbb{R} \backslash ( \tau \cup \sigma_{\mathrm{pp}}(H) ) )$.

Let $\chi \in \mathrm{C}_0^{\infty}(\mathbb{R} \backslash ( \tau \cup \sigma_{\mathrm{pp}}(H) ) )$ and let $u = \chi(H) v$. Let $J$ be defined as in the statement of Theorem \ref{th:inverse_wave}. By Theorem \ref{th:inverse_wave} (iv), we have that
\begin{align*}
\chi(H) = \Omega^{\mathrm{ext},\pm} \chi(H^{\mathrm{ext}}) W^{\pm}(J) = \Omega^{\mathrm{ext},\pm} ( \mathds{1} \otimes \Pi_\Omega ) \chi(H^{\mathrm{ext}}) W^{\pm}(J) + \Omega^{\mathrm{ext},\pm} ( \mathds{1} \otimes P_\Omega^\perp ) \chi(H^{\mathrm{ext}}) W^{\pm}(J) ,
\end{align*}
where $\Pi_\Omega$ denotes the projection onto the Fock vacuum and $\Pi_\Omega^\perp$ the projection onto its orthogonal complement. We claim that the first term vanishes. Indeed, we can write
\begin{align*}
  \Omega^{\mathrm{ext},\pm} ( \mathds{1} \otimes \Pi_\Omega ) \chi(H^{\mathrm{ext}}) W^{\pm}(J)  & =  \Omega^{\mathrm{ext},\pm} ( \chi(H) \otimes \Pi_\Omega ) W^{\pm}(J)\\
  & =  \Omega^{\mathrm{ext},\pm} ( \chi(H) \otimes \Pi_\Omega ) (\Gamma^\pm( q ) \otimes \mathds{1} ) W^{\pm}(J) ,
\end{align*}
by Theorem \ref{th:inverse_wave} (iii), where $q$ is as in the statement of that result. Since $\chi( H )\Gamma^\pm( q )=0$ by Proposition \ref{Thvelocity}, we see that this term indeed vanishes. Hence we have proven that
\begin{equation*}
\chi( H ) v = \Omega^{\mathrm{ext},\pm} ( \mathds{1} \otimes \Pi_\Omega^\perp ) \chi(H^{\mathrm{ext}}) W^{\pm}(J) v.
\end{equation*}
Since $\Omega^{\mathrm{ext},\pm} ( \mathds{1} \otimes d^*(h) ) = d^{\pm,*}(h) \Omega^{\mathrm{ext},\pm}$ for any kind of creation operator $d^*(h)$, the last equality clearly shows that $\chi( H ) v \in ( \mathscr{K}^{\pm} )^\perp$. This concludes the proof of the theorem.
\end{proof}
Finally, as a consequence of Theorem \ref{th:AC}, we deduce that $H-E$ and $H_0$ are unitary equivalent if the conditions on $G$ are strengthened.
\begin{Cor}
Suppose that the masses of the neutrinos vanish and consider the Hamiltonian $H = H_0 + g ( H_I^{(1)} + H_I^{(2)} )$ with $H_I^{(1)}$ and $H_I^{(2)}$ given by \eqref{interaction_term}. Under the conditions of Theorem \ref{th:AC}, and assuming in addition that
\begin{equation}
 b_{(i),\cdot} G \in L^2, \quad i = 1,2,3, \quad b_{(i),\cdot} b_{(i'),\cdot} G \in L^2 , \quad i,i'=1,2,3 , \label{eq:cond_bibi'}
\end{equation}
the operators $H-E$ and $H_0$ are unitarily equivalent.

If the masses of the neutrinos are positive and if one considers the Hamiltonian \eqref{total_hamiltonian_g} with $H_I$ given by \eqref{interaction_term}, then, under the conditions of Theorem \ref{th:AC} and assuming in addition that \eqref{eq:cond_bibi'} holds, there exists $g_0>0$, which does not depend on $m_{\nu_e}$, $m_{\nu_{\mu}}$, $m_{\nu_{\tau}}$, such that, for all $|g|\le g_0$, $H-E$ and $H_0$ are unitarily equivalent.
\end{Cor}
\begin{proof}
It follows from Proposition \ref{prop:intert}, Theorem \ref{th:unitary} and Theorem \ref{th:AC} that $H$ is unitary equivalent to $H |_{ \mathscr{H}_{\mathrm{pp}}(H) } \otimes \mathds{1} + \mathds{1} \otimes H_0$. In the case where neutrinos are massive, Theorems \ref{Thess2} and \ref{MourreII} imply that, for $|g|\le g_0$, $\mathscr{H}_{\mathrm{pp}}(H)  = \{ E \}$ and $E$ is a simple eigenvalue of $H$. This shows that $H-E$ and $H_0$ are unitarily equivalent. Moreover, by Theorem \ref{MourreII}, $g_0$ can be chosen independently of the values of $m_{\nu_e}$, $m_{\nu_{\mu}}$, $m_{\nu_{\tau}}$.

The same holds in the massless case, by Theorems \ref{Thess} and \ref{MourreII_1}. This proves the corollary.
\end{proof}

\appendix

\section{The interaction term}\label{app:interaction}

In this appendix, we provide the full expression of the formal interaction Hamiltonian \eqref{Interac} in terms of creation and annihilation operators. It is given by
\begin{align*}
I =\sum_{l=1}^3 \sum_{\epsilon=\pm} \int & \left\{ \left[ G^{(2)}_{l,\epsilon} (\xi_1,\xi_2,\xi_3,x)b^*_{l,-\epsilon}(\xi_1)c^*_{l,\epsilon}(\xi_2)a^*_{\epsilon}(\xi_3) + \mathrm{h.c.}\right]\right.\\
& -   \left[ G^{(2)}_{l,\epsilon} (\xi_1,\xi_2,\xi_3,x)b^*_{l,-\epsilon}(\xi_2)c^*_{l,\epsilon}(\xi_1)a^*_{\epsilon}(\xi_3) + \mathrm{h.c.}\right]\\
&+\left[ G^{(1)}_{l,\epsilon} (\xi_1,\xi_2,\xi_3,x)b^*_{l,\epsilon}(\xi_1)c^*_{l,-\epsilon}(\xi_2)a_{\epsilon}(\xi_3) + \mathrm{h.c.} \right]\\
&-\left[ G^{(1)}_{l,\epsilon} (\xi_1,\xi_2,\xi_3,x)b^*_{l,\epsilon}(\xi_2)c^*_{l,-\epsilon}(\xi_1)a_{\epsilon}(\xi_3) +  \mathrm{h.c.} \right]\\
&+\left[ G^{(3)}_{l,\epsilon} (\xi_1,\xi_2,\xi_3,x)b^*_{l,-\epsilon}(\xi_1)c_{l,-\epsilon}(\xi_2)a^*_{\epsilon}(\xi_3) + \mathrm{h.c.} \right]\\
 &-\left[ G^{(3)}_{l,\epsilon} (\xi_1,\xi_2,\xi_3,x)b^*_{l,-\epsilon}(\xi_2)c_{l,-\epsilon}(\xi_1)a^*_{\epsilon}(\xi_3) + \mathrm{h.c.} \right]\\
 &+ \left[ G^{(4)}_{l,\epsilon} (\xi_1,\xi_2,\xi_3,x)a_{\epsilon}(\xi_3)b^*_{l,\epsilon}(\xi_1)c_{l,\epsilon}(\xi_2) + \mathrm{h.c.} \right]\\
 &- \left. \left[ G^{(4)}_{l,\epsilon} (\xi_1,\xi_2,\xi_3,x)a_{\epsilon}(\xi_3)b^*_{l,\epsilon}(\xi_2)c_{l,\epsilon}(\xi_1) +  \mathrm{h.c.} \right] \right\}  d xd \xi_1d \xi_2d \xi_3 ,
\end{align*}
where $-\epsilon = \mp$ if $\epsilon = \pm$ and
\begin{small}
\begin{equation*}
G^{(1)}_{l,\epsilon}(\xi_1, \xi_2, \xi_3,x)=(2\pi)^{-\frac{9}{2}} \left\{
\begin{array}{cl}
\frac{\overline{u(p_1,s_1)}\gamma^{\alpha}(1-\gamma_5)v(p_2,s_2)\epsilon_{\alpha}(p_3,\lambda)}{(2(|p_2|^2+m_{\nu_l}^2)^{\frac{1}{2}})^{\frac{1}{2}} (2(|p_3|^2+m_{W}^2)^{\frac{1}{2}})^{\frac{1}{2}}(2(|p_1|^2+m_l^2)^{\frac{1}{2}})^{\frac{1}{2}}}e^{\mathrm{i}(-p_1-p_2+p_3)\cdot x} & \text{if } \, \epsilon=+ , \\
&\\
\frac{\overline{u(p_2,s_2)}\gamma^{\alpha}(1-\gamma_5)v(p_1,s_1)\epsilon_{\alpha}(p_3,\lambda)}{(2(|p_2|^2+m_{\nu_l}^2)^{\frac{1}{2}})^{\frac{1}{2}} (2(|p_3|^2+m_{W}^2)^{\frac{1}{2}})^{\frac{1}{2}}(2(|p_1|^2+m_l^2)^{\frac{1}{2}})^{\frac{1}{2}}}e^{\mathrm{i}(-p_1-p_2+p_3)\cdot x}  & \text{if }  \epsilon=- , \\ 
\end{array}
\right.
\end{equation*}
\end{small}
\begin{small}
\begin{equation*}
G^{(2)}_{l,\epsilon}(\xi_1, \xi_2, \xi_3,x)=(2\pi)^{-\frac{9}{2}}  \left\{
\begin{array}{cl}
\frac{\overline{u(p_1,s_1)}\gamma^{\alpha}(1-\gamma_5)v(p_2,s_2)\epsilon^*_{\alpha}(p_3,\lambda)}{(2(|p_2|^2+m_{\nu_l}^2)^{\frac{1}{2}})^{\frac{1}{2}} (2(|p_3|^2+m_{W}^2)^{\frac{1}{2}})^{\frac{1}{2}}(2(|p_1|^2+m_l^2)^{\frac{1}{2}})^{\frac{1}{2}}}e^{-\mathrm{i}(p_1+p_2+p_3) \cdot x} &  \text{if }  \epsilon=- , \\
&\\
\frac{\overline{u(p_2,s_2)}\gamma^{\alpha}(1-\gamma_5)v(p_1,s_1)\epsilon^*_{\alpha}(p_3,\lambda)}{(2(|p_2|^2+m_{\nu_l}^2)^{\frac{1}{2}})^{\frac{1}{2}} (2(|p_3|^2+m_{W}^2)^{\frac{1}{2}})^{\frac{1}{2}}(2(|p_1|^2+m_l^2)^{\frac{1}{2}})^{\frac{1}{2}}} e^{-\mathrm{i}(p_1+p_2+p_3) \cdot x}&  \text{if }  \epsilon=+ , \\
\end{array}
\right.
\end{equation*}
\end{small}
\begin{small}
\begin{equation*}
G^{(3)}_{l,\epsilon}(\xi_1, \xi_2, \xi_3,x)=(2\pi)^{-\frac{9}{2}}  \left\{
\begin{array}{cl}
\frac{\overline{u(p_1,s_1)}\gamma^{\alpha}(1-\gamma_5)u(p_2,s_2)\epsilon^*_{\alpha}(p_3,\lambda)}{(2(|p_2|^2+m_{\nu_l}^2)^{\frac{1}{2}})^{\frac{1}{2}} (2(|p_3|^2+m_{W}^2)^{\frac{1}{2}})^{\frac{1}{2}}(2(|p_1|^2+m_l^2)^{\frac{1}{2}})^{\frac{1}{2}}}e^{\mathrm{i}(-p_1+p_2-p_3) \cdot x} &  \text{if } \epsilon=- , \\
&\\
\frac{\overline{v(p_2,s_2)}\gamma^{\alpha}(1-\gamma_5)v(p_1,s_1)\epsilon^*_{\alpha}(p_3,\lambda)}{(2(|p_2|^2+m_{\nu_l}^2)^{\frac{1}{2}})^{\frac{1}{2}} (2(|p_3|^2+m_{W}^2)^{\frac{1}{2}})^{\frac{1}{2}}(2(|p_1|^2+m_l^2)^{\frac{1}{2}})^{\frac{1}{2}}} e^{\mathrm{i}(-p_1+p_2-p_3) \cdot x}&  \text{if }  \epsilon=+ , \\
\end{array}
\right.
\end{equation*}
\end{small}
\begin{small}
\begin{equation*}
\label{fct4}
G^{(4)}_{l,\epsilon}(\xi_1, \xi_2, \xi_3,x)=(2\pi)^{-\frac{9}{2}}  \left\{
\begin{array}{c c c}
\frac{\overline{u(p_1,s_1)}\gamma^{\alpha}(1-\gamma_5)u(p_2,s_2)\epsilon_{\alpha}(p_3,\lambda)}{(2(|p_2|^2+m_{\nu_l}^2)^{\frac{1}{2}})^{\frac{1}{2}} (2(|p_3|^2+m_{W}^2)^{\frac{1}{2}})^{\frac{1}{2}}(2(|p_1|^2+m_l^2)^{\frac{1}{2}})^{\frac{1}{2}}}e^{\mathrm{i}(-p_1+p_2+p_3) \cdot x} &  \text{if } \epsilon=+ , \\
&\\
\frac{\overline{v(p_2,s_2)}\gamma^{\alpha}(1-\gamma_5)v(p_1,s_1)\epsilon_{\alpha}(p_3,\lambda)}{(2(|p_2|^2+m_{\nu_l}^2)^{\frac{1}{2}})^{\frac{1}{2}} (2(|p_3|^2+m_{W}^2)^{\frac{1}{2}})^{\frac{1}{2}}(2(|p_1|^2+m_l^2)^{\frac{1}{2}})^{\frac{1}{2}}} e^{\mathrm{i}(-p_1+p_2+p_3) \cdot x}&  \text{if }  \epsilon=-  .
\end{array}
\right.
\end{equation*}
\end{small}

From these expressions, the properties of the solutions $u$, $v$ to the Dirac equation (see, e.g., \cite{Das:2008zze} and recall that $u$ and $v$ are normalized as in \cite[(2.13)]{GrMu00_01}) and of the polarization vectors $\epsilon_\alpha$, one can verify that the maps $p_i \mapsto f^{(j)}_{l,\epsilon,i}( \xi_i )$ in \eqref{eq:kernel1}--\eqref{eq:kernel4} are bounded.

\section{Definition and properties of operators in Fock spaces}\label{sec:standard_def}

In this section, some tools and results, inspired by \cite{DeGe99_01} and \cite{Am04_01}, are presented. In particular, standard objects such as the functors $\mathrm{d}\Gamma$ and $\Gamma$ are written in the case of a finite tensor product of Fock spaces. For simplicity of exposition, the domains of the operator involved are not specified, see for instance \cite{RS} or \cite{DeGe99_01} for more details. If an operator is closable, its closure is denoted in the same way. Given a Hilbert space $\mathfrak{h}$, $\mathfrak{F}_s(\mathfrak{h})$, respectively  $\mathfrak{F}_a(\mathfrak{h})$, stands for the symmetric, respectively antisymmetric, Fock space over $\mathfrak{h}$. The notation $\mathfrak{F}_{\sharp}(\mathfrak{h})$ will be used if a statement is true for both the symmetric and antisymmetric Fock spaces over $\mathfrak{h}$.

\subsection{The operator $\Gamma(B)$}\label{subsec:Gamma}

Let $b$ be an operator on $\mathfrak{h}$. The operator $\Gamma(b):\mathfrak{F}_{\sharp}(\mathfrak{h}) \to \mathfrak{F}_{\sharp}(\mathfrak{h})$ is defined by
\begin{align*}
& \Gamma(b)|_{\otimes^n_{\sharp}\mathfrak{h}}=  \underset{n}{\underbrace{b \otimes \dots \otimes b}}, \quad \Gamma(b) \Omega = \Omega.
\end{align*}
This definition is extended to a finite tensor product of Fock spaces as follows. Let $\mathscr{H}=\mathfrak{F}_{\sharp}(\mathfrak{h}_1)\otimes \dots \otimes \mathfrak{F}_{\sharp}(\mathfrak{h}_p)$. and let $B=(b_1, b_2, \dots , b_p)$ be a finite sequence of operators, where each operator $b_i$ acts on $\mathfrak{h}_i$. We define the operator $\Gamma(B)$ on $\mathscr{H}$ by setting
\begin{align*}
\Gamma(B) :=  \Gamma(b_1)\otimes \dots \otimes\Gamma(b_p).
\end{align*}

\subsection{The operator $\mathrm{d}\Gamma(B)$}\label{Gensecqu}

Let $b$ be an operator on $\mathfrak{h}$. Its second quantization $\mathrm{d}\Gamma(B)$ is defined on $\mathfrak{F}_{\sharp}(\mathfrak{h})$ by
\begin{align*}
&\mathrm{d}\Gamma(b)|_{\otimes^n_{\sharp}\mathfrak{h}}  =  \overset{n}{\underset{j=1}{\sum}} \,   \underset{j-1}{\underbrace{ \mathds{1}\otimes \dots \otimes \mathds{1}} }\otimes B \otimes \underset{n-j}{\underbrace{ \mathds{1} \otimes \dots \otimes \mathds{1}} }, \quad \mathrm{d}\Gamma(B) \Omega  =  0. 
\end{align*}
In particular, the number operator $N$ is $N=\mathrm{d}\Gamma(\mathds{1})$. We recall from \cite{DeGe99_01,Am04_01} that
\begin{equation}
\| N^{-\frac{1}{2}} \mathrm{d}\Gamma(b) u \| \leq \| \mathrm{d}\Gamma(b^*b)^{\frac{1}{2}} u \| , \label{eq:recall_estim}
\end{equation}
for any $u \in \mathscr{D}( \mathrm{d}\Gamma(b^*b)^{ 1/2 } )$.

The definition of $\mathrm{d}\Gamma(b)$ is extended to a finite tensor product of Fock spaces as follows. Let $\mathscr{H}=\mathfrak{F}_{\sharp}(\mathfrak{h}_1)\otimes \dots \otimes \mathfrak{F}_{\sharp}(\mathfrak{h}_p)$. Let $B=(b_1, b_2, \dots , b_p)$ be a finite sequence of operators, with $b_i$ acting on $\mathfrak{h}_i$. We define $\mathrm{d}\Gamma(B)$ on $\mathscr{H}$ by setting
\begin{align*}
\mathrm{d}\Gamma(B) := \sum_{j=1}^p \mathrm{d} \Gamma_j( b_j ) := \overset{p}{\underset{j=1}{\sum}} \,   \underset{j-1}{\underbrace{ \mathds{1} \otimes \dots \otimes \mathds{1}} }\otimes \mathrm{d}\Gamma(b_j) \otimes \underset{p-j}{\underbrace{ \mathds{1} \otimes \dots \otimes \mathds{1}} }.
\end{align*}
Then one can define the total number operator in $\mathscr{H}$ by $N := \mathrm{d}\Gamma( (\mathds{1}, \dots, \mathds{1}) )$ and \eqref{eq:recall_estim} becomes
\begin{equation*}
\| N^{-\frac{1}{2}} \mathrm{d}\Gamma(B) u \| \leq \| \mathrm{d}\Gamma(B^*B)^{\frac{1}{2}} u \|.
\end{equation*}

\subsection{The unitary operator $U$}\label{subsec:unitary_op}

We consider the specific Hilbert space of our model, defined by
\begin{equation*}
\mathscr{H}(\mathfrak{h}_1,\mathfrak{h}_2)=\mathfrak{F}_W \otimes \mathfrak{F}_L = \mathfrak{F}_s(\mathfrak{h}_2) \otimes \mathfrak{F}_s(\mathfrak{h}_2) \otimes \underset{12}{\underbrace{\mathfrak{F}_a(\mathfrak{h}_1)\otimes \dots \otimes \mathfrak{F}_a(\mathfrak{h}_1)}},
\end{equation*}
where $\mathfrak{h}_1$ and $\mathfrak{h}_2$ are the one-fermion Hilbert space, respectively the one-boson Hilbert space, see Section \ref{2intro}. In this section, we define an analog of the unitary operators $U$ considered in \cite{DeGe99_01} (in the bosonic case) and \cite{Am04_01} (in the fermionic case). In our setting, the operator $U$ is an operator from $\mathscr{H}(\mathfrak{h}_1 \oplus \mathfrak{h}_1,\mathfrak{h}_2 \oplus \mathfrak{h}_2)$ to $\mathscr{H}(\mathfrak{h}_1,\mathfrak{h}_2) \otimes \mathscr{H}(\mathfrak{h}_1,\mathfrak{h}_2)$,
\begin{equation*}
U : \mathscr{H}(\mathfrak{h}_1 \oplus \mathfrak{h}_1,\mathfrak{h}_2 \oplus \mathfrak{h}_2) \rightarrow \mathscr{H}(\mathfrak{h}_1,\mathfrak{h}_2) \otimes \mathscr{H}(\mathfrak{h}_1,\mathfrak{h}_2).
\end{equation*}

We first define an operator $U_{s_1}$ associated to the first bosonic Fock space by setting
\begin{align*}
& U_{s_1} : \mathfrak{F}_s(\mathfrak{h}_2 \oplus \mathfrak{h}_2) \otimes \mathfrak{F}_s(\mathfrak{h}_2) \otimes \underset{12}{\underbrace{\mathfrak{F}_a(\mathfrak{h}_1)\otimes \dots \otimes \mathfrak{F}_a(\mathfrak{h}_1)}} \\
&\qquad \rightarrow \mathfrak{F}_s(\mathfrak{h}_2) \otimes \mathfrak{F}_s(\mathfrak{h}_2) \otimes \underset{12}{\underbrace{\mathfrak{F}_a(\mathfrak{h}_1)\otimes \dots \otimes \mathfrak{F}_a(\mathfrak{h}_1)}}\otimes \mathfrak{F}_s(\mathfrak{h}_2) \\
& U_{s_1}\Omega \otimes \Omega \otimes \underset{12}{\underbrace{\Omega \otimes \dots \otimes \Omega}} = \Omega \otimes \Omega \otimes \underset{12}{\underbrace{\Omega \otimes \dots \otimes \Omega}}\otimes \Omega\\
& U_{s_1} \Big(a^{\sharp}(h_1+h_3) \otimes b \Big) = \Big(a^{\sharp}(h_1) \otimes \psi \otimes \mathds{1}+\mathds{1} \otimes b \otimes a^{\sharp}(h_3)\Big) U_{s_1},
\end{align*}
for any operator $b$ acting on $\mathfrak{F}_s(\mathfrak{h}_2) \otimes \underset{12}{\underbrace{\mathfrak{F}_a(\mathfrak{h}_1)\otimes \dots \otimes \mathfrak{F}_a(\mathfrak{h}_1)}}$.

The operator $U_{s_2}$ associated to the second bosonic Fock space is defined similarly by
\begin{align*}
& U_{s_2} : \mathfrak{F}_s(\mathfrak{h}_2) \otimes \mathfrak{F}_s(\mathfrak{h}_2 \oplus \mathfrak{h}_2) \otimes \underset{12}{\underbrace{\mathfrak{F}_a(\mathfrak{h}_1)\otimes \dots \otimes \mathfrak{F}_a(\mathfrak{h}_1)}} \otimes \mathfrak{F}_s(\mathfrak{h}_2)   \\
&\qquad \rightarrow  \mathfrak{F}_s(\mathfrak{h}_2) \otimes \mathfrak{F}_s(\mathfrak{h}_2) \otimes \underset{12}{\underbrace{\mathfrak{F}_a(\mathfrak{h}_1)\otimes \dots \otimes \mathfrak{F}_a(\mathfrak{h}_1)}} \otimes \mathfrak{F}_s(\mathfrak{h}_2) \otimes \mathfrak{F}_s(\mathfrak{h}_2) \\
&U_{s_2}\Omega \otimes \Omega \otimes \underset{12}{\underbrace{\Omega \otimes \dots \otimes \Omega}} \otimes \Omega = \Omega \otimes \Omega \otimes \underset{12}{\underbrace{\Omega \otimes \dots \otimes \Omega}}\otimes \Omega \otimes \Omega\\
& U_{s_2} \Big( b_1 \otimes a^{\sharp}(h_1+h_3) \otimes b_2 \Big) = \Big( b_1\otimes a^{\sharp}(h_1) \otimes b_2 \otimes \mathds{1}+ b_1 \otimes \mathds{1} \otimes b_2 \otimes a^{\sharp}(h_3) \Big) U_{s_2} ,
\end{align*}
for any operators $b_1$ acting on $\mathfrak{F}_s(\mathfrak{h}_2)$ and $b_2$ acting on $\underset{12}{\underbrace{\mathfrak{F}_a(\mathfrak{h}_1)\otimes \dots \otimes \mathfrak{F}_a(\mathfrak{h}_1)}} \otimes \mathfrak{F}_s(\mathfrak{h}_2)$.

Following \cite{Am04_01}, the unitary operator associated to the fermionic Fock spaces are defined as follows. Two different operators may be considered. They are, however, equal up to the left multiplication by an operator of the type $(-\mathds{1})^{B}$ where $B$ is a finite tensor product of number operators. The operators $U_{a_3,l/r}$ associated to the first fermionic Fock space are defined by
\begin{align*}
& U_{a_3,l/r} : \mathfrak{F}_s(\mathfrak{h}_2) \otimes \underset{12}{\underbrace{\mathfrak{F}_a(\mathfrak{h}_1 \oplus \mathfrak{h}_1)\otimes \dots \otimes \mathfrak{F}_a(\mathfrak{h}_1)}} \otimes \mathfrak{F}_s(\mathfrak{h}_2) \\
& \qquad \rightarrow \mathfrak{F}_s(\mathfrak{h}_2) \otimes \underset{12}{\underbrace{\mathfrak{F}_a(\mathfrak{h}_1)\otimes \dots \otimes \mathfrak{F}_a(\mathfrak{h}_1)}} \otimes \mathfrak{F}_s(\mathfrak{h}_2) \otimes \mathfrak{F}_a(\mathfrak{h}_1) \\
& U_{a_3,l/r} \Omega \otimes \underset{12}{\underbrace{\Omega \otimes \dots \otimes \Omega}} \otimes \Omega =  \Omega \otimes \underset{12}{\underbrace{\Omega \otimes \dots \otimes \Omega}}\otimes \Omega \otimes \Omega\\
& U_{a_3,l} \Big(b_1\otimes b^{\sharp}(h_2+h_4) \otimes b_2 \Big) = \Big( b_1\otimes b^{\sharp}(h_2) \otimes b_2 \otimes \mathds{1} + b_1\otimes (-\mathds{1})^{N} \otimes b_2 \otimes b^{\sharp}(h_4) \Big) U_{a_3,l}\\
& U_{a_3,r} \Big(b_1\otimes b^{\sharp}(h_2+h_4) \otimes b_2 \Big) = \Big( b_1\otimes b^{\sharp}(h_2) \otimes b_2 \otimes (-\mathds{1})^{N} + b_1\otimes \mathds{1} \otimes b_2 \otimes b^{\sharp}(h_4) \Big) U_{a_3,r},
\end{align*}
for any operators $b_1$ acting on $\mathfrak{F}_s(\mathfrak{h}_2)$ and $b_2$ acting on $\underset{11}{\underbrace{\mathfrak{F}_a(\mathfrak{h}_1)\otimes \dots \otimes \mathfrak{F}_a(\mathfrak{h}_1)}} \otimes \mathfrak{F}_s(\mathfrak{h}_2)$.

The operators $U_{s_1}$, $U_{s_2}$, $U_{a_3,l/r}$ extend to unitary operators. One can define similarly unitary operators $U_{a_4,l/r}$ \dots $U_{a_{14},l/r}$. The unitary operators $U_{L/R}$ are then defined by
\begin{align*}
& U_{L,R} : \mathscr{H}(\mathfrak{h}_1\oplus \mathfrak{h}_1, \mathfrak{h}_2 \oplus \mathfrak{h}_2) \rightarrow \mathscr{H}(\mathfrak{h}_1, \mathfrak{h}_2) \otimes \mathscr{H}(\mathfrak{h}_1, \mathfrak{h}_2)\\
&U_L = U_{a_{14},l} \cdots U_{a_3,l} U_{s_2} U_{s_1} , \qquad U_R = U_{a_{14},r} \cdots U_{a_3,r} U_{s_2} U_{s_1}. 
\end{align*}

\subsection{Partition of unity of the total Hilbert space and scattering identification operator}\label{EGFSIO}

Let $j_0, j_{\infty}$ be operators on $\mathfrak{h}_i$, $i=1,2$. As in \cite{DeGe99_01,Am04_01}, we define an operator $j$ associated to $j_0$ and $j_\infty$ as
\begin{align*}
 j ~ : ~ \mathfrak{h}_i & \rightarrow   \mathfrak{h}_i \oplus \mathfrak{h}_i\\
 h & \mapsto  (j_0 h , j_{\infty} h ).
\end{align*}
It follows that
\begin{align*}
j^*~:~\mathfrak{h}_i \oplus \mathfrak{h}_i & \rightarrow \mathfrak{h}_i \\
(h_0,h_{\infty})& \mapsto  j^*_0 h_0 + j^*_{\infty} h_{\infty}.
\end{align*} 
Considering $J=(j_1,\dots,j_{14})$ a family of such maps, we consider the operator $\Gamma( J )$ (see Section \ref{subsec:Gamma}),
\begin{align*}
\Gamma(J)~:~\mathscr{H}(\mathfrak{h}_1,\mathfrak{h}_2) & \rightarrow \mathscr{H}(\mathfrak{h}_1\oplus \mathfrak{h}_1,\mathfrak{h}_2 \oplus \mathfrak{h}_2).
\end{align*}
The partition of unity $\check{\Gamma}(J)$ is then defined by
\begin{align*}
& \check{\Gamma}(J)~:~\mathscr{H}(\mathfrak{h}_1,\mathfrak{h}_2) \rightarrow \mathscr{H}(\mathfrak{h}_1, \mathfrak{h}_2) \otimes \mathscr{H}(\mathfrak{h}_1, \mathfrak{h}_2)\\
& \check{\Gamma}(J) = U_L \Gamma(J) ,
\end{align*}
where $U_L$ is the unitary operator of Section \ref{subsec:unitary_op}.

Now, let 
\begin{align*}
i~:~\mathfrak{h}_i \oplus \mathfrak{h}_i & \to \mathfrak{h}_i \\
(h_0, h_{\infty}) & \mapsto h_0 + h_{\infty}.
\end{align*}
Considering $\hat{i}=(i_1,\dots,i_{14})$ where each $i_k$ is defined by the previous identity, we define the scattering identification operator $I: \mathscr{H}(\mathfrak{h}_1, \mathfrak{h}_2) \otimes \mathscr{H}(\mathfrak{h}_1, \mathfrak{h}_2) \to \mathscr{H}(\mathfrak{h}_1,\mathfrak{h}_2)$ by
\begin{equation}
\label{SIO}
I=\Gamma(\hat{i})U_L^* = \check{\Gamma}(\hat{i}^*)^*.
\end{equation}

\subsection{Extended objects}

Recall that the total Hilbert space of our model is denoted by $\mathscr{H}=\mathscr{H}(\mathfrak{h}_1,\mathfrak{h}_2)$. As mentioned in the main text, the ``extended Hilbert space'' in our setting is defined by
\begin{equation*}
\mathscr{H}^{\mathrm{ext}} = \mathscr{H} \otimes \mathscr{H}.
\end{equation*}
In $\mathscr{H}^{\mathrm{ext}}$, one defines the number operators
\begin{equation}\label{eq:defN0Ninfty}
N_0= N \otimes \mathds{1}_\mathscr{H}, \qquad N_{\infty}= \mathds{1}_\mathscr{H} \otimes N .
\end{equation}
The ``extended Hamiltonian'' and ``extended free Hamiltonian'' are
\begin{align*}
&H^{\mathrm{ext}} = H \otimes \mathds{1}_\mathscr{H} + \mathds{1}_\mathscr{H} \otimes H_0,\\
&H^{\mathrm{ext}}_0 = H_0 \otimes \mathds{1}_\mathscr{H} + \mathds{1}_\mathscr{H} \otimes H_0.
\end{align*}

\subsection{The operators $\mathrm{d}\Gamma (Q,R)$ and $\mathrm{d}\check{\Gamma} (Q,R)$}

Let $q,r$ be two operators on $\mathfrak{h}_i$. The operator $\mathrm{d}\Gamma(q,r):\mathfrak{F}_\sharp(\mathfrak{h}_i) \rightarrow \mathfrak{F}_\sharp(\mathfrak{h}_i)$ considered in \cite{DeGe99_01,Am04_01} is defined by
\begin{align*}
\mathrm{d}\Gamma(q,r)|_{\otimes^n\mathfrak{h}_i}  =  \overset{n}{\underset{j=1}{\sum}}\underset{j-1}{\underbrace{q \otimes \dots \otimes q}}  \otimes r \otimes \underset{n-j}{\underbrace{q \otimes \dots \otimes q}}.
\end{align*}
Given $q,r,s$ three operators in $\mathfrak{h}_i$, with $\|q\|\le1$, the following estimates are proven in \cite{DeGe99_01,Am04_01}:
\begin{equation}
 \big|\big\langle\mathrm{d}\Gamma(q,rs)u,v\big\rangle\big| \leq \big\| \mathrm{d}\Gamma(r^* r)^{\frac{1}{2}} v \big\| \big\| \mathrm{d}\Gamma(s^* s)^{\frac{1}{2}} u\big\| , \label{eq:lau_1}
\end{equation}
for all $u \in \mathscr{D}( \mathrm{d}\Gamma(r^* r)^{1/2} )$ and $v \in \mathscr{D}( \mathrm{d}\Gamma( s^*s )^{1/2} )$, and
\begin{equation}\label{eq:fhza1}
\big\|N^{-\frac{1}{2}} \mathrm{d}\Gamma(q,r)u \| \leq \| \mathrm{d}\Gamma(r^* r)^{\frac{1}{2}} u\| ,
\end{equation}
for all $u \in \mathscr{D}( \mathrm{d}\Gamma(r^* r)^{1/2} )$.

Let now $Q=(q_1,\dots,q_{14})$ and $R=(r_1,\dots,r_{14})$ be defined as in Section \ref{Gensecqu}. The operator $\mathrm{d}\Gamma(Q,R):\mathscr{H} \rightarrow \mathscr{H}$ is defined by
\begin{align*}
\mathrm{d}\Gamma(Q,R) = \sum_{i=1}^{14} \Gamma\big( (q_1,\dots, q_{i-1})\big) \otimes \mathrm{d}\Gamma(q_i,r_i) \otimes \Gamma\big( (q_{i+1},\dots, q_{14})\big).
\end{align*}
Moreover, similarly as in Section \ref{EGFSIO}, we define
\begin{align*}
& \mathrm{d}\check{\Gamma}(Q,R)~:~\mathscr{H}(\mathfrak{h}_1,\mathfrak{h}_2) \rightarrow \mathscr{H}(\mathfrak{h}_1, \mathfrak{h}_2) \otimes \mathscr{H}(\mathfrak{h}_1, \mathfrak{h}_2)\\
& \mathrm{d}\check{\Gamma}(J) = U_L \Gamma(J) ,
\end{align*}
where $U_L$ is the unitary operator of Section \ref{subsec:unitary_op}.

With these definitions, the estimate recalled in \eqref{eq:fhza1} easily generalizes to the following lemma.
\begin{lem}
\label{estimationN2}
Let $Q=(q_1, \dots, q_{14})$ and $R=(r_1,\dots,r_{14})$ be finite sequences of operators, with $\|q_i \| \leq 1$. We have that
\begin{equation*}
\big\|N^{-\frac12} \mathrm{d}\Gamma(Q,R) u \big \| \leq \big\| \mathrm{d}\Gamma(R^* R)^{\frac{1}{2}} u\big\| ,
\end{equation*}
for all $u \in \mathscr{D}( \mathrm{d}\Gamma(R^* R)^{1/2} )$, where $N = \mathrm{d}\Gamma( (\mathds{1},\dots,\mathds{1}) )$ is the total number operator in $\mathscr{H}$. Moreover,
\begin{equation*}
\big \| ( N_{0}+N_{\infty} ) )^{-\frac{1}{2}} \mathrm{d}\check{\Gamma}(Q,R) u \big\| \leq \big\| \mathrm{d}\check{\Gamma}(R^* R)^{\frac{1}{2}} u \big\|
\end{equation*}
for al $u \in \mathscr{D}( \mathrm{d}\check{\Gamma}(R^* R)^{1/2} )$, where $N_0$ and $N_\infty$ are defined in \eqref{eq:defN0Ninfty}.
\end{lem}
\subsection{Intertwining property }

In this section we state some intertwining properties that we used throughout the main text.

In the next lemma, $a^{\sharp}_{i}$ stands for a bosonic creation or annihilation operator acting on the $i^{th}$ Fock space (note that $\mathscr{H}$ is the tensor product of $14$ Fock spaces, and hence $\mathscr{H} \otimes \mathscr{H}$ is the tensor product of $28$ Fock spaces). Likewise, $d^{\sharp}_{i,b}$ stands for a fermionic creation or annihilation operator acting on the $i^{\mathrm{th}}$ Fock space. 
\begin{lem}\label{lem:B6}
Let $J=\{(j_{1,0},j_{1,\infty}),\dots, (j_{14,0},j_{14,\infty})  \}$ be a family of operators defined as in Section \ref{EGFSIO}. For $i=1,2$, and $h_2 \in \mathfrak{h}_2$, we have that
\begin{align*}
\check{\Gamma}(J) \int h_2(\xi) a^{\sharp}_i(\xi) d\xi  = \int \left\{ (j_{i,0}h_2)(\xi) a^{\sharp}_i(\xi)+ ( j_{i,\infty}h_2 )(\xi) a^{\sharp}_{i+14}(\xi)\right\} d\xi \, \check{\Gamma}(J).
\end{align*}
Likewise, for $i =3 , \dots , 14$ and $h_1 \in \mathfrak{h}_1$, we have that
\begin{align*}
\check{\Gamma}(J) \int h_1(\xi) d^{\sharp}_i(\xi) d\xi = \int \left\{ (j_{i,0}h_1)(\xi)d^{\sharp}_{i,b}(\xi)  + (j_{i,\infty}h_1)(\xi) \left(- 1\right)^{N_{i}} d^{\sharp}_{i+14,b}(\xi) \right\} d\xi \, \check{\Gamma}(J) ,
\end{align*}
where we have set $N_i= \int d^*_{i,b}(\xi) d_{i,b}(\xi) d\xi$.
\end{lem}
\begin{proof}
We prove for instance the first intertwining property with $i=1$ and $a^\sharp = a^*$, the proof of the other statements is analogous. Recall the notation $j_i=(j_{i,0},j_{i,\infty})$. We have that
\begin{equation*}
\Gamma\left(j_1\right) a^{\sharp}(h_2) = a^{\sharp}(j_1 h_2) \Gamma\left(j_1\right) ,
\end{equation*}
(see \cite{DeGe99_01}). Therefore,
\begin{align*}
\check{\Gamma}(J) a^*(h_2) & = U_L \Gamma(J) \left\{a^*(h_2)\otimes \mathds{1}_{\mathfrak{F}_s} \otimes \mathds{1}_{\mathfrak{F}_L}\right\}\\
 & =  U_L  \left\{\Gamma(j_1) a^*(h_2)\otimes \Gamma(j_2)  \otimes \Gamma\left(\{ j_3,\dots, j_{14} \}\right) \right\}\\
 & =  U_L  \left\{a^*(j_1 h_2) \Gamma(j_1)\otimes \Gamma(j_2)  \otimes \Gamma\left(\{ j_3,\dots, j_{14} \}\right)\right\}\\
  & = U_L  \left\{a^*(j_1 h_2)\otimes \mathds{1}_{\mathfrak{F}_s} \otimes \mathds{1}_{\mathfrak{F}_L}\right\}\Gamma(J)\\
 & =    \left\{a^*(j_{1,0} h_2)\otimes \mathds{1}_{\mathfrak{F}_s} \otimes \mathds{1}_{\mathfrak{F}_L} \otimes \mathds{1}_\mathscr{H} +\mathds{1}_\mathscr{H} \otimes a^*(j_{1,\infty} h_2)\otimes \mathds{1}_{\mathfrak{F}_s} \otimes \mathds{1}_{\mathfrak{F}_L} \right\} \check{ \Gamma}(J).
\end{align*}
This corresponds to the first equality in the statement of the lemma, for $i=1$ and $a^\sharp=a^*$.
\end{proof} 
We conclude this appendix with another useful intertwining property.
\begin{lem}\label{lm:B3}
Let $B=(b_1,\dots,b_{14})$ be a finite sequence of operators defined as in Section \ref{Gensecqu} and let $J=\{(j_{1,0},j_{1,\infty}),\dots, (j_{14,0},j_{14,\infty})  \}$ be a family of operators defined as in Section \ref{EGFSIO}. We have that
\begin{equation*}
\left( \mathrm{d}\Gamma(B) \otimes \mathds{1}_{\mathscr{H}} + \mathds{1}_{\mathscr{H}} \otimes \mathrm{d}\Gamma(B)  \right)  \check{\Gamma}(J) -  \check{\Gamma}(J) \mathrm{d}\Gamma(B) =\mathrm{d}\check{\Gamma}(J,[B,J]).
\end{equation*}
\end{lem}
\begin{proof}
It suffices to write
\begin{align*}
& \left( \mathrm{d}\Gamma(B) \otimes \mathds{1}_{\mathscr{H}} + \mathds{1}_{\mathscr{H}} \otimes \mathrm{d}\Gamma(B)  \right)  \check{\Gamma}(J) -  \check{\Gamma}(J) \mathrm{d}\Gamma(B) \\
& = \left( \mathrm{d}\Gamma(B) \otimes \mathds{1}_{\mathscr{H}} + \mathds{1}_{\mathscr{H}} \otimes \mathrm{d}\Gamma(B)  \right)  U_L \Gamma(J) - U_L \Gamma(J) \mathrm{d}\Gamma(B) \\
& = U_L\left\{ \mathrm{d}\Gamma\left(\left[ 
\begin{array}{cc}
B & 0\\
0 & B
\end{array}
\right]
\right) \Gamma(J) -  \Gamma(J) \mathrm{d}\Gamma(B)\right\} \\
&= U_L\overset{n}{\underset{i=1}{\sum}}~\left\{ \Gamma((j_1,\dots , j_{i-1}))\otimes \left(\mathrm{d}\Gamma\left(\left[ 
\begin{array}{cc}
b_i & 0\\
0 & b_i
\end{array}
\right]
\right)\Gamma(j_i) - \Gamma(j_i) \mathrm{d}\Gamma \left(b_i\right)\right) \otimes \Gamma((j_{i+1}, \dots, j_{n}))   \right\} \\
& = U_L\overset{n}{\underset{i=1}{\sum}}~\left\{ \Gamma((j_1,\dots , j_{i-1}))\otimes \mathrm{d}\Gamma(j_i, [b_i,j_i]) \otimes \overline{\Gamma}((j_{i+1}, \dots, j_{n}))   \right\} \\ 
& = \mathrm{d}\check{\Gamma}(J,[B,J]). 
\end{align*}
This proves the lemma.
\end{proof}
\section{Asymptotic creation and annihilation operators}\label{sec:asympt_a}

In this section, we recall the existence and basic properties of the asymptotic creation and annihilation operators.

For $h_2 \in \mathfrak{h}_2$ we define $h^{(3)}_{2,t}:=e^{- \mathrm{i} t \omega^{(3)}} h_2$. Likewise, for $h_1 \in \mathfrak{h}_1$, we set $h^{(2)}_{1,t}:=e^{- \mathrm{i} t \omega^{(2)}} h_1$ and $h^{(1)}_{1,t}:=e^{-\mathrm{i} t \omega^{(1)}} h_1 $. For all $h_1 \in \mathfrak{h}_1$, $h_2 \in \mathfrak{h}_2$ and $\epsilon=\pm$, we introduce the following notations
\begin{align*}
& \phi^{(a)}_{\epsilon}(h_2) = \frac{1}{\sqrt{2}} \left( a_{\epsilon}^{*}(h_2) + a_{\epsilon}(h_2) \right), \\
& \phi^{(b)}_{l,\epsilon}(h_1) = \frac{1}{\sqrt{2}} \left( b_{l,\epsilon}^{*}(h_1) + b_{l,\epsilon}(h_1) \right), \\
& \phi^{(c)}_{l,\epsilon}(h_1) = \frac{1}{\sqrt{2}} \left( c_{l,\epsilon}^{*}(h_1) + c_{l,\epsilon}(h_1) \right).
\end{align*}
Assuming that $G \in \mathbb{H}^{1+\mu}$ for some $\mu>0$, the asymptotic bosonic fields can be defined in the same way as in \cite[Section 5.2]{DeGe99_01}, as generators of the asymptotic Weyl operators. The latter are defined as the strong limits
\begin{equation*}
W^{(a),+}_{\epsilon} (h_2) = \underset{t \to +\infty}{\slim} e^{\mathrm{i} t H } e^{\mathrm{i}t \phi^{(a)}_{\epsilon}(h_{2,t})} e^{ - \mathrm{i} t H } .
\end{equation*}
The asymptotic fermionic fields can be defined similarly, or, equivalently, as the strong limits
\begin{align*}
& \phi^{(b),+}_{l,\epsilon} (h_1) = \underset{t \to +\infty}{\slim} e^{ \mathrm{i} t H } \phi^{(b)}_{l,\epsilon}(h_{1,t}) e^{ - \mathrm{i} t H } , \quad \phi^{(c),+}_{l,\epsilon} (h_1) = \underset{t \to +\infty}{\slim} e^{ \mathrm{i} t H } \phi^{(c)}_{l,\epsilon}(h_{1,t}) e^{ - \mathrm{i} t H }.
\end{align*}

The results stated in the next theorem are straightforward adaptations of corresponding results established in \cite{DeGe99_01,Am04_01}. Details of the proof are left to the reader.
\begin{Th}
\label{TH7}
Suppose that $G \in \mathbb{H}^{1+\mu}$ for some $\mu>0$.
\begin{enumerate}[i)]
\item For any $h_2 \in \mathfrak{h}_2$, the asymptotic bosonic creation and annihilation operators $a_{\epsilon}^{+\sharp}(h_2)$ defined on $\mathscr{D}(a_{\epsilon}^{+\sharp}(h_2)) = \mathscr{D}(\phi^{(a),+}_{\epsilon} (h_2)) \cap \mathscr{D}(\phi^{(a),+}_{\epsilon} (\mathrm{i}h_2))$  by
\begin{align*}
a_{\epsilon}^{+*}(h_2) =  \frac{1}{\sqrt{2}} (\phi_{\epsilon}^{+}(h_1)- \mathrm{i}\phi_{\epsilon}^{+}(\mathrm{i}h_2)), \qquad a_{\epsilon}^{+}(h_2)  =  \frac{1}{\sqrt{2}} (\phi_{\epsilon}^{+}(h_2)+\mathrm{i}\phi_{\epsilon}^{+}(\mathrm{i}h_2)).
\end{align*}
are closed operators. Moreover, we have that $\mathscr{D}((|H|+1)^{\frac{1}{2}}) \subset \mathscr{D}(a_{\epsilon}^{+\sharp}(h_2))$ and 
\begin{align*}
& \| a^{+ \sharp}(h_2) u \| \leq C \| h_2 \| \| (|H|+1)^{\frac{1}{2}} u \| .
\end{align*}
\item For any $h_1 \in \mathfrak{h}_1$, the asymptotic fermionic creation and annihilation operators defined by
\begin{align*}
&b_{l,\epsilon}^{+*}(h_1) = \frac{1}{\sqrt{2}} (\phi_{l,\epsilon}^{(b),+}(h_1) - \mathrm{i} \phi_{l,\epsilon}^{(b),+}(\mathrm{i}h_1)) , \qquad b_{l,\epsilon}^{+}(h_1) = \frac{1}{\sqrt{2}} (\phi_{l,\epsilon}^{(b),+}(h_1)+\mathrm{i}\phi_{l,\epsilon}^{(b),+}(\mathrm{i}h_1)), \\
&c_{l,\epsilon}^{+*}(h_1) = \frac{1}{\sqrt{2}} (\phi_{l,\epsilon}^{(c),+}(h_1)-\mathrm{i}\phi_{l,\epsilon}^{(c),+}(\mathrm{i}h_1)) , \qquad c_{l,\epsilon}^{+}(h_1) = \frac{1}{\sqrt{2}} (\phi_{l,\epsilon}^{(c),+}(h_1)+\mathrm{i}\phi_{l,\epsilon}^{(c),+}(\mathrm{i}h_1)).
\end{align*}
are bounded operators.
\item The following commutation relations hold (in the sense of quadratic forms) 
\begin{align*}
& [a^{+}(h_2), a^{+*}(g_2)]  =  \langle h_2 , g_2 \rangle \mathds{1} ,\\
& [a^{+}(h_2), a^{+}(g_2)]  =  [a^{+ *}(h_2), a^{+ *}(g_2)] = 0 ,\\
&\{b^{+}(h_1), b^{+*}(g_1)\}  =  \langle h_1 , g_1 \rangle \mathds{1}, \\
&\{b^{+}(h_1), b^{+}(g_1)\}  =  \{b^{+ *}(h_1), b^{+ *}(g_1)\} = 0 , \\
&\{c^{+}(h_1), c^{+*}(g_1)\}  =  \langle h_1 , g_1 \rangle \mathds{1} , \\
&\{c^{+}(h_1), c^{+}(g_1)\}  =  \{c^{+ *}(h_1), c^{+ *}(g_1)\} = 0 , \\
&\{b^{+}(h_1), c^{+}(g_1)\}  =  \{b^{+}(h_1), c^{+ *}(g_1)\} =  0 , \\
&\{b^{+ *}(h_1), c^{+}(g_1)\} = \{b^{+ *}(h_1), c^{+ *}(g_1)\} = 0 , \\
&[b^{+ \sharp}(h_1),a^{+  \sharp}(h_2)] = [c^{+  \sharp}(h_1),a^{+  \sharp}(h_2)] = 0 .
\end{align*}
\item \label{PullThr} We have that
\begin{align*}
e^{\mathrm{i} t H}a^{+ \sharp} (h_2) e^{ - \mathrm{i} t H} = a^{+ \sharp}(h_{2,-t}) , \quad e^{ \mathrm{i} t H} b^{+ \sharp} (h_1) e^{ - \mathrm{i} t H}  =   b^{+ \sharp}(h_{1,-t}) , \quad e^{ \mathrm{i} t H} c^{+ \sharp} (h_1) e^{ - \mathrm{i} t H} =  c^{+ \sharp}(h_{1,-t}),
\end{align*}
and the following ``pulltrough formulae'' are satisfied 
\begin{align*}
&a^{+*}(h_2) H  =  H a^{+*}(h_2)-a^{+*}(\omega^{(3)} h_2) , \quad a^{+}(h_2) H = H a^{+}(h_2)+a^{+}(\omega^{(3)} h_2) , \\
&b^{+*}(h_1) H = H b^{+*}(h_1)-b^{+*}(\omega^{(1)} h_1) , \quad b^{+}(h_1) H  =  H b^{+}(h_1)+b^{+}(\omega^{(1)} h_1) , \\
&c^{+*}(h_1) H = H c^{+*}(h_1)-c^{+*}(\omega^{(2)} h_1) , \quad c^{+}(h_1) H =  H c^{+}(h_1)+c^{+}(\omega^{(2)} h_1). 
\end{align*}
\end{enumerate}
\end{Th}
\section{Technical computations}\label{app:technical}

In this section, we prove the estimates \eqref{eq:new_propag2} and \eqref{eq:rewrite_ai} that were used in the proof of Theorem \ref{th:propag_massless_2}.
\subsection{Proof of \eqref{eq:new_propag2}}
We use the notations of the proof of Theorem \ref{th:propag_massless_2}. Using a commutator expansion at second order, proceeding as in \cite[Lemma 5.2]{BoFaSi12_01}, we compute
\begin{align*}
\mathbf{d}_0 \tilde F_i \Big ( \frac{ x_{t,\rho,i}^2 }{ c^2t^2 } \Big ) &= \Big \{ - \frac{2}{t} \frac{ x_{t,\rho,i}^2 }{ c^2t^2 } + \frac{ \rho }{ t } \Big (  \frac{ t^{-\rho} }{ |p_2| + t^{-\rho} } \frac{ x_{t,\rho,i}^2 }{ c^2t^2 } + \mathrm{h.c.} \Big ) + \frac{ 1 }{ t } \Big (  \frac{ |p_2| }{ |p_2| + t^{-\rho} } \frac{ p_2 }{ |p_2| } \cdot \frac{ x_{t,\rho,i} }{ ct } + \mathrm{h.c.} \Big ) \Big \} \tilde F_i' \Big ( \frac{ x_{t,\rho,i}^2 }{ c^2t^2 } \Big ) \\
&\quad + \mathcal{O}( t^{-2+\rho} ) \\
&= \tilde F_i' \Big ( \frac{ x_{t,\rho,i}^2 }{ c^2 t^2 } \Big )^{\frac12} \Big \{ - \frac{2}{t} \frac{ x_{t,\rho,i}^2 }{ c^2 t^2 } + \frac{ \rho }{ t } \Big (  \frac{ t^{-\rho} }{ |p_2| + t^{-\rho} } \frac{ x_{t,\rho,i}^2 }{ c^2 t^2 } + \mathrm{h.c.} \Big ) \\
& \qquad \qquad \qquad \quad + \frac{ 1 }{ t } \Big (  \frac{ |p_2| }{ |p_2| + t^{-\rho} } \frac{ p_2 }{ |p_2| } \cdot \frac{ x_{t,\rho,i} }{ ct } + \mathrm{h.c.} \Big ) \Big \} \tilde F_i' \Big ( \frac{ x_{t,\rho,i}^2 }{ c^2 t^2 } \Big )^{\frac12} + \mathcal{O}( t^{-2+\rho} ) .
\end{align*}
Using that $t^{-\rho} ( |p_2| + t^{-\rho} )^1 = 1 - |p_2| ( |p_2| + t^{-\rho} )^{-1}$ and commuting $|p_2|^{1/2} ( |p_2| + t^{-\rho} )^{-1/2}$ with $F_i' (  x_{t,\rho,i}^2 / c^2 t^2  )^{1/2}$ (using again \cite[Lemma 5.2]{BoFaSi12_01}), we obtain that
\begin{align*}
& = \tilde F_i' \Big ( \frac{ x_{t,\rho,i}^2 }{ c^2 t^2 } \Big )^{\frac12}  \Big \{ - \frac{2-2\rho}{t} \frac{ x_{t,\rho,i}^2 }{ c^2 t^2 }  \Big \} \tilde F_i' \Big ( \frac{ x_{t,\rho,i}^2 }{ c^2 t^2 } \Big )^{\frac12}  \\
& + \Big ( \frac{ |p_2| }{ |p_2| + t^{-\rho} } \Big )^{\frac12} \tilde F_i' \Big ( \frac{ x_{t,\rho,i}^2 }{ c^2 t^2 } \Big )^{\frac12} \Big \{ - \frac{ 2 \rho }{ t } \frac{ x_{t,\rho,i}^2 }{ c^2 t^2 } + \frac{ 1 }{ c t } \Big ( \frac{ p_2 }{ |p_2| } \cdot \frac{ x_{t,\rho,i} }{ c t } + \mathrm{h.c.} \Big ) \Big \} \tilde F_i' \Big ( \frac{ x_{t,\rho,i}^2 }{ c^2 t^2 } \Big )^{\frac12} \Big (  \frac{ |p_2| }{ |p_2| + t^{-\rho} } \Big )^{\frac12}  + \mathcal{O}( t^{-2+\rho} ) .
\end{align*}
From the properties of the support of $F'_i$, we then deduce that
\begin{align*}
 & \mathbf{d}_0 \tilde F_i \Big ( \frac{ x_{t,\rho,i}^2 }{ c^2 t^2 } \Big )  \le - 2 \tilde F_i' \Big ( \frac{ x_{t,\rho,i}^2 }{ c^2 t^2 } \Big )^{\frac12}  \Big \{ \frac{1-\rho}{t} \frac{ x_{t,\rho,i}^2 }{ c^2 t^2 }  \Big \} \tilde F_i' \Big ( \frac{ x_{t,\rho,i}^2 }{ c^2 t^2 } \Big )^{\frac12}  \\
& - 2 \Big ( \frac{ |p_2| }{ |p_2| + t^{-\rho} } \Big )^{\frac12} \tilde F_i' \Big ( \frac{ x_{t,\rho,i}^2 }{ c^2 t^2 } \Big )^{\frac12} \Big \{ \frac{ \rho - c^{-1} }{ t } \frac{ x_{t,\rho,i}^2 }{ c^2 t^2 } \Big \} \tilde F_i' \Big ( \frac{ x_{t,\rho,i}^2 }{ c^2 t^2 } \Big )^{\frac12} \Big (  \frac{ |p_2| }{ |p_2| + t^{-\rho} } \Big )^{\frac12}  + \mathcal{O}( t^{-2+\rho} ).
\end{align*}
Using again the properties of the support of $F'_i$ proves \eqref{eq:new_propag2}.

\subsection{Proof of \eqref{eq:rewrite_ai}}
We use again the notations of the proof of Theorem \ref{th:propag_massless_2}. We compute (with $x_i=\mathrm{i} \nabla_{p_2}$)
\begin{align*}
a_i &= \frac{1}{2} \Big ( \frac{p_2}{|p_2|} \cdot x_i + x_i \cdot \frac{p_2}{|p_2|} \Big ) \\
& = \sum_{\ell=1}^3 \Big ( x_i^{(\ell)}\frac{ p_{2}^{(\ell)} }{|p_2|} + \frac{1}{2|p_2|} - \frac{(p_{2}^{(\ell)})^2}{2|p_2|^3} \Big ) \\
&= \sum_{\ell=1}^3 \Big ( x_i^{(\ell) } \frac{ |p_2| }{ |p_2| + t^{-\rho} } \Big ( 1 + \frac{ t^{-\rho} }{ |p_2| } \Big ) \frac{ p_{2}^{(\ell)} }{|p_2|} + \frac{1}{2|p_2|} - \frac{ (p_{2}^{(\ell)} )^2}{2|p_2|^3} \Big )\\
&= \sum_{\ell=1}^3 \left ( \bigg(\Big(\frac{ |p_2| }{ |p_2| + t^{-\rho} }\Big)^{\frac{1}{2}}x_i^{(\ell) } + \bigg[x_i^{(\ell) },\Big(\frac{ |p_2| }{ |p_2| + t^{-\rho} }\Big)^{\frac{1}{2}} \bigg]\bigg)\Big(\frac{ |p_2| }{ |p_2| + t^{-\rho} }\Big)^{\frac{1}{2}} \Big ( 1 + \frac{ t^{-\rho} }{ |p_2| } \Big ) \frac{ p_{2}^{(\ell)} }{|p_2|} + \frac{1}{2|p_2|} - \frac{ (p_{2}^{(\ell)} )^2}{2|p_2|^3}\right ) \\
&= \sum_{\ell=1}^3 \Big ( \Big ( x_{t,\rho,j}^{(\ell)} +  \mathcal{O}( \frac{1}{ |p_2|+t^{-\rho} } ) \Big ) \Big ( 1 + \frac{ t^{-\rho} }{ |p_2| } \Big ) \frac{p_{2}^{(\ell)} }{|p_2|} + \frac{1}{2|p_2|} - \frac{ ( p_{2}^{(\ell)} )^2}{2|p_2|^3} \Big ) \\
&= \sum_{\ell=1}^3 x_{t,\rho,j}^{(\ell)}\big ( 1 + \mathcal{O}(t^{-\rho} ) |p_2|^{-1} \big ) +  \mathcal{O}( 1 ) |p_2|^{-1}.
\end{align*}
This proves \eqref{eq:rewrite_ai}.

\bibliographystyle{amsalpha}

\end{document}